\documentclass[11pt]{article}
\usepackage[margin=2.1cm]{geometry}

\usepackage{amsfonts,amssymb,amsmath,amsthm,dsfont,bm}
\usepackage[footnotesize]{caption}
\usepackage{graphicx}
\usepackage[linesnumbered,ruled,lined,boxed,commentsnumbered]{algorithm2e} 
\usepackage{mathtools} 
\usepackage{subfig}
\usepackage{hyperref}
\usepackage{url} 
\usepackage[dvipsnames]{xcolor}
\usepackage{cite}
\usepackage{pgfplots}
\usepackage{filecontents}
\usepackage{csvsimple}

\usepackage{cancel}

\title{Breaking the waves: asymmetric random periodic features\\
  for low-bitrate kernel machines}

\author{Vincent Schellekens$^*$ and Laurent Jacques\thanks{E-mail: {\em \{vincent.schellekens,~laurent.jacques\}@uclouvain.be}. ISPGroup, ELEN/ICTEAM, UCLouvain (UCL), B1348 Louvain-la-Neuve, Belgium. VS and LJ are funded by Belgian National Science Foundation (F.R.S.-FNRS).}}

\newcommand{\scp}[2]{\langle #1, #2 \rangle}

\newcommand{\Zbb}{\mathbb{Z}}
\newtheorem{theorem}{Theorem}
\newtheorem{definition}{Definition}
\newtheorem{proposition}{Proposition}
\newtheorem{corollary}{Corollary}
\newtheorem{lemma}{Lemma}
\newtheorem*{remark}{Remark}
\newtheorem{example}{Example}

\newcommand{\supp}{{\rm supp}\,}

\newcommand{\sign}{{\rm sign}\,}

\newcommand{\ud}{\mathrm{d}}

\renewcommand{\leq}{\leqslant}
\renewcommand{\geq}{\geqslant}

\newcommand{\iid}{%
  \ifmmode
  \mathrm{i.i.d.}%
  \else%
  i.i.d.\@\xspace%
  \fi%
}

\newcommand{\Id}{\bs I}

\newcommand{\bb}{\mathbb}

\newcommand{\ts}{\textstyle}

\newcommand{\bs}{\boldsymbol}
\newcommand{\cl}{\mathcal}
\newcommand{\ie}{\emph{i.e.}, }
\newcommand{\pd}{{p.d.}\xspace}
\newcommand{\eg}{\emph{e.g.}, }

\newcommand{\wlogg}{\emph{w.l.o.g.}\xspace}

\renewcommand{\Vec}[1]{\bs{#1}} 
\newcommand{\distiid}{\sim_{\iid}} 
\newcommand{\expec}[1]{\mathop{{}\mathbb{E}}_{#1}} 

\newcommand{\im}{\mathrm{i}\mkern1mu} 
\DeclarePairedDelimiterX{\norm}[1]{\lVert}{\rVert}{#1} 
\newcommand{\Integr}[4]{\int_{#1}^{#2}#3\,\ud #4} 

\newcommand{\wt}{\widetilde}
\newcommand{\wh}{\widehat}

\newcommand{\pf}{\mathrm{PF}}
\newcommand{\sik}{\kappa^{\scriptscriptstyle \Delta}} 
\newcommand{\dsik}{\dot{\kappa}^{\scriptscriptstyle \Delta}} 

\makeatletter
\newcommand{\RemoveAlgoNumber}{\renewcommand{\fnum@algocf}{\AlCapSty{\AlCapFnt\algorithmcfname}}}
\newcommand{\RevertAlgoNumber}{\algocf@resetfnum}
\makeatother

\begin{document}

\maketitle

\begin{abstract}
  Many signal processing and machine learning applications are built from evaluating a kernel on pairs of signals, \eg to assess the similarity of an incoming query to a database of known signals. This nonlinear evaluation can be simplified to a linear inner product of the random Fourier features of those signals: random projections followed by a periodic map, the complex exponential. It is known that a simple quantization of those features (corresponding to replacing the complex exponential by a different periodic map that takes binary values, which is appealing for their transmission and storage), distorts the approximated kernel, which may be undesirable in practice. Our take-home message is that when the features of only one of the two signals are quantized, the original kernel is recovered without distortion; its practical interest appears in several cases where the kernel evaluations are asymmetric by nature, such as a client-server scheme.
  
  Concretely, we introduce the general framework of \emph{asymmetric random periodic features}, where the two signals of interest are observed through random periodic features---random projections followed by a general periodic map, which is allowed to be different for both signals. We derive the influence of those periodic maps on the approximated kernel, and prove uniform probabilistic error bounds holding for all pair of signals picked in an infinite low-complexity set. Interestingly, our results allow the periodic maps to be discontinuous, thanks to a new mathematical tool, \ie the mean Lipschitz smoothness. We then apply this generic framework to semi-quantized kernel machines (where only one of the signals has quantized features and the other has classical random Fourier features), for which we show theoretically that the approximated kernel remains unchanged (with the associated error bound), and confirm the power of the approach with numerical simulations. 
\end{abstract}

\section{Introduction}
\label{sec:introduction}

Rather than to directly process high-dimensional signals, it is often more efficient to first summarize them to their main \emph{features}. This assumes that these capture essential information for the considered processing, such as the proximity of any pair of signals. Mathematically, the signal summarization is modeled by a feature map $\bs \varphi$ from the signal space $\Sigma$ to the feature (or embedding) space $\cl E$. This map $\Vec{\varphi}$ transforms the representation of signals while encoding some aspects of their geometry; loosely speaking, this can be written as $\cl D_{\cl E}(\Vec{\varphi}(\Vec{x}),\Vec{\varphi}(\Vec{y}))\approx \cl D_{\Sigma}(\Vec{x},\Vec{y})$ for any pair of signals $\Vec{x},\Vec{y} \in \Sigma$, where $\cl D_{\Sigma}$ is the preserved geometric quantity (such as an inner product or a distance), and $\cl D_{\cl E}$ is an evaluation procedure acting only on the signal features. This approach is useful whenever the features $\Vec{\varphi}(\Vec{x})$ are easier to process with respect to some critical computational resource (\eg memory usage, computing time)---often at the price of an approximation error (as suggested by the approximation symbol $\cl D_{\cl E} \approx \cl D_{\Sigma}$ above). The map $\Vec{\varphi}$ is most often than not a randomized function (drawn from a distribution). There are essentially three main ways to save computational resources with features: \textit{(i)}~leveraging \emph{dimensionality reduction} (\ie $\Vec{\varphi}(\Vec{x})$ is encoded by much less coefficients than the dimension of $\Vec{x}$), such as in compressive sensing techniques~\cite{foucart2017mathematical}, where typically the ``embedding'' $\Vec{\varphi}$ is linear and $\cl D_{\Sigma}$ and $\cl D_{\cl E}$ are the Euclidean distances; \textit{(ii)}~using mappings that \emph{linearize} the evaluation of an otherwise nonlinear quantity, such as random Fourier features (RFF)~\cite{Rahimi2008RFF}, where $\cl D_{\Sigma}$ is a kernel $\kappa$ and $\cl D_{\cl E}$ is simply the inner product; and \textit{(iii)}~\emph{quantizing features}, where $\Vec{\varphi}(\Vec{x})$ produces a quantized output that can be encoded with a highly reduced bitrate compared to the initial signal, allowing reducing the memory and/or transmission load, as well as paving the way for hardware-based procedures. This quantization objective is often combined with either \textit{(i)}, as in quantized compressive sensing~\cite{boufounos2015quantization,gunturk2013sobolev,dirksen2019quantized}, or \textit{(ii)}, as in one-bit universal embeddings~\cite{boufounos2015universal}. As made clear below, our contributions also target the combination of \textit{(ii)} with \textit{(iii)}.

In all those applications, it is almost always assumed that the available features for the two signals $\Vec{x},\Vec{y} \in \Sigma$ come from the \emph{same} feature map $\Vec{\varphi}$ (we say the features are \emph{symmetric}). However, we can legitimately wonder if removing this assumption, \ie accessing the signals through features $\Vec{\varphi}(\Vec{x})$ and $\Vec{\psi}(\Vec{y})$, where we have the freedom to set $\Vec{\varphi} \neq \Vec{\psi}$, can further reduce specific computational resources. The practical interest of this relaxation---that we call \emph{asymmetric features}---arises when the two signals come from different sources (\ie when the setting is intrinsically asymmetric): for example, those sources might have different resources (such as memory or power) at their disposal, and will therefore benefit differently from techniques \textit{(i)-(iii)}.

In this work, we are interested in asymmetric features for linearizing kernel estimations, as explained in \textit{(ii)}. Anticipating over the detailed description of Sec.~\ref{sec:background}, we work with \emph{random periodic features}, $\Vec{\varphi}(\Vec{x}) := f(\bs\Omega^T\Vec{x} + \Vec{\xi})$ and $\Vec{\psi}(\Vec{y})  :=  g(\bs\Omega^T\Vec{y} + \Vec{\xi})$ with $\bs\Omega$ a random projection matrix, $\Vec{\xi}$ a random dither, and $f, g$ two periodic functions. We thus generalize the context of random Fourier features \cite{Rahimi2008RFF}, where $f(\cdot) = g(\cdot) = \exp(\im \cdot)$ (the complex exponential), by ``\emph{breaking the waves}'' with possibly discontinuous, distinct functions $f$ and $g$ (as described in Sec.~\ref{sec:nonasymptotic}). We show how those features can be used to approximate shift-invariant kernels $\kappa$, \ie $\langle \Vec{\varphi}(\Vec{x}),\Vec{\psi}(\Vec{y}) \rangle \approx \kappa(\Vec{x},\Vec{y})$, in expectation over the random quantities $\bs\Omega$, $\Vec{\xi}$. Our motivating use-case is to combine this approach with harsh quantization of some features, objective \textit{(iii)}, as we explain in the next paragraph. However, all our developments are generic, and of interest for any machine learning algorithm that processes or takes decisions from the local geometry of data.

\paragraph{Semi-binary kernel machines as a motivating application:}

If one of the periodic functions incorporates the quantization of the feature vector (say, the one-bit universal quantization, or square wave function $q: \bb R \rightarrow \{0,1\}$~\cite{Boufounos2013efficientCodingQuantized}; see Fig.~\ref{fig:univquant}), then that feature vector can be stored or transmitted (or both) much more efficiently than the usual (infinite-precision) random Fourier features. Consider for example a machine learning context, where a kernel method~\cite{scholkopf2002learning} such as a Support Vector Machine (SVM) \cite{boser1992training} has been trained in advance on some dataset $X = \{\Vec{x}_i\}_{i=1}^n \subset \Sigma$. To actually use this model for prediction on a new signal $\Vec{x}' \in \Sigma$, the physical device that records it must either  communicate with a server where the inference is performed remotely (Fig.~\ref{fig:intro}a), or implement this model directly (Fig.~\ref{fig:intro}b); in either case, this is an expensive operation whenever this device is under tight computational resources constraints, and quantization of feature vectors is potentially very helpful.

\begin{figure}
  \centering
  \includegraphics[width=0.75\linewidth]{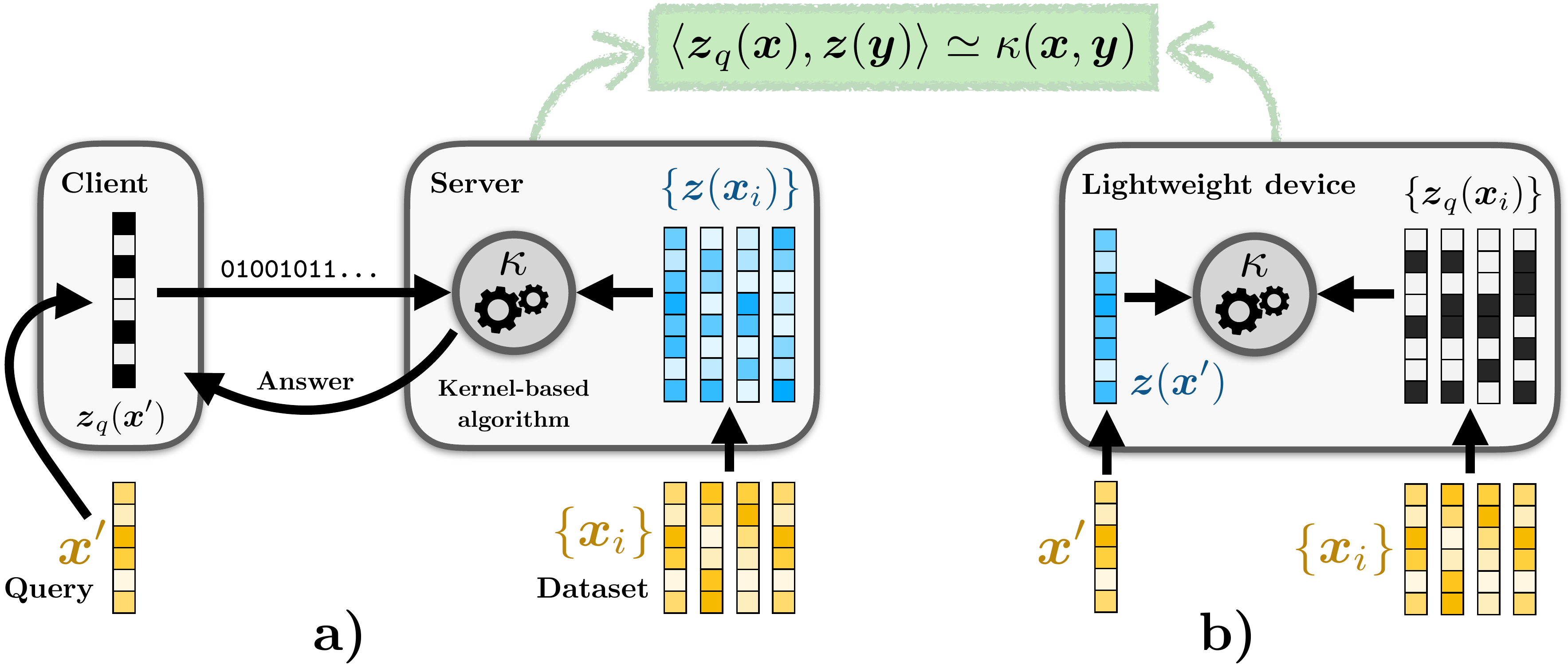}
  \caption{Two motivating applications of our results (from Sec.~\ref{sec:quantized}, in green): combining one-bit universal features with usual random Fourier features yields the same desired kernel $\kappa$. (a) A ``client'' device records a ``query'' signal $\Vec{x}'$, and transmits its quantized features $\Vec{z}_q(\Vec{x}')$, encoded efficiently as only $m$ bits, to a ``server'' device that can evaluate the kernel similarity with the rest of a dataset from their usual $n$ full-precision RFF $\{\Vec{z}(\Vec{x}_i)\}_{i=1}^n$. (b) A lightweight device implements a kernel method with very low memory requirements, only having to store $\{\Vec{z}_q(\Vec{x}_i) \}_{i=1}^n$ the $n$ one-bit feature vectors of the dataset instead of the full-precision ones, provided the usual RFF $\Vec{z}(\Vec{x}')$ are used for the incoming query vectors.}
  \label{fig:intro}
\end{figure}

In the first scenario (inference done remotely on a server), we might quantize the feature vector of the query signal, $\Vec{\varphi}(\Vec{x}') \in \{0,1\}^m$ (but not the feature vectors of the dataset $\Vec{\psi}(\Vec{x}_i)$). This allows to heavily reduce the bitrate when communicating this vector to the server, and even paves way for computing those features directly in hardware, \eg relying on voltage-controlled oscillators~\cite{yoon2008time}. In the second context, we could conversely binarize the feature vectors of the dataset so that $\Vec{\psi}(\Vec{x}_i) \in \{0,1\}^m$ for all $\Vec{x}_i \in X$, but not the incoming query vector $\Vec{\varphi}(\Vec{x}')$ (it is even possible to encode only a subset of the dataset features for models that only need to access some entries, such as SVM with the support vectors). The advantage here is that the memory needed to store the model is heavily reduced, with additional computational benefits coming from the embedded processing of binary values. This idea has received significant attention in the literature, \eg following~\cite{jegou2010product} for nearest-neighbor search.

For both of those examples, the main question that we seek to answer is to quantify the loss of accuracy (induced by quantization) as a function of the feature vector length $m$. More precisely, our goal is to obtain (probabilistic) guarantees on the decay of the kernel approximation error as a function of $m$, that hold uniformly for any pairwise comparisons of signals taken an infinite (but compact) set $\Sigma$. In this case, the main challenge lies in dealing with the discontinuous nature of the quantization operation---handling discontinuities is thus one of the key features of this paper.

\subsection{Related work}
\label{sec:related-work}

Before detailing the elements of our approach, we find useful to mention a few related works, showing how they inspired us, and stressing their connections and differences with our contributions. 

\paragraph{Quantization of (symmetric) random Fourier features:} The construction of the general random periodic features considered in this work is instantiated in Sec.~\ref{sec:quantized} to the case where the corresponding periodic map is the universal quantizer (or square wave function). This approach was introduced in~\cite{raginsky2009locality,Boufounos2013efficientCodingQuantized} as a binary map preserving local distances (\ie up to a given radius), the universal quantized embedding. Those features have subsequently been used for kernel methods in~\cite{boufounos2015universal}, which is similar to the framework we propose but where not one but both signal features being compared are quantized in a symmetric fashion, which distorts the kernel to be recovered. This line of work was further generalized is~\cite{boufounos2017representation}, where uniform guarantees are derived for generic periodic function (possibly discontinuous) instead of the one-bit universal quantizer, holding on infinite signal sets. This defines the random periodic features approach (see Sec.~\ref{sec:background} for details) that we also consider; we provide an in-depth description of how our results relate to (and complement) those from~\cite{boufounos2017representation} in Sec.~\ref{sec:correctPTB}.

Back to the particular problem of quantizing random Fourier features, another line of work~\cite{zhang2018lowPrecisionRFF} shows that a specific stochastic quantization hardly harms the generalization performances of RFF-based algorithms. The ultimate objective of this last work is, however, fairly different from ours: the authors seek to reduce the memory requirements \emph{during training} by performing a more sophisticated quantization, and then use the full-precision RFF for the subsequent inference stage; on the other hand, our objective is to provide a simple quantization scheme to reduce the resources \emph{during the inference stage}, without concerns for how the training was performed.

\paragraph{Asymmetric features and quantizations:}
The possibility to use asymmetric features has been explored for linear embeddings in~\cite{ryder2019asymmetric}, as an additional degree of freedom to minimize (in a data-dependent fashion) the average error of the distance estimation. In~\cite{dong2008asymmetric}, weighted universal embeddings are used for distance estimation, where the weights depend upon one of the two signals (which makes the scheme asymmetric) to decrease the error on this estimation. Closer to our context, in~\cite{gordo2013asymmetric}, it is experimentally shown for a broad set of feature maps (such as Locality Sensitive Hashing, universal embeddings, and several variants of PCA) that quantizing the features of the dataset but not of the query (as in scenario Fig.~\ref{fig:intro}b) significantly improves the performances compared to quantizing both features. Similarly, the authors of~\cite{li2019random,li2019sign} recently considered linear random projections (with the same matrix) of two signals that have been quantized with different quantization levels.

\paragraph{Compressive learning:} In a nutshell, compressive learning \cite{gribonval2017compressiveStatisticalLearning,keriven2016compressive} aims at estimating the parameters of a distribution $\cl P$, or the parameters of a parametric distribution approximating it, from the sketching of an entire dataset of $n$ signals $X = \{\bs x_i\}_{i=1}^n$ generated by $\cl P$, \ie such that $\bs x_i \sim_{\iid} \cl P$. Given a random projection matrix $\bs \Omega$, the sketch $\bs s_X := \frac{1}{n} \sum_{i=1}^n \Vec{z}(\bs x_i)$ is computed from the \emph{pooling} (averaging) of the random Fourier features $\Vec{z}: \bs u \to \exp(\im \bs \Omega^\top \bs u)$ (the exponential being computed componentwise onto vectors) of each dataset signal. For large value of $n$, this sketch estimates the characteristic function $\cl A(\cl P; \bs \Omega) := \bb E_{\bs x \sim \cl P} \exp(\im \Omega^\top \bs x)$ of $\cl P$ over the ``frequencies'' supported by the rows of $\bs\Omega$.  Therefore, under appropriate conditions, one can formulate an inverse problem aiming to estimate the parameters of $\cl P$ by matching the characteristic function of a probing distribution $\widehat{\cl P}$ (estimated over $\bs\Omega$) from $\bs s_X \approx \cl A(\cl P; \bs \Omega)$.  We considered in \cite{schellekens2018quantized} the possibility to replace the random Fourier features used to build $\bs s_X$ with a general (dithered) periodic function $f$, such as the universal quantizer, thus computing $\bs s_X' = \frac{1}{n} \sum_{i=1}^n \Vec{z}_f(\bs x_i) = \frac{1}{n} \sum_{i=1}^n f(\bs \Omega^\top \bs x_i + \bs \xi)$ with a random dither $\bs \xi$. While the dataset sketch is altered (\eg quantized with the universal quantizer), we showed that the estimation of the distribution parameters from the observed sketch is still accurate if we use the RFF (for the probing distribution), as if the sketch was not quantized, hence leading to an asymmetric scheme between the dataset sketching and the estimation procedure.

\subsection{Paper organization}
\label{sec:paper-organization}

We provide in Sec.~\ref{sec:background} several preliminary elements as well as important concepts of the relevant literature: random Fourier features and their (possibly quantized) extension to any periodic nonlinearity. We then start by analyzing how the kernel approached by asymmetric random periodic features behaves in expectation (in the asymptotic case), which is proved in Sec.~\ref{sec:asymptotic}. Our main results come in Sec.~\ref{sec:nonasymptotic}, where we prove uniform error bounds of the kernel approximation for infinite signal sets. In order to do so, we introduce a new tool, the mean Lipschitz smoothness property. Sec.~\ref{sec:correctPTB} relates our approach to the context of geometry-preserving embedding (or coding) developed in \cite{boufounos2017representation}, solving in the same time an error in the proof of one of their results. Next, we apply our general results of Sec.~\ref{sec:nonasymptotic} to the semi-quantized setting motivated above in Sec.~\ref{sec:quantized}, and illustrate with numerical experiments in Sec.~\ref{sec:experiments}, before concluding in Sec.~\ref{sec:conclusion}.

\subsection{Notations}
\label{sec:notations}

Vectors and matrices are denoted by bold symbols. The unit imaginary number is noted $\im = \sqrt{-1}$. The real part, the imaginary part, and the complex conjugation of $a \in \bb C$ read $\Re(a)$, $\Im(a)$, and $a^*$, respectively. We will often consider bounded $2\pi$-periodic functions $f,g :\bb R \to \bb C$ for which the 2-norm and the infinity norm read $\|f\|^2 = \frac{1}{2\pi} \int_0^{2\pi} |f(t)|^2 \ud t$ and  $\|f\|_{\infty} := \sup_{t \in [0,2\pi]} |f(t)|$, respectively, and the inner product of $f$ and $g$ is $\scp{f}{g} = \frac{1}{2\pi} \int_0^{2\pi} f(t) g ^*(t) \ud t$.
For brevity and clarity, we will sometimes refer to a function using the ``dot'' notation, \eg $\exp(\im \cdot)$ for the function $t \mapsto \exp(\im t) \in \bb C$ for $t \in \bb R$.

The $\ell_p$-norm of a vector $\bs u \in \bb R^d$ reads $\|\bs u\|_p = (\sum_i |u_i|^p)^{1/p}$ for $p \geq 1$, with $\|\bs u\|_\infty = \max_i |u_i|$, and $\|\bs u\|_0 = |\supp \bs u| = |\{i: u_i \neq 0\}|$. The unit $\ell_p-$ball ($p\geq 1$) in dimension $d$ is noted $\bb B^d_p := \{\Vec{u} \in \bb R^d \: | \: \|\Vec{u}\|_p \leq 1\}$, with the shorthand $\bb B^d = \bb B^d_2$. The cardinality of a finite set $\cl S$ is $|\cl S|$, the Minkowski sum of two sets $\cl A$ and $\cl B$ is $\cl A + \cl B = \{a + b: a \in \cl A, b \in \cl B\}$, the index set in $\bb R^d$ is $[d] := \{1,\,\cdots,d\}$ for $d \in \bb N$, the identity matrix in $\bb R^d$ is $\Id_d \in \bb R^{d \times d}$, and the Kronecker delta $\delta_{k,k'}$ is defined as $\delta_{k,k'} = 1$ if $k = k'$ and $\delta_{k,k'} = 0$ otherwise. By abuse of notation, evaluating a scalar function $f : \bb R \rightarrow \bb C$ on a vector $\Vec{u} \in \bb R^m$ means applying this function componentwise, \ie $f(\Vec{u}) \in \bb C^m$ with $(f(\Vec{u}))_j = f(u_j)$.

We use the convention where the ($d$-dimensional) Fourier transform of a function $f: \bb R^d \to \bb C$ reads $\hat f(\bs \omega) = (\cl F f)(\bs \omega) := \frac{1}{(2 \pi)^d} \int_{\bb R^d} e^{- \im \bs u^{\top} \bs \omega} f(\bs u) \ud \bs u$, with inverse  $(\cl F^{-1} \hat f)(\bs u) := \int_{\bb R^d} e^{\im \bs u^{\top} \bs \omega} \hat{f}(\bs \omega) \ud \bs \omega$. The same convention is used for the Fourier transform of finite measures on $\bb R^d$. The notation $\sim \cl P$ denotes that a random variable, vector or function is distributed according to the distribution $\cl P$. The uniform distribution on a set $\cl A$ is noted $\cl U(\cl A)$, and ``\iid'' means ``identically and independently'' distributed. In all our developments, except if specified differently, $C,C',\ldots, c,c',\ldots > 0$ denote \emph{universal} constants whose value may change from one instance to the other.

\section{Preliminaries}
\label{sec:background}

We introduce here several fundamental concepts supporting our approach. We first precise the kind of signal space we consider, as well as how signals are compared through a kernel, before to briefly explain the principles sustaining the definition of random Fourier features (RFF). Next, we generalize RFF to any random periodic features for a family of bounded periodic functions.

\subsection{Signals and kernels}
\label{sec:signals-kernel}

In this work, we focus on signals belonging to a bounded signal space $\Sigma \subset \bb R^d$ having finite Kolmogorov $\eta$-entropy $\cl H_\eta(\Sigma)$ for any radius $\eta > 0$ \cite{kolmogorov1961}. This entropy,  defined as $\cl H_\eta(\Sigma) := \log \cl C_{\eta}(\Sigma)$, is related to the covering number $\cl C_{\eta}(\Sigma)$ of $\Sigma$, the cardinality of the smallest finite subset of $\Sigma$ that \emph{covers} it with balls of radius $\eta$. Using the Minkowski sum, this means that
$$
\ts \cl C_{\eta}(\Sigma) := \min \{ |\cl S| : \cl S \subset \Sigma \subset \cl S + \eta \bb B_2^d \},
$$
which is finite for any compact set $\Sigma$.

The Kolmogorov entropy measures the intrinsic dimension of $\Sigma$ in $\bb R^d$. In particular, $\cl H_\eta(\cl V \cap \bb B^d_2) \leq C d'\log(1+1/\eta)$ for any subspace $\cl V \subset \bb R^d$ of dimension $d'\leq d$~\cite{pisier1999volume}, and the set of $s$-sparse vectors $\Sigma_s := \{\Vec{x} \in \bb R^d, \|\Vec{x}\|_0 \leq s \}$ restricted to the unit ball has entropy bounded by $\cl H_{\eta}(\Sigma_s \cap \bb B^d_2) \leq C \cdot s \log(d/s) \log(1+1/\eta)$, see for example~\cite{baraniuk2008simple}. Other bounds exist for, \eg the set of bounded group sparse signals \cite{ayaz2016uniform}, bounded low-rank matrices \cite{candes2011tight}, or for specific low-dimensional manifolds \cite{eftekhari2015new}. 

At the heart of our study is the comparison of two signals through a \emph{kernel}, \ie a function over pairs of signals $\kappa : \Sigma \times \Sigma \rightarrow \bb C$ (in the machine learning literature, kernels are usually real-valued). Typically, invoking the so-called ``kernel trick'' \cite{boser1992training}, $\kappa$ represents the inner product between the input signals $\Vec{x},\Vec{x}'$ when they are mapped in some \textit{implicit} feature space by an appropriate map $\phi : \Sigma \rightarrow \bb H$: $\kappa(\Vec{x},\Vec{x}') = \langle \phi(\Vec{x}),\phi(\Vec{x}') \rangle_{\bb H}$ for some Hilbert space $\bb H$. By definition of the inner product, the kernel $\kappa$ must then necessarily be conjugate symmetric ($\kappa(\Vec{x},\Vec{x}') = \kappa^*(\Vec{x}',\Vec{x})$), and positive definite (\pd), \ie for any number $n$, $\sum_{i,j = 1}^n \, c_ic^*_j\kappa(\Vec{x}_i,\Vec{x}_j) \geq 0$ for all $\Vec{x}_1, \dots, \Vec{x}_n \in \Sigma$ and $c_1, \dots, c_n \in \bb C$.

\subsection{Random Fourier features}
\label{sec:rand-four-feat}

Random Fourier features (RFF) are implicitly built on \emph{Bochner's theorem}~\cite{Rudin1962bochnerBook}. This theorem states that a shift-invariant continuous kernel $\kappa(\Vec{x},\Vec{y}) = \sik(\Vec{x}-\Vec{y})$ (for some $\sik: \Sigma - \Sigma \to \bb C$) is positive definite if and only if it is the (inverse) Fourier transform of a nonnegative finite measure $\Lambda$, \ie
\begin{equation}
  \label{eq:bochner}
  \kappa \text{ positive definite} \quad \Leftrightarrow \quad \ts \sik(\Vec{u}) = (\cl F^{-1} \Lambda)(\Vec{u}) = \int_{\bb R^d} e^{\im \Vec{\omega}^\top\Vec{u}} \mathrm{d}\Lambda(\Vec{\omega}).
\end{equation}
In particular, assuming \wlogg the normalization $\kappa(\Vec{x},\Vec{x}) = \sik(\Vec{0}) = 1$, $\Lambda$ is a probability distribution over $\bb R^d$, and the kernel can be written $\sik(\Vec{u}) = \expec{\Vec{\omega} \sim \Lambda} e^{\im \Vec{\omega}^\top\Vec{u}}$. The key idea of random Fourier features~\cite{Rahimi2008RFF} is thus to construct low-dimensional features $\Vec{z}(\Vec{x}), \Vec{z}(\Vec{y})$ whose inner product approximates the kernel $\kappa(\Vec{x},\Vec{y})$ by Monte Carlo sampling of this expectation.

\begin{definition}[Random Fourier features]
  \label{def:RFF}
  Let $\kappa(\Vec{x},\Vec{y}) = \sik(\Vec{x}-\Vec{y})$ be a shift-invariant \pd kernel, normalized such that $\sik(\Vec{0}) = 1$, with Fourier transform $\Lambda = \cl F \sik$. Given a target dimension $m$, the associated ``complex'' random Fourier features are
  \begin{equation}
    \label{eq:RFF}
    \ts \Vec{z}(\Vec{x}) := \frac{1}{\sqrt{m}}\exp\left(\im ( \bs\Omega^\top \Vec{x} + \Vec{\xi})\right) \quad \in \bb C^m, 
  \end{equation}
  with random projections (or ``frequencies'') $\bs\Omega := (\Vec{\omega}_1,\,\cdots, \Vec{\omega}_m) \in \bb R^{d \times m}$ generated as $\bs\Omega \sim \Lambda^m$, \ie with $\Vec{\omega}_j \distiid \Lambda$ for $j \in [m]$, and a random dither $\Vec{\xi} \in \bb R^m$ generated as $\bs \xi \sim \cl U^m([0,2\pi))$, \ie with $\xi_j \distiid \cl U([0,2\pi))$ for $j \in [m]$. We also define the ``real'' random Fourier features $\Vec{z}_{\cos}(\Vec{x})$ as $\Re[\Vec{z}(\Vec{x})]$, the real part of those features:
  \begin{equation}
  \label{eq:realRFF}
  \ts \Vec{z}_{\cos}(\Vec{x}) := \frac{1}{\sqrt{m}}\cos\left( \bs\Omega^\top \Vec{x} + \Vec{\xi}\right) \quad \in \bb R^m.
  \end{equation}
\end{definition}

\begin{remark}
  The dither $\Vec{\xi}$ was initially introduced in~\cite{Rahimi2008RFF} when only the real RFF $\Vec{z}_{\cos}(\Vec{x})$ were used; in the (more widely used) complex case $\Vec{\xi}$ is not necessary (see~\cite{sutherland2015error} for an in-depth comparison of the ``real'' versus ``complex'' random Fourier features). We still included it in this definition for the sake of consistency with Def.~\ref{def:generalPeriodicFeatures} below.
\end{remark}

By direct application of Bochner's theorem, the inner product of RFF indeed approaches (in expectation over the draw of the frequencies $\bs\Omega$) the target kernel: $\expec{}\langle \Vec{z}(\Vec{x}), \Vec{z}(\Vec{y}) \rangle = \kappa(\Vec{x},\Vec{y})$. Moreover, for a finite feature dimension $m$, the error of the kernel approximation $\wh{\kappa}(\Vec{x},\Vec{y}) := \langle \Vec{z}(\Vec{x}), \Vec{z}(\Vec{y}) \rangle$ can be uniformly bounded (\ie bound the absolute error $|\wh{\kappa}(\Vec{x},\Vec{y})-\kappa(\Vec{x},\Vec{y})|$ for all values $\Vec{x},\Vec{y}$ in $\Sigma$), with high probability on the draw of $\bs\Omega$ (we work with different normalization choices, so the result we present here differs slightly from the initial bound~\cite[Claim 1]{Rahimi2008RFF}). Finer bounds can be found, among others, in~\cite{sutherland2015error,sriperumbudur2015optimal}.

\begin{proposition}[Uniform kernel approximation error for RFF]
  \label{prop:RFF}
  Let $\Sigma$ be a compact set, and $\Vec{z}(\Vec{x})$ be the RFF defined above. Assume that there exists an associated constant $C_{\Lambda}$, such that 
  \begin{equation}
    \label{eq:CLambda}
    \ts \expec{\Vec{\omega} \sim \Lambda} |\Vec{\omega}^\top\Vec{a}| \leq C_{\Lambda} \|\Vec{a}\|_2,\quad \forall \Vec{a} \in \bb R^d.
  \end{equation} 
  Provided that, for $\epsilon > 0$,
  \begin{equation*}
    \ts m \geq  C \epsilon^{-2} \cl H_{c \epsilon/C_{\Lambda}}(\Sigma),
  \end{equation*}
  the kernel approximation $\wh{\kappa}(\Vec{x},\Vec{y}) = \langle \Vec{z}(\Vec{x}), \Vec{z}(\Vec{y}) \rangle$ has error uniformly bounded by
  \begin{equation*}
    \ts \big| \wh{\kappa}(\Vec{x},\Vec{y}) - \kappa(\Vec{x},\Vec{y}) \big| \leq \epsilon, \quad \forall \, \Vec{x},\Vec{y} \in \Sigma,
  \end{equation*}
  with probability exceeding $1 - C' e^{ - c' m \epsilon^2 }$.
\end{proposition}
\begin{proof}
  This version of the RFF approximation error is obtained as a particular case of our Prop.~\ref{prop:kernelApproxUniform}; see~\cite{Rahimi2008RFF} for the initial result.
\end{proof}

The constant $C_{\Lambda}$ defined in~\eqref{eq:CLambda} characterizes the smoothness of the kernel (if the kernel is smoother, it exhibits less high-frequency content, and $C_{\Lambda}$ will be lower). In most of the RFF literature, this constant is bounded by the Cauchy-Schwarz inequality as $C_{\Lambda} = \expec{\Vec{\omega} \sim \Lambda} \|\Vec{\omega}\|_2$. Then, one can (as done in~\cite{Rahimi2008RFF}) further bound $\expec{\Vec{\omega} \sim \Lambda} \|\Vec{\omega}\|_2 \leq \sigma_{\Lambda}$ where $\sigma_{\Lambda}^2$ is the second moment of $\Lambda$, equivalent to the kernel curvature at the origin, \ie $C_{\Lambda}^2 \leq \sigma_{\Lambda}^2 := \expec{\Vec{\omega} \sim \Lambda}\|\Vec{\omega}\|_2^2 = \nabla^2 \sik \rvert_{\Vec{u} = \Vec{0}}$, with $\nabla^2$ the Laplacian operator. However, for specific distributions, using the Cauchy-Schwarz inequality in high dimension leads to a loose bound of $C_\Lambda$. For example, if the covariance matrix of $\Vec{\omega}$ is upper bounded by $\wt{\sigma}^2_{\Lambda} \Id_d$ for some $\wt{\sigma}^2_{\Lambda} > 0$ (if, \eg  $\cl F^{-1} \Lambda$ is the Gaussian RBF (``radial basis function'') kernel with radius $1/\wt{\sigma}^2_\Lambda$, $\Lambda$ is isotropic \cite{Vershynin12}, or if each component of $\bs \omega$ are \iid with variance bounded by $\wt{\sigma}^2_{\Lambda}$) then
$$
\ts \big(\expec{\Vec{\omega} \sim \Lambda} |\Vec{\omega}^\top\Vec{a}| \big)^2 \leq \expec{\Vec{\omega} \sim \Lambda} |\Vec{\omega}^\top\Vec{a}|^2 = \Vec{a}^\top \big( \expec{\Vec{\omega} \sim \Lambda} \Vec{\omega} \Vec{\omega}^\top \big) \Vec{a} \leq \Vec{a}^\top \cdot \wt{\sigma}^2_\Lambda \Id_d \cdot \Vec{a} = \wt{\sigma}^2_\Lambda \|\Vec{a}\|_2^2.
$$
In this case we obtain $C_{\Lambda} = \wt{\sigma}_\Lambda$, while Cauchy-Schwarz gives $\sqrt{d} \cdot \wt{\sigma}_\Lambda$, hence overestimating the constant by a factor $\sqrt{d}$.

\begin{example}
    \label{ex:UoS-example}
    Consider the simple case where the signals of interest have an $\ell_2$-norm smaller than 1 and lie in a union of $S$ subspaces of $\bb R^d$, each with dimension $s$. This signal space model encompasses, for instance, $\Sigma = \bb B_2^d$ (that is, where $S=1$ and $s=d$), the set of bounded $s$-sparse signals in $\bb R^d$ for which $S = {d \choose s} \leq (\frac{e d}{s})^s$ and each subspace (one per fixed sparse signal support) has dimension $s$, or more advanced models with structured sparsity \cite{BCD10,ayaz2016uniform}. For such a model, the Kolmogorov entropy is bounded by $C s \cdot \log\big(\frac{1}{\eta}\big) \leq \cl H_\eta(\Sigma) \leq C' s \cdot \log\big(1 + \frac{2}{\eta}\big) + \log S$ (see, \eg \cite[Lemma 10]{JC17timefordithering}). Assume that we target the usual Gaussian kernel with unit bandwidth, hence $C_{\Lambda} = 1$. In this case, the RFF kernel approximation error is uniformly bounded over $\Sigma$, with high probability, provided that the number of features satisfies $m \geq C \epsilon^{-2} \big( s \log(\frac{1}{c \epsilon}) + \log S\big)$. For instance, for bounded $s$-sparse signals we need $m \geq C \epsilon^{-2}  s  \log(\frac{e d}{c s\, \epsilon})$.    
\end{example}

\subsection{Random periodic features}
\label{sec:rand-peri-feat}

A crucial generalization to RFF has been proposed in~\cite{boufounos2017representation}, where the complex exponential is replaced by a \emph{generic periodic function} $f$. We refer to this approach as \emph{random periodic features} (RPF). Without loss of generality, we make the following normalization assumptions throughout this work: $f$ has period given by $2\pi$, is centered (zero mean), and takes (absolute) values bounded by one. We note this compactly as
\begin{equation*}
  f \in \pf, \quad \text{with} \quad \pf := \ts \{ f : \bb R \rightarrow \bb C \: | \: f \text{ is $2\pi$-periodic}, \: \int_{0}^{2\pi} f(t) \mathrm{d}t = 0, \: \|f\|_{\infty} \leq 1 \}.
\end{equation*}
Functions of PF can be expressed as a Fourier series of the following form
\begin{equation}
\label{eq:decomp-f-F-series}
  \ts f(t) = \sum_{k \in \bb Z} F_k e^{\im k t}, \quad \text{where} \quad F_k := \frac{1}{2\pi} \int_{0}^{2\pi} f(t) e^{-\im k t} \mathrm{d}t.
\end{equation}
Note that $f \in \pf$ implies $F_0 = 0$ (because $f$ is centered) and $|F_k| \leq 1$ (because $f$ is bounded).

\begin{definition}[Random periodic features]
  \label{def:generalPeriodicFeatures}
  Let $f$ be a generic periodic function, normalized such that $f \in \pf$, and $\Lambda$ a probability distribution on $\bb R^d$. Given a target dimension $m$, the associated \emph{random periodic features} are
  \begin{equation}
    \label{eq:def_periodic_features}
    \ts \Vec{z}_f(\Vec{x}) := \frac{1}{\sqrt{m}} f( \bs\Omega^\top \Vec{x} + \Vec{\xi}) \quad \in \bb C^m,
  \end{equation}
  with a $d \times m$ random projection matrix $\bs\Omega := (\Vec{\omega}_1,\,\cdots, \Vec{\omega}_m) \sim \Lambda^m$, and a random dither $\Vec{\xi} \sim \cl U^m([0,2\pi))$.
\end{definition}

\begin{remark}
  As the complex exponentiation satisfies $\exp(\im \cdot) \in \pf$, this definition includes the classical random Fourier features, with $\Vec{z}(\Vec{x}) = \Vec{z}_{\exp(\im \cdot)}(\Vec{x})$. The real RFF $\Vec{z}_{\cos}(\Vec{x})$ are also a particular case of this definition.       
\end{remark}

The geometry induced by such generic features can be characterized the inner product $\wh{\kappa}_{f,f}(\bs x, \bs y) := \langle \Vec{z}_f(\Vec{x}), \Vec{z}_f(\Vec{y}) \rangle$. As explained by the following result (adapted from~\cite[Theorem 4.4]{boufounos2017representation}), this product is associated with a modified kernel $\kappa_{f,f}(\bs x,\bs y)$ (the rationale for these notations is clarified in the next section).

\begin{proposition}[Kernel from symmetric RPF]
  \label{prop:BoufounosKernel}
	The inner product of random periodic features~\eqref{eq:def_periodic_features} approaches, on average, a kernel $\kappa_{f,f}(\Vec{x},\Vec{y}) := \expec{} \langle \Vec{z}_{f}(\Vec{x}), \Vec{z}_f(\Vec{y}) \rangle$ that is shift-invariant and given by
  \begin{equation}
    \label{eq:symmetric_kernel_distortion}
    \ts \kappa_{f,f}(\Vec{x},\Vec{y}) =  \sum_{k \in \bb Z} |F_k|^2 \sik(k(\Vec{x} - \Vec{y})) =: \sik_{f,f}(\Vec{x}-\Vec{y}),
  \end{equation}
  where $\sik(\Vec{u}) = (\cl F^{-1} \Lambda)(\Vec{u})$ is the shift-invariant kernel associated with the distribution of $\bs\Omega$ in the RPF.
\end{proposition}
\begin{proof}
  This version is obtained as a particular case of our Prop.~\ref{prop:expected_kernel}; see~\cite{boufounos2017representation} for the initial result.
\end{proof}

The modified kernel $\kappa_{f,f}$ is thus a \emph{scale mixture} of the initial kernel $\kappa$ (that is approached by the ``classical'' RFF), where the weight of scale $k$ is given by $|F_k|^2$. In the non-asymptotic case, the authors of~\cite{boufounos2017representation} show that, for all pairs of vectors taken in a finite set $\Sigma$ of size $N$, $\wh{\kappa}_{f,f}(\bs x, \bs y)$ quickly concentrates around $\kappa_{f,f}(\bs x, \bs y)$ when $m$ is large compared to $\log N$; the deviation error scaling as $O(\sqrt{\log N/m})$ when $m$ increases. Our result in Prop.~\ref{prop:kernelApproxUniform} provides a uniform approximation bound valid for infinite sets.

Random periodic features were introduced as a general theoretical framework to analyze the so-called universal quantization embeddings~\cite{Boufounos2013efficientCodingQuantized}; those binary embeddings encode the local distances (\ie the distances below a given threshold) on an efficiently small number of bits. This embedding relies on the ``one-bit universal quantization'' given by $\cl Q_{\Delta} : \bb R \rightarrow \{0,1\} : t \mapsto \cl Q_{\Delta}(t) = 1$ if $(2k-1) \leq t/\Delta \leq 2k$  for any $k\in \bb Z$ and $0$ otherwise. It can be interpreted as the least significant bit of a usual, plain scalar quantizer with stepsize $\Delta$, and visualized as a square wave: see Fig.~\ref{fig:univquant}, left. Here, we will for convenience use $q$ instead, its normalized equivalent in $\pf$,
\begin{equation}
  \label{eq:universalquantizer}
  \ts	q(t) := \sign \circ \cos (t) = \sum_{k \in \bb Z} Q_k e^{\im k t}, \quad \text{with coefficients} \ \, Q_k = \begin{cases}
    \frac{2}{k \pi}(-1)^{(k-1)/2}  & \text{if } k \text{ odd,} \\
    0        & \text{if } k \text{ even.}
  \end{cases} 
\end{equation}

\begin{figure}
  \centering
  \subfloat{
    \includegraphics[scale=0.45]{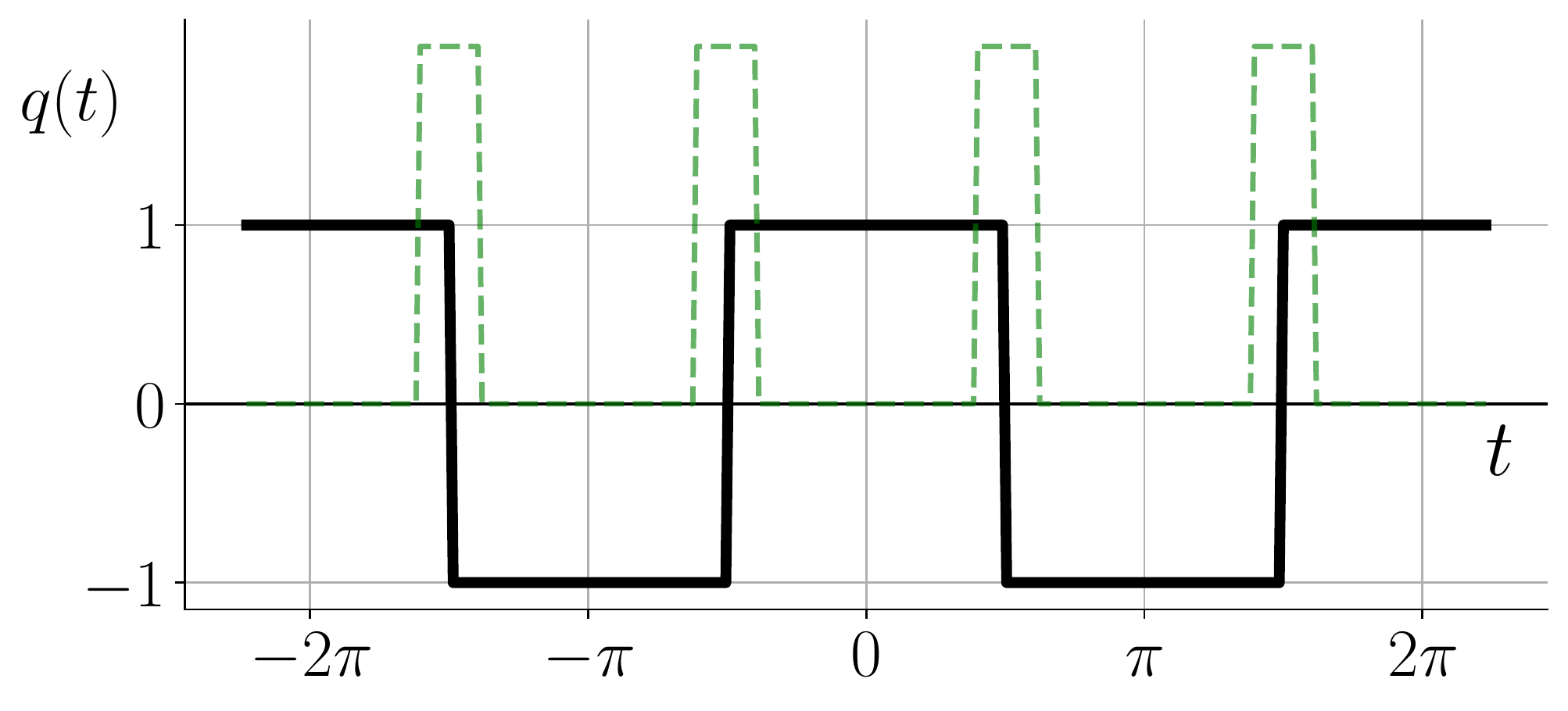}}
  \hspace{23px}
  \subfloat{\raisebox{-1mm}	{\includegraphics[scale=0.47]{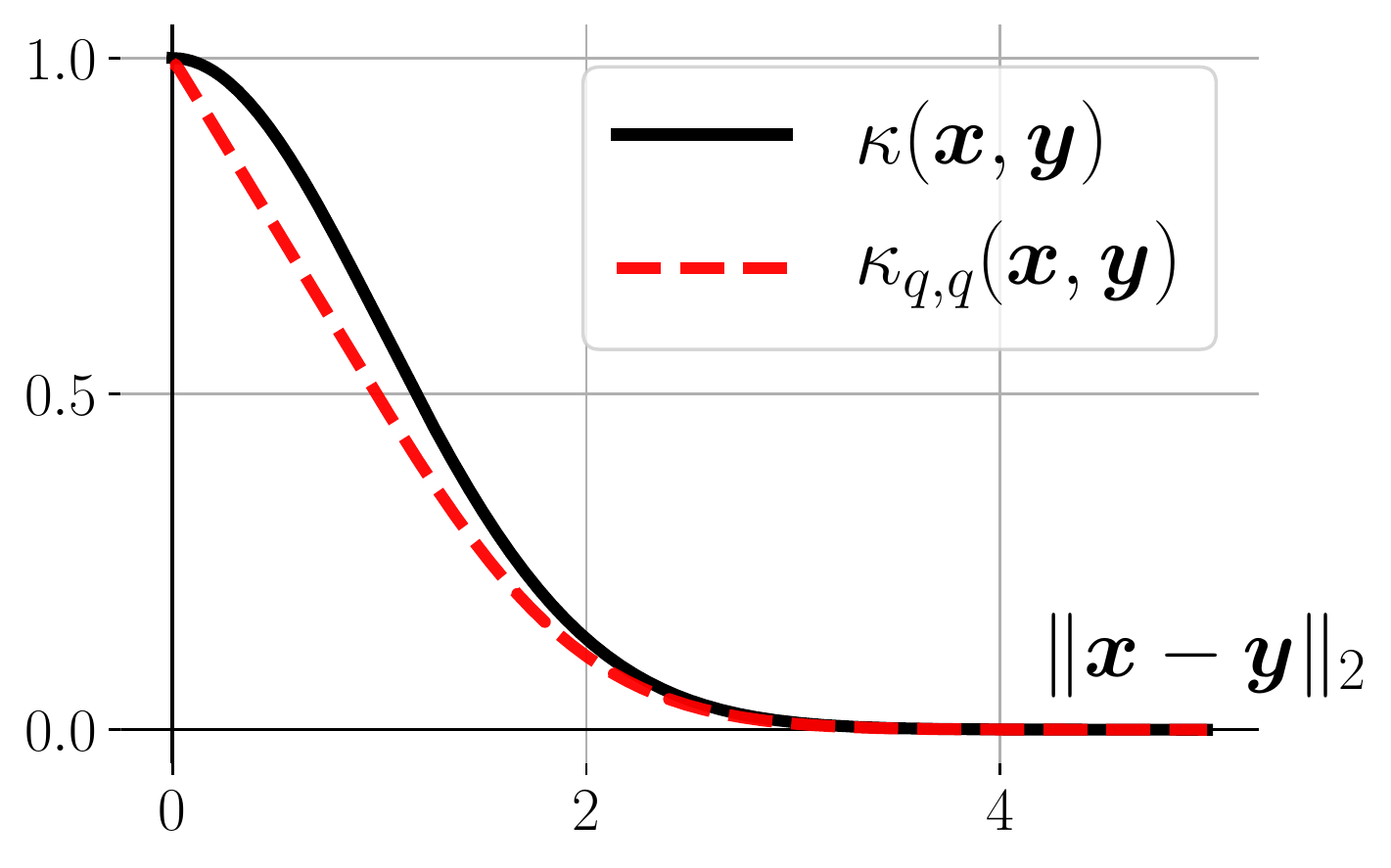}}}
  \caption{(Left) The solid black curve represents the universal quantization function $q(t)$ (with $q\in \pf$) defined in~\eqref{eq:universalquantizer}. Up to a shift and rescaling, this function corresponds to the \emph{least significant bit} of a standard uniform scalar quantizer. In dashed green, we display the related integrand $I_{\delta}(t)$, with $\delta = 0.35$. This quantity refers to the proof of the mean smoothness (Def.~\ref{def:meanSmoothness}) of $q$ in Prop.~\ref{prop:meanSmoothUniversalQuantization} (Sec.~\ref{sec:quantized}). (Right) When drawing $\bs \Omega$ from a Gaussian distribution $\Lambda = \cl N(0,\bs I_d)$, the associated RFF recover the Gaussian kernel $\kappa$ (in black), but the RPF with universal quantization approximate a ``distorted'' kernel $\kappa_{q,q}$ (dashed red) whose almost linear behavior close to the origin is explained by \eqref{eq:linear-kernel-qq-at-origin}.}
  \label{fig:univquant}
\end{figure}

Using the universal quantization as periodic nonlinearity is appealing because $\Vec{z}_q(\Vec{x}) \in \{-1,+1\}^m$, which can thus be encoded/transmitted by only $m$ bits. However, as predicted by~\eqref{eq:symmetric_kernel_distortion}, the approximated kernel is modified, as illustrated for the Gaussian kernel Fig.~\ref{fig:univquant}, right. Moreover, proving uniform kernel approximation bounds (as in Prop.~\ref{prop:RFF}) for infinite sets $\Sigma$ is specially challenging when the nonlinearity $f$ presents discontinuities (which is the case when $f = q$, for example). In~\cite{boufounos2017representation}, the authors introduced a formalism (the $T$-part Lipschitz functions) to deal with this problem and to obtain uniform approximation bounds on infinite signal sets for the universal embeddings. As we explain in Sec.~\ref{sec:correctPTB}, the proof relying on this approach is however wrong, which motivates us to introduce another tool, the mean Lipschitz smoothness, to deal with discontinuous maps.

\section{Expected kernel (asymptotic case)}
\label{sec:asymptotic}

Following the considerations of the Introduction, let us now consider the \emph{asymmetric features} setting where a pair of signals of interest, $\Vec{x},\Vec{y} \in \bb R^d$, are available only through their random periodic features, $\Vec{z}_{f}(\Vec{x})$ and $\Vec{z}_{g}(\Vec{y})$, as defined in~\eqref{eq:def_periodic_features}. Those features are allowed to result from different periodic maps $f, g \in \pf$, but the preceding projection $\bs\Omega$ and dithering $\Vec{\xi}$ are kept identical.

In this section, we characterize the properties of the \emph{expected kernel} yielded by the expectation, over the draw of $\bs \Omega$ and $\bs \xi$, of the following ``asymmetric'' inner product:
\begin{equation}
  \wh{\kappa}_{f,g}(\Vec{x},\Vec{y}) := \langle \Vec{z}_{f}(\Vec{x}), \Vec{z}_{g}(\Vec{y}) \rangle.
\end{equation}

\noindent This asymmetric RPF kernel is defined from
\begin{equation}
  \label{eq:def_expected_kernel}
  \begin{split}
    \kappa_{f,g}(\Vec{x},\Vec{y}) &:= \bb E_{\bs\Omega,\, \Vec{\xi}} \: \langle\Vec{z}_f(\Vec{x}),\Vec{z}_g(\Vec{y})\rangle\\
    &\: = \ts \frac{1}{m} \sum_{j = 1}^m \expec{\Vec{\omega}_j,\xi_j} f(\Vec{\omega}_j^\top\Vec{x} + \xi_j)\: g^*(\Vec{\omega}_j^\top\Vec{y} + \xi_j)\\
    &\: = \ts \expec{\Vec{\omega} \sim \Lambda,\,\xi \sim \cl U([0,2\pi))} f(\Vec{\omega}^\top\Vec{x} + \xi)\: g^*(\Vec{\omega}^\top\Vec{y} + \xi).
  \end{split}
\end{equation}
In the two bottom lines, we used the fact that $\Vec{\omega}_j$ (and $\xi_j$) are independently and identically distributed, for all $j$, with $\Vec{\omega}_j$ and $\xi_j$ mutually independent. Remark that by the law of large numbers, $\kappa_{f,g}$ thus corresponds to the kernel that the asymmetric inner product $\wh{\kappa}_{f,g}$ approximates when we let the feature space dimension $m$ grow to infinity.

\begin{proposition}[Expected kernel for asymmetric periodic random features]
  \label{prop:expected_kernel}
  Let $\Vec{z}_f$ and $\Vec{z}_g$ be random periodic features, associated with functions $f,g \in \pf$, frequencies $\Vec{\omega}_j \distiid \Lambda = \cl F \sik$ and $\xi_j \distiid \cl U([0,2\pi))$. For any pair $\Vec{x},\Vec{y} \in \bb R^d$, the expected kernel $\kappa_{f,g}(\Vec{x},\Vec{y}) = \bb E_{\bs\Omega,\, \Vec{\xi}} \: \langle\Vec{z}_f(\Vec{x}),\Vec{z}_g(\Vec{y})\rangle$ satisfies
  \begin{equation}
    \label{eq:expected_kernel_expression}
    \ts \kappa_{f,g}(\Vec{x},\Vec{y}) = \sum_{k \in \Zbb} \, F_kG_k^* \, \sik(k(\Vec{x}-\Vec{y})) =: \sik_{f,g}(\Vec{x}-\Vec{y}).
  \end{equation}
  Although here expanded as an infinite series, this kernel is bounded $|\kappa_{f,g}| \leq 1$ since $f,g \in \pf$.
\end{proposition}
\begin{proof}
  Starting from the last line of~\eqref{eq:def_expected_kernel}, and decomposing $f$ and $g$ as their Fourier series,
  \begin{align}
    \kappa_{f,g}(\Vec{x},\Vec{y}) &= \ts \expec{\Vec{\omega},\xi} \: \sum_{k \in \bb Z}\sum_{k' \in \bb Z} F_kG^*_{k'} \: e^{\im k ( \Vec{\omega}^\top\Vec{x} + \xi)}e^{-\im k' ( \Vec{\omega}^\top\Vec{y} + \xi)} \nonumber\\
                                  &= \ts \sum_{k,k'} F_kG^*_{k'}  \,\expec{\Vec{\omega} \sim \Lambda} e^{\im \Vec{\omega}^\top(k\Vec{x}-k'\Vec{y})}  \,\expec{\xi \sim \cl U([0,2\pi))} e^{\im (k-k')\xi}\nonumber\\
                                  &= \ts \sum_{k,k'} F_k G_{k'}^* \, \sik( k \Vec{x} - k' \Vec{y}) \, \delta_{k,k'}\label{eq:tmp1}\\
                                  &= \ts \sum_{k}  F_k G_{k}^* \, \sik( k (\Vec{x} - \Vec{y})), \nonumber
  \end{align}
  where in the third line we used Bochner's theorem~\eqref{eq:bochner} and the orthogonality of complex exponentials on one period: $\frac{1}{2\pi}\int_{0}^{2\pi}e^{\im kt}e^{-\im k't}\mathrm{d}t = \delta_{k,k'}$.
\end{proof}

\begin{example}
	As will be further developed in Sec.~\ref{sec:quantized}, when $f(\cdot) = \cos(\cdot)$ and $g(\cdot) = q(\cdot)$ the universal quantization defined in~\eqref{eq:universalquantizer}, we observe that the expected kernel is (up to a proportionality constant) exactly the ``initial'' kernel approximated by the RFF, \ie $\kappa_{\cos,q}(\Vec{x},\Vec{y}) = \frac{2}{\pi} \kappa(\Vec{x},\Vec{y})$.
\end{example}

The dither $\Vec{\xi}$ plays here a crucial role: it cancels out (in expectation) the ``cross-terms'' in~\eqref{eq:tmp1}, each related to $F_k G_{k'}^* \, \sik( k \Vec{x} - k' \Vec{y})  = F_kG_{k'}^*\,\kappa(k\Vec{x},k'\Vec{y})$, that have different scales $k \neq k'$ for~$\Vec{x}$ and~$\Vec{y}$. As a consequence, the expected kernel is---as any kernel should be---conjugate symmetric, \ie $\kappa_{f,g}(\Vec{x},\Vec{y}) = \kappa^*_{f,g}(\Vec{y},\Vec{x})$, despite the asymmetry of its empirical approximation,  \ie $\wh{\kappa}_{f,g}(\Vec{x},\Vec{y}) \neq \wh{\kappa}^*_{f,g}(\Vec{y},\Vec{x})$. The dithering can thus be thought of as a means to symmetrize, through expectation, the kernel associated with the asymmetric features inner product.

For the same reason, the dithering ensures that the expected kernel remains shift-invariant; Prop.~\ref{prop:expected_kernel} provides $\kappa_{f,g}(\Vec{x},\Vec{y}) = \sik_{f,g}(\Vec{x}-\Vec{y})$, where $\sik$ in~\eqref{eq:expected_kernel_expression} is the kernel related to the frequency sampling pattern $\Lambda$. The expected kernel is thus a \emph{scale mixture}, a linear combination of copies of $\sik$, scaled (actually contracted) by an integer factor $k$ (which is non-zero, since $F_0 = G_0 = 0$), and weighted by coefficients $F_kG^*_k$. In general, we expect this scale mixture $\sik_{f,g}$ to be narrower than the initial kernel $\sik$ (or more spread out in the frequency domain).

In general, however, the positive definiteness of $\kappa$ does not imply that $\kappa_{f,g}$ is \pd, since taking, for instance, $g = -f$ (\ie $F_kG_k^* = -|F_k|^2 < 0$) induces that $\kappa_{f,g}(\Vec{x},\Vec{x}) = - \sum_k |F_k|^2 \sik(\Vec{0}) < 0$ for all $\bs x \in \bb R^d$. Whether $\kappa_{f,g}$ is a positive definite kernel or not depends on the phase synchronization between the Fourier coefficients of $f$ and $g$. A sufficient condition for $\kappa_{f,g}$ to be \pd is to ensure that $F_kG_k^* \in \bb R^+$ for all $k$, as verified by taking $f = g$ in Sec.~\ref{sec:correctPTB}, or the combination $f=q$, $g(\cdot) = \cos(\cdot)$ in Sec.~\ref{sec:quantized} and Sec.~\ref{sec:experiments}.

\begin{remark}
  \label{rem:other-normalisation}
In light of \eqref{eq:expected_kernel_expression}, we could decide to normalize our approach differently. Assuming that $f,g \in \pf$ are not orthogonal, \ie $\scp{f}{g} \neq 0$, we can define, for $\bs x,\bs y \in \Sigma$, the normalized kernels
\begin{equation}
  \label{eq:other-normalisation}
  \ts \wt{\kappa}_{f,g}(\Vec{x},\Vec{y}) := \frac{1}{\scp{f}{g}} \langle \Vec{z}_{f}(\Vec{x}), \Vec{z}_{g}(\Vec{y}) \rangle,\quad \dot{\kappa}_{f,g}(\bs x,\bs y) := \frac{1}{\scp{f}{g}}  \kappa_{f,g}(\bs x,\bs y),\quad \dsik_{f,g}(\bs u) := \frac{1}{\scp{f}{g}}  \sik_{f,g}(\bs u).   
\end{equation}
Since $\sik(\bs 0) = 1$ and, from \eqref{eq:expected_kernel_expression}, $\scp{f}{g} = \sum_k F_k G^*_k = \sik_{f,g}(\bs 0)$, \eqref{eq:other-normalisation} ensures that, for any $\bs x \in \Sigma$, ${\bb E\, \wt{\kappa}_{f,g}(\bs x,\bs x)} = \dot{\kappa}_{f,g}(\bs x,\bs x) = \dsik_{f,g}(\bs 0) = 1 = \sik(\bs 0)$. Without guaranteeing that $\dot{\kappa}_{f,g}$ is \pd, this normalization prevents the counterexample $f=-g$ to lead to a kernel with negative value on the origin. For clarity, we do not base our following developments on $\dot{\kappa}_{f,g}(\bs x,\bs y)$ but we will refer to this useful quantity in Sec.~\ref{sec:quantized} when, for $f=q$ and $g(\cdot)=\cos(\cdot)$, we will need to compare $\dsik_{f,g}$ to the RFF kernel $\sik$. 
\end{remark}

Let us now provide an alternative expression of the expected kernel $\kappa_{f,g}(\Vec{x},\Vec{y}) = \sik_{f,g}(\Vec{x}-\Vec{y})$, that will prove to be useful in the next section. 
\begin{lemma}[]
  \label{lem:expected_kernel_with_convolution}
  Define the \emph{correlation} $h$ between $f$ and $g$,
  \begin{equation}
    h(t) := \ts (f * \bar{g})(t) = \frac{1}{2\pi}\Integr{0}{2\pi}{f(\tau)g^*(\tau-t)}{\tau},
  \end{equation}
  where $\bar{g}(t) := g^*(-t)$ denotes the conjugate reverse of $g$, and $*$ the convolution operator on $[0,2\pi]$. The expected (shift-invariant) kernel $\sik_{f,g}$ can be expressed by
  \begin{equation}
    \sik_{f,g}(\Vec{u}) = \ts \expec{\Vec{\omega} \sim \Lambda} h(\Vec{\omega}^\top\Vec{u}) = \expec{\Vec{\omega} \sim \Lambda} (f * \bar{g})(\Vec{\omega}^\top\Vec{u}).
  \end{equation}
\end{lemma}
\begin{proof}
  By the convolution theorem, the Fourier series coefficients of $h$ are given by $H_k = F_kG_k^*$. The result follows from plugging this fact into the proof of Prop.~\ref{prop:expected_kernel}.
\end{proof}

\noindent Lemma~\ref{lem:expected_kernel_with_convolution} can be interpreted as an expansion similar to the one Bochner's theorem provides: whereas the initial kernel $\sik(\Vec{u}) = \Integr{}{}{e^{\im \Vec{\omega}^\top\Vec{u}}}{\Lambda(\Vec{\omega})}$ can be expressed in a basis that is a family complex exponentials $e^{\im \Vec{\omega}^\top\Vec{u}}$ with ``coordinates'' given by $\Lambda$, $\sik_{f,g}(\Vec{u}) = \Integr{}{}{h(\Vec{\omega}^\top\Vec{u})}{\Lambda(\Vec{\omega})}$ can be expressed in a basis that is the family of functions $\{ \bs u \to h(\Vec{\omega}^\top\Vec{u}): \bs \omega \in \bb R^d\}$, another type of $2\pi-$periodic functions (replacing the complex exponential $\exp(\im \cdot)$ with $h(\cdot)$). 

As a side note, when $\Lambda \sim \cl N(0, \Id_d)$ and $q$ is the square wave representing the universal quantization, Lemma~\ref{lem:expected_kernel_with_convolution} allows us to easily explain the linear slope of $\sik_{q,q}$ at the origin (see Fig.~\ref{fig:univquant}, left). Indeed, in this case $h$ is the autocorrelation of $q$, the triangular wave
$$
\ts h(t) = (q \ast \bar q)(t) = \max(1 - \frac{|t'|}{\pi}, \frac{|t'|}{\pi} - 1),\quad \text{with}\ t' := t \!\!\!\mod 2\pi.
$$
Therefore, defining $\tilde h(t) := 1 - {|t|}/{\pi}$, since $h(t) = \tilde h(t)$ for $|t| < \pi$, and $0 \leq h(t) - \tilde h(t) \leq 2 (|t|-\pi)/\pi$ for $|t| \geq \pi$, we find 
$$
\ts \sik_{q,q}(\bs u) = \bb E_{\bs \omega \sim \Lambda} h(|\bs \omega^\top \bs u|) = \bb E_{\bs \omega \sim \Lambda} \tilde h(|\bs \omega^\top \bs u|) + R(\bs u, \Lambda), 
$$
with $R(\bs u, \Lambda) := \bb E_{\bs \omega \sim \Lambda} [h(|\bs \omega^\top \bs u|) - \tilde h(|\bs \omega^\top \bs u|)] \geq 0$.
Since $\bs \omega^\top \bs u \sim \|\bs u\| g$ with $g \sim \cl N(0,1)$, and $\bb E|g| = \sqrt{2/\pi}$, we find
\begin{align*}
  R(\bs u, \Lambda)&\ts \leq \frac{2}{\sqrt{2\pi}}\,\int_{\pi/\|\bs u\|}^{+\infty} \frac{2}{\pi} (r \|\bs u\|  - \pi) e^{-\frac{r^2}{2}}\ud r = \frac{c}{\|\bs u\|}\,\int_{\pi}^{+\infty} (s - \pi) e^{-\frac{s^2}{2 \|\bs u\|^2}}\ud s\\
                   &\ts = \frac{c}{\|\bs u\|}\,\int_{0}^{+\infty} s e^{-\frac{s^2 + \pi^2 + 2\pi s}{2 \|\bs u\|^2}}\ud s \leq \frac{c}{2} \,e^{-\frac{\pi^2}{2 \|\bs u\|^2}}\,\bb E[\|\bs u\| |g|] = c' \|\bs u\| e^{-\frac{\pi^2}{2 \|\bs u\|^2}}.
\end{align*}
Moreover, $\bb E_{\bs \omega \sim \Lambda} h(|\bs \omega^\top \bs u|) = 1 - \frac{\|\bs u\|}{\pi} \bb E_{\bs \omega \sim \Lambda} |g| = 1 - \frac{\sqrt 2}{\pi^{3/2}} \|\bs u\|$, which finally proves that
\begin{equation}
  \label{eq:linear-kernel-qq-at-origin}
  \ts \big | \sik_{q,q}(\bs u) - (1 - \frac{\sqrt 2}{\pi^{3/2}} \|\bs u\|)\big| = O\big(\|\bs u\| e^{-\frac{\pi^2}{2\|\bs u\|^2}}\big).  
\end{equation}
This shows that for $\|\bs u\|\ll \pi$,  $\sik_{q,q}(\bs u) \asymp 1 - \frac{\sqrt 2}{\pi^{3/2}} \|\bs u\|$.

\section{Approximation error analysis (non-asymptotic case)}
\label{sec:nonasymptotic}

In the practical setting where the vectors $\Vec{z}_{f}(\Vec{x})$ and $\Vec{z}_{g}(\Vec{y})$ are to be quickly processed or stored in memory, their size $m$ must be as small as possible. On the other hand, setting $m$ too small hurts the empirical estimation $\wh{\kappa}_{f,g}(\Vec{x},\Vec{y})$ of the expected kernel $\kappa_{f,g}(\Vec{x},\Vec{y}) = \bb E \, \wh{\kappa}_{f,g}(\Vec{x},\Vec{y})$. To understand this trade-off, we are thus interested in a probabilistic bound for the (absolute) kernel approximation error $|\wh{\kappa}_{f,g} - \kappa_{f,g}|$, as a function of the RPF dimension $m$. We give here an answer to this question under generic assumptions, and show how to apply it in a concrete situation---for asymmetric kernel estimation with one-bit quantized RFF---in Sec.~\ref{sec:quantized}.

\subsection{Non-uniform approximation error}
\label{sec:non-unif-appr}

Ultimately, we want to obtain a (probabilistic) bound for the kernel approximation error that holds uniformly over all $\Vec{x},\Vec{y} \in \Sigma$. First bounding the error for one \emph{fixed} pair $(\Vec{x},\Vec{y})$ is often used as an easier intermediary step. This is provided by the following proposition.
\begin{proposition}[Non-uniform kernel approximation error from asymmetric periodic random features]
  \label{prop:kernelApproxFixedPair}
  For two functions $f,g \in \pf$, let $\Vec{z}_f, \Vec{z}_g$ be random periodic features associated with frequencies $\bs\Omega$ and a dither $\Vec{\xi}$.  For any \emph{fixed} pair $(\Vec{x},\Vec{y}) \in \bb R^d \times \bb R^d$, the inner product $\wh{\kappa}_{f,g}(\Vec{x},\Vec{y}) = \langle \Vec{z}_{f}(\Vec{x}), \Vec{z}_{g}(\Vec{y}) \rangle$ concentrates, in probability over the draw of $\bs\Omega \sim \Lambda^m, \Vec{\xi} \sim \cl U^m([0,2\pi))$, around $\kappa_{f,g}(\Vec{x},\Vec{y}) = \bb E_{\bs\Omega,\, \Vec{\xi}} \: \langle\Vec{z}_f(\Vec{x}),\Vec{z}_g(\Vec{y})\rangle$ as
  \begin{equation}
    \mathbb{P} \left[\left| \wh{\kappa}_{f,g}(\Vec{x},\Vec{y}) - \kappa_{f,g}(\Vec{x},\Vec{y}) \right| \leq \epsilon\right] \geq 1 - 2e^{-m\epsilon^2/2}.
  \end{equation}
\end{proposition}
\begin{proof}
  We rewrite $\langle \Vec{z}_{f}(\Vec{x}),\Vec{z}_{g}(\Vec{y}) \rangle = \frac{1}{m}\sum_j Z_j$, with the random variables $Z_j := f(\Vec{\omega}_j^\top\Vec{x} + \xi_j)g^*(\Vec{\omega}_j^\top\Vec{y} + \xi_j)$. The $Z_j$ variables are \iid, have mean $\kappa_{f,g}(\Vec{x},\Vec{y})$ by definition of the expected kernel~\eqref{eq:def_expected_kernel}, and are bounded by $|Z_j| \leq \|f\|_{\infty}\|g\|_{\infty} \leq 1$ (because $f,g \in \pf$). The result follows by Hoeffding's inequality.
\end{proof}

\subsection{Uniform approximation error}
\label{sec:unif-appr-error}

We now want to extend the error bound in Prop.~\ref{prop:kernelApproxFixedPair} to hold not only for one fixed pair $(\Vec{x},\Vec{y})$ but simultaneously over all pairs $(\Vec{x},\Vec{y}) \in \Sigma \times \Sigma$; this is called a \emph{uniform bound}. The classical argument invoked in this type of proofs (\eg~\cite{Rahimi2008RFF,baraniuk2008simple,puy2017recipes2}) goes as follows. If $\Sigma$ is a finite set (of finite cardinality $|\Sigma|$), the uniform bound is obtained by applying a union bound over $|\Sigma|^2$ instances of Prop.~\ref{prop:kernelApproxFixedPair} (one for each pair in $\Sigma \times \Sigma$). In the case where $\Sigma \subset \bb R^d$ is an infinite but compact set, the strategy is to bound the approximation error on a finite set $\Sigma_{\eta}$ that covers $\Sigma$ by balls of some radius $\eta > 0$, then to extend this bound by some notion of continuity (smoothness) over the $\eta-$balls. We then obtain a bound which holds over $\Sigma_{\eta} + \eta \bb B^d_2 \supseteq \Sigma$, which concludes the proof.

In our setting, the last step of this proof technique would ideally use Lipschitz continuity; we say that a function $f : \bb R \rightarrow \bb C$ is \emph{Lipschitz continuous} with constant $L_f$ if, for all $t,t' \in \bb R$, $|f(t) - f(t')| \leq L_f |t-t'|$, which is equivalent to  
\begin{equation}
  \label{eq:Lipschitz}
  \forall t \in \bb R, \forall \delta > 0, \quad \sup_{r \in [-\delta, \delta]} \{ |f(t + r) - f(t)| \} \leq L_f \cdot \delta.
\end{equation}
However, this strategy fails when any of the maps $f$ or $g$ is not Lipschitz continuous (\eg when they present discontinuities, such as the ``square wave'' universal quantization map $q$ from~\eqref{eq:universalquantizer}). To be able to include such maps in our analysis, we must define a more permissive notion of smoothness, just as the $T$-part Lipschitz property defined in~\cite{boufounos2017representation} (but without the limitations explained in Sec.~\ref{sec:correctPTB}). In this work, we rather introduce the concept of \emph{mean Lipschitz smoothness property} for periodic function in $\pf$. Intuitively, a periodic function is smooth in the mean Lipschitz sense if its largest local deviation is small \emph{on average}. 

\begin{definition}[Mean Lipschitz property]
  \label{def:meanSmoothness}
  Let $f : \bb R \rightarrow \bb C$ be a generic periodic function (here \wlogg assumed of period $2\pi$). We say it is \emph{mean Lipschitz smooth} with mean Lipschitz constant $L^{\mu}_f$ if for all radii $\delta \in (0,\pi]$, the average maximum deviation of $f$ in $[-\delta,\delta]$ is bounded by $L^{\mu}_f \cdot \delta$:
  \begin{equation}
    \label{eq:meanSmooth}
    \expec{t \sim \cl U([0,2\pi))} \sup_{r \in [-\delta, \delta]} \{ |f(t + r) - f(t)| \} = \tfrac{1}{2\pi} \int_{0}^{2\pi} \sup_{r \in [-\delta, \delta]} \{ |f(t + r) - f(t)| \} \: \mathrm{d}t \leq L^{\mu}_f \cdot \delta.
  \end{equation}
\end{definition}

The mean Lipschitz property (that we will refer to as ``mean smoothness'' to avoid confusion with the usual Lipschitz continuity when necessary) can truly be understood as the Lipschitz continuity after an averaging.
It is reminiscent of the mean modulus of continuity from~\cite{szasz1937fourier,wik1972criteria} but where the order of the supremum and averaging operations are reversed (which is less restrictive).
If $f$ is Lipschitz continuous with Lipschitz constant $L_f$, then it has necessarily also the mean smoothness property with constant $L^{\mu}_f \leq L_f$. However, it is possible that $L^{\mu}_f \ll L_f$ (if the large slopes of $f$ are concentrated on a small portion of $[0,2\pi]$), and discontinuous function can have a finite $L^{\mu}_f$ constant---for example, the square wave $q$ representing the universal quantization is mean smooth with constant $L_q^{\mu} = \frac{4}{\pi}$ (see Prop.~\ref{prop:meanSmoothUniversalQuantization}), although it is not a Lipschitz continuous function. Leaving the detailed proof to Sec.~\ref{sec:quantized}, the trick is to observe that the integrand $I_{\delta}(t) = \sup_{r \in [-\delta, \delta]} \{ |f(t + r) - f(t)| \}$ is supported on an interval whose length is proportional to~$\delta$, as shown Fig.~\ref{fig:univquant}, left. Moreover, the convolution of any $\pf$ function with a mean smooth $\pf$ function yields a Lipschitz continuous one.
\begin{lemma}
  \label{lem:mean_smooth_convolution}
  Given two functions $f,g \in \pf$, among which $f$ is mean smooth with constant $L^{\mu}_f$, their convolution $(f * g)$ is Lipschitz continuous with constant $L_{f * g} \leq L^{\mu}_f$.
\end{lemma}
\begin{proof}
  Re-writing~\eqref{eq:Lipschitz} for $(f * g)(t) = \frac{1}{2\pi}\Integr{0}{2\pi}{f(t - \tau) g(\tau)}{\tau}$ gives, since $\|g\|_{\infty} \leq 1$,
  \begin{equation*}
    \begin{split}
      \ts \sup_{r\in [-\delta, \delta]} \{ |(f*g)(t + r) - (f*g)(t)| \} &= \ts \sup_{|r|\leq \delta} \left|\frac{1}{2\pi} \Integr{0}{2\pi}{\left[f(t + r - \tau) - f(t - \tau) \right]g(\tau)}{\tau} \right|\\
      &\leq \ts \frac{1}{2\pi} \sup_{|r|\leq \delta}  \Integr{0}{2\pi}{\left|f(t + r - \tau) - f(t - \tau)\right|  }{\tau}\\
      &\leq \ts \frac{1}{2\pi} \Integr{0}{2\pi}{ \sup_{|r|\leq \delta} \left|f(\tau' + r) - f(\tau')\right|  }{\tau'} \leq \ts L^{\mu}_f \delta.
    \end{split}
  \end{equation*}
\end{proof}

\noindent In other words, the convolution of two $\pf$ functions, among whom one of them is mean Lipschitz, is ``smoother'' than its factors, a property that comes from the convolution itself (which is to be put in correspondence with the fact that for $f$ differentiable and $g$ discontinuous, $f * g$ is differentiable).
In particular, if \textit{both} $f$ and $g$ are mean smooth with constants $L^{\mu}_f$ and $L^{\mu}_g$ respectively, their \emph{correlation} $h = f * \bar{g}$ is Lipschitz with $L_h \leq \min(L^{\mu}_f,L^{\mu}_g)$. 
Coming back to our setting, this fact allows us (using Lemma~\ref{lem:expected_kernel_with_convolution}) to characterize the Lipschitz continuity of the expected kernel $\sik_{f,g}$. With that, we have all the tools to prove our main result, a uniform bound on the kernel approximation error obtained with possibly discontinuous (but mean smooth) maps.

\begin{proposition}[Uniform kernel approximation error from asymmetric periodic random features]
  \label{prop:kernelApproxUniform}
  Let $\Sigma$ be a compact set and $f,g \in \pf$ periodic functions with finite mean smoothness constants $L^{\mu}_f$ and $L^{\mu}_g$, respectively, and let $C_{\Lambda} < \infty$ such that $\expec{\Vec{\omega} \sim \Lambda} |\Vec{\omega}^\top\Vec{a}| \leq C_{\Lambda} \|\Vec{a}\|_2$ for all $\Vec{a}$ (the kernel smoothness constant).
  
  \noindent For all error level $\epsilon > 0$, provided the feature dimension is larger than
  \begin{equation}
    \label{eq:unif-kern-approx-sample-complex}
    \ts m \geq 128 \cdot \frac{1}{\epsilon^2} \cdot \cl H_{\epsilon/c}(\Sigma),
  \end{equation}
  with the constant $c = 4 C_{\Lambda}(L^{\mu}_f + L^{\mu}_g + 2\min(L^{\mu}_f,L^{\mu}_g))$, the following kernel approximation bounds holds uniformly:
  \begin{equation}
    \ts \big| \wh{\kappa}_{f,g}(\Vec{x},\Vec{y}) - \kappa_{f,g}(\Vec{x},\Vec{y}) \big| \leq \epsilon,\quad \forall \bs x, \bs y \in \Sigma,
  \end{equation}
  with probability exceeding $1-3 \exp( - \frac{m \epsilon^2}{64})$.
\end{proposition}
\begin{proof}
  With $\Sigma_{\eta}$ a finite optimal $\eta-$covering of $\Sigma$, any $\Vec{x}' \in \Sigma$ (resp. $\Vec{y}'$) can be written $\Vec{x}' = \Vec{x} + \Vec{r_x}$ (resp. $\Vec{y}' = \Vec{y} + \Vec{r_y}$) for centers $\Vec{x},\Vec{y} \in \Sigma_{\eta}$, and $\Vec{r_x},\Vec{r_y} \in \eta \bb B^d_2$. The proof proceeds by defining three events $\cl E_1$, $\cl E_2$, $\cl E_3$ from which Prop.~\ref{prop:kernelApproxUniform} follows, and then by bounding the failure probability of their joint occurrence. First, for any covering center $\Vec{x} \in \Sigma_{\eta}$, we can expect that the set of $m$ functions $h^f_{j} : \eta \bb B^d_2 \rightarrow \bb C$ defined for $j \in [m]$ as $ h^f_{j}(\Vec{r}; \Vec{x}) := f(\Vec{\omega}_j^\top\Vec{x} + \Vec{\omega}_j^\top\Vec{r} + \xi_j)$  contains, on average, few ``variations'' over the $\eta-$ball. More precisely, defining the largest variation of $h^f_{j}(\Vec{r}; \Vec{x})$ over the $\eta-$ball as
  \begin{equation*}
    H^f_{j}(\eta;\Vec{x}) := \sup_{\Vec{r} \in \eta \bb B^d_2} |h^f_{j}(\Vec{r}; \Vec{x}) - h^f_{j}(\Vec{0}; \Vec{x})| = \sup_{\Vec{r} \in \eta \bb B^d_2} |f(\Vec{\omega}_j^\top\Vec{x} + \Vec{\omega}_j^\top\Vec{r} + \xi_j) - f(\Vec{\omega}_j^\top\Vec{x} + \xi_j)|,
  \end{equation*}
  we first assume that, given $\epsilon_{1}>0$, the event $\cl E_1$ holds, with
  \begin{equation*}
    \ts \cl E_1 : \sup_{\Vec{x} \in \Sigma_{\eta}} \frac{1}{m}\sum_{j = 1}^m H^f_{j}(\eta;\Vec{x}) \leq L^{\mu}_f \eta C_{\Lambda} + \epsilon_{1}.
  \end{equation*}
  Similarly for $g$, we define the event $\cl E_2$ such that
  \begin{equation*}
    \ts \cl E_2 : \sup_{\Vec{y} \in \Sigma_{\eta}} \frac{1}{m}\sum_{j = 1}^m H^g_{j}(\eta;\Vec{y}) \leq L^{\mu}_g \eta C_{\Lambda} + \epsilon_{2}.
  \end{equation*}
  Next, we suppose that the kernel approximation has error bounded by $\epsilon_{3}$ for all the covering centers $\Vec{x},\Vec{y} \in \Sigma_{\eta}$, \ie
  \begin{equation*}
    \ts \cl E_3 : \sup_{\Vec{x},\Vec{y} \in \Sigma_{\eta}} \left| \wh{\kappa}_{f,g}(\Vec{x},\Vec{y}) - \kappa_{f,g}(\Vec{x},\Vec{y}) \right| \leq \epsilon_{3}.
  \end{equation*}
  Under those events, we establish a deterministic bound for all $\Vec{x}',\Vec{y}'$ using a chain of triangle inequalities:
  \begin{equation*}
    |\wh{\kappa}_{f,g}(\Vec{x}',\Vec{y}') - \kappa_{f,g}(\Vec{x}',\Vec{y}')| = |\wh{\kappa}_{f,g}(\Vec{x} + \Vec{r_x},\Vec{y} + \Vec{r_y}) - \kappa_{f,g}(\Vec{x} + \Vec{r_x},\Vec{y} + \Vec{r_y})| \leq \delta_1 + \delta_2 + \delta_3 + \delta_4,
  \end{equation*}
  where the error terms are defined as
  \begin{equation*}
    \begin{array}{r@{\ }l}
      \delta_1 &:= |\wh{\kappa}_{f,g}(\Vec{x} + \Vec{r_x},\Vec{y} + \Vec{r_y}) - \wh{\kappa}_{f,g}(\Vec{x},\Vec{y} + \Vec{r_y})|,\\
      \delta_2 &:= |\wh{\kappa}_{f,g}(\Vec{x},\Vec{y} + \Vec{r_y})-\wh{\kappa}_{f,g}(\Vec{x},\Vec{y})|,
    \end{array}
    \begin{array}{r@{\ }l}
      \delta_3 &:= |\wh{\kappa}_{f,g}(\Vec{x},\Vec{y}) - \kappa_{f,g}(\Vec{x},\Vec{y})|,\\
      \delta_4 &:= |\kappa_{f,g}(\Vec{x} ,\Vec{y} ) - \kappa_{f,g}(\Vec{x}+ \Vec{r_x},\Vec{y}+ \Vec{r_y})|.
    \end{array}
  \end{equation*}
  First, we observe that, thanks to $\cl E_1$: 
  \begin{equation*}
    \begin{split}
      \delta_1 &= \ts \frac{1}{m} \left|\sum_{j = 1}^m \left[ f(\Vec{\omega}_j^\top(\Vec{x} + \Vec{r_x}) + \xi_j ) - f(\Vec{\omega}_j^\top\Vec{x} + \xi_j )  \right] g^* (\Vec{\omega}_j^\top\Vec{y}' + \xi_j)  \right| \\
      &\leq \ts \|g\|_{\infty} \cdot \frac{1}{m} \sum_{j = 1}^m  \left| f(\Vec{\omega}_j^\top\Vec{x} + \Vec{\omega}_j^\top\Vec{r_x} + \xi_j) - f(\Vec{\omega}_j^\top\Vec{x} + \xi_j)  \right| \leq \eta L^{\mu}_f C_{\Lambda} + \epsilon_{1}.
    \end{split}
  \end{equation*}
  Similarly, using $\cl E_2$, we get for $\delta_2$:
  \begin{equation*}
    \delta_2 = \ts \frac{1}{m} \left| \sum_{j = 1}^m f(\Vec{\omega}_j^\top\Vec{x} + \xi_j) \left[ g^*(\Vec{\omega}_j^\top\Vec{y} + \Vec{\omega}_j^\top\Vec{r_y} + \xi_j) - g^*(\Vec{\omega}_j^\top\Vec{y} + \xi_j)  \right] \right| \leq \eta L^{\mu}_g C_{\Lambda} + \epsilon_{2}.
  \end{equation*}
  Regarding $\delta_3$, we directly get from $\cl E_3$:
  \begin{equation*}
    \delta_3 \leq \ts \sup_{\Vec{x},\Vec{y} \in \Sigma_{\eta}} | \wh{\kappa}_{f,g}(\Vec{x},\Vec{y}) - \kappa_{f,g}(\Vec{x},\Vec{y}) | \leq \epsilon_{3}.
  \end{equation*}
  Finally, for $\delta_4$, denoting $\Vec{u} := \Vec{x}-\Vec{y}$, $\Vec{r_u} := \Vec{r_x}-\Vec{r_y} \in 2\eta \bb B^d_2$ and using Lemma~\ref{lem:expected_kernel_with_convolution} (also recall that $h = f * \bar{g}$ is Lipschitz continuous with $L_h \leq \min(L^{\mu}_f, L^{\mu}_g)$ by Lemma~\ref{lem:mean_smooth_convolution}), as well as the definition of $C_{\Lambda}$ in~\eqref{eq:CLambda},
  \begin{equation*}
    \begin{split}
      \delta_4 &= \ts \left| \sik_{f,g}(\Vec{u})- \sik_{f,g}(\Vec{u}+\Vec{r_u}) \right| = \left|\expec{\Vec{\omega} \sim \Lambda} h(\Vec{\omega}^\top\Vec{u}) - h(\Vec{\omega}^\top\Vec{u}+\Vec{\omega}^\top\Vec{r_u})\right|\\
      &\leq \ts \expec{\Vec{\omega}} \left| h(\Vec{\omega}^\top\Vec{u}) - h(\Vec{\omega}^\top\Vec{u}+\Vec{\omega}^\top\Vec{r_u}) \right| \leq \expec{\Vec{\omega}} L_h \cdot \left| \Vec{\omega}^\top\Vec{r_u}  \right| \leq \ts L_h C_{\Lambda} \|  \Vec{r_u}  \|_2  \leq 2 \eta C_{\Lambda} \cdot \min(L^{\mu}_f, L^{\mu}_g).
    \end{split}
  \end{equation*}
  
  \noindent Putting everything back together, under $\cl E_1$, $\cl E_2$, and $\cl E_3$, for any $\Vec{x}', \Vec{y}' \in \Sigma$:
  \begin{equation}
    \label{eq:deterministicGeneral}
    \left|\wh{\kappa}_{f,g}(\Vec{x}',\Vec{y}') - \kappa_{f,g}(\Vec{x}',\Vec{y}') \right| \leq \epsilon_{1} + \epsilon_{2} + \epsilon_{3} + \eta C_{\Lambda} \left(L^{\mu}_f + L^{\mu}_g + 2 \cdot \min(L^{\mu}_f, L^{\mu}_g)\right).
  \end{equation}
  
  \noindent It remains to bound the failure probability for each event. For event $\cl E_1$, we have to bound the probability
  \begin{equation*}
    \ts \bb P \left[ \bar{\cl E_1} \right] = \bb P \left[ \exists{\Vec{x} \in \Sigma_{\eta}} \text{ s.t. } \frac{1}{m}\sum_{j = 1}^m H^f_{j}(\eta;\Vec{x}) \geq L^{\mu}_f \eta C_{\Lambda} + \epsilon_{1} \right].
  \end{equation*}
  We first focus on one single center $\Vec{x} \in \Sigma_{\eta}$. Each associated $H^f_{j}(\eta;\Vec{x}) = \sup_{\Vec{r} \in \eta \bb B^d_2} |f(\Vec{\omega}_j^\top\Vec{x} + \Vec{\omega}_j^\top\Vec{r} + \xi_j) - f(\Vec{\omega}_j^\top\Vec{x} + \xi_j)|$ is a random variable identically and independently distributed (where the randomness is due to the draw of $\Vec{\omega}_j$ and $\xi_j$). The expectation $\expec{}H^f_{j}$ of those variables is bounded by (we use the mean Lipschitz smoothness~\eqref{eq:meanSmooth} with $t = \xi_j$, $r = \Vec{\omega}_j^\top\Vec{r}$ and $\delta = |\Vec{\omega}_j^\top\Vec{r}|$):
  $$ \expec{}H^f_{j} = \ts \expec{\Vec{\omega}_j} \expec{\xi_j} \sup_{\Vec{r} \in \eta \bb B^d_2} |f(\Vec{\omega}_j^\top\Vec{x} + \Vec{\omega}_j^\top\Vec{r} + \xi_j) - f(\Vec{\omega}_j^\top\Vec{x} + \xi_j)| \leq L^{\mu}_f \expec{\Vec{\omega}_j} |\Vec{\omega}_j^\top\Vec{r}| = L^{\mu}_f \eta C_{\Lambda}.$$
  
  \noindent Now we describe how the sum $\frac{1}{m}\sum_{j = 1}^m H^f_{j}(\eta;\Vec{x})$ concentrates around its mean $\expec{}H^f_{j}$ with Hoeffding's inequality (note that $0 \leq H^f_{j}(\eta;\Vec{x}) \leq 2\|f\|_{\infty} \leq 2$), and use $\expec{}H^f_{j} \leq L^{\mu}_f \eta C_{\Lambda}$ to get a probabilistic bound
  $$\ts \bb P \left[ \frac{1}{m} \sum_j H^f_{j}(\eta;\Vec{x}) \geq L^{\mu}_f \eta C_{\Lambda} + \epsilon_{1} \right] \leq \bb P \left[ \frac{1}{m} \sum_j H^f_{j}(\eta;\Vec{x}) - \expec{}H^f_{j} \geq \epsilon_{1} \right] \leq \exp\left(-\frac{m\epsilon_{1}^2}{2}\right).$$
  
  \noindent We take a union bound of this result over the $\cl C_{\eta}(\Sigma) = |\Sigma_{\eta}|$ centers $\Vec{x} \in \Sigma_{\eta}$ to obtain
  \begin{equation*}
    \bb P \left[ \bar{\cl E_1} \right] \leq \cl C_{\eta}(\Sigma) \exp\left(-m\epsilon_{1}^2/2\right) = \exp\left( \cl H_{\eta}(\Sigma) -m\epsilon_{1}^2/2\right).
  \end{equation*}
  Moreover, if $m \geq 4 \cl H_{\eta}(\Sigma) \epsilon_{1}^{-2}$, then we get $\mathbb{P}\left[ \bar{\cl E_1} \right] \leq e^{- m \epsilon_{1}^2/4 }$. An identical development for $\cl E_2$ yields $\mathbb{P}\left[ \bar{\cl E_2} \right] \leq e^{- m \epsilon_{2}^2/4 }$ if $m \geq 4 \cl H_{\eta}(\Sigma) \epsilon_{2}^{-2}$. For $\cl E_3$, an union bound of Prop.~\ref{prop:kernelApproxFixedPair} on all pairs in $\Sigma_{\eta} \times \Sigma_{\eta}$ gives
  \begin{equation*}
    \ts \mathbb{P}\left[ \bar{\cl E_3} \right]  \leq 2\binom{\cl C_{\eta}(\Sigma)}{2}  e^{ - m \epsilon_{3}^2/2} \leq \cl C_{\eta}(\Sigma)^2 e^{ - m \epsilon_{3}^2/2} = e^{ 2 \cl H_{\eta}(\Sigma)  - m \epsilon_{3}^2/2}.
  \end{equation*}
  Moreover, if $m \geq 8 \cl H_{\eta}(\Sigma) \epsilon_{3}^{-2}$, we get $\mathbb{P}\left[ \bar{\cl E_3} \right] \leq e^{- m \epsilon_{3}^2/4}$. By union bound, and provided
  $$m \geq 4 \cl H_{\eta}(\Sigma) \cdot \max\left(\epsilon^{-2}_{1}, \epsilon^{-2}_2,2\epsilon^{-2}_3\right),$$
  the probability of failure of the deterministic bound above is lower than 
  \begin{equation*}
    \bb P \left[ \bar{\cl E_1} \cup \bar{\cl E_2} \cup \bar{\cl E_3} \right] \leq \bb P \left[ \bar{\cl E_1} \right] + \bb P \left[ \bar{\cl E_2} \right] + \bb P \left[ \bar{\cl E_3} \right] = e^{- m \epsilon_{1}^2/4 } + e^{- m \epsilon_{2}^2/4} + e^{-m \epsilon_{3}^2/4} .
  \end{equation*}
  
  \noindent Finally, the desired result (less generic but more meaningful) is found by imposing equal contributions $\epsilon/4$ by each error term in~\eqref{eq:deterministicGeneral}, \ie $\epsilon_{1} = \epsilon_{2} = \epsilon_{3} = \epsilon/4$ and $\eta = \epsilon/(4 C_{\Lambda}[L^{\mu}_f + L^{\mu}_g + 2 \min(L^{\mu}_f, L^{\mu}_g)]) $.
\end{proof}

Prop.~\ref{prop:kernelApproxUniform} shows that we can control (\eg by increasing $m$) the kernel approximation error uniformly, provided we control the smoothness of the ``initial'' kernel $\sik = \cl F^{-1}\Lambda$ (through $C_\Lambda$) and the mean smoothness of the maps $f$ and $g$. Improvements are possible, for example, by more carefully setting the values of $\{\epsilon_1, \epsilon_2, \epsilon_3\}$  and $\eta$. If $\Sigma$ is \emph{structured} (\eg if it consists of sparse vectors or low-rank matrices) and $\Lambda$ is Gaussian, the value of $C_\Lambda$ in $\expec{\Vec{\omega}}|\Vec{\omega}^\top\Vec{r}| \leq \eta C_{\Lambda}$ for $\bs r \in (\Sigma-\Sigma) \cap \eta \bb B_2^d$ (which controls the bounds on $\delta_1$, $\delta_2$ and $\delta_4$) can be related to the Gaussian mean width of $\Sigma-\Sigma$ \cite{chandrasekaran2012convex}.

  \begin{example}
    Consider once again the case of a Gaussian kernel with unit bandwidth (\ie $C_\Lambda = 1$), with a signal space $\Sigma$ made of bounded signals (inside the unit Euclidean ball $\bb B_2^d$) lying in a union of $S$ subspaces of $\bb R^d$ with dimension $s$. In this case, according to the entropy of this signal model (see Ex.~\ref{ex:UoS-example}), the kernel approximation error $| \wh{\kappa}_{f,g}(\Vec{x},\Vec{y}) - \kappa_{f,g}(\Vec{x},\Vec{y}) |$ is uniformly bounded over $\Sigma$, with high probability, provided that the number of features satisfies $ m \geq  C\, \epsilon^{-2} [s \log\big(  \frac{4}{\epsilon}  (L^{\mu}_f + L^{\mu}_g + 2\min(L^{\mu}_f,L^{\mu}_g)) \big) + \log S]$. For instance, for bounded $s$-sparse signals, we need $m \geq  C\, s \cdot \epsilon^{-2} \log\big( \frac{4 e d}{c s \epsilon} (L^{\mu}_f + L^{\mu}_g + 2\min(L^{\mu}_f,L^{\mu}_g)) \big)$.
\end{example}

We conclude this section by showing that Prop.~\ref{prop:kernelApproxUniform} allows characterizing the proximity of two approximated kernels $\wh{\kappa}_{f,g}$ and $\wh{\kappa}_{f',g}$ (given three functions $f,f',g \in \pf$) when they are related by identical expectations $\bb E \wh{\kappa}_{f,g} = \bb E \wh{\kappa}_{f',g} = \kappa_0$. While this result could be achieved by a simple use of the triangular inequality---from $|\wh{\kappa}_{f,g} - \wh{\kappa}_{f',g}|  \leq |\wh{\kappa}_{f,g} - \kappa_0| + |\wh{\kappa}_{f,g} - \kappa_0|$ and using the same proposition to bound the last two terms---the following corollary provides a more direct bound, possibly tighter. 

\begin{corollary}[Proximity of approximated RPF kernels]
  \label{cor:proximity-approx-kernels}
Given $\epsilon > 0$, a compact set $\Sigma$, two $2\pi$-periodic functions $f,f'$ such that their difference $f-f' \in \pf$, as well as a third periodic function $g \in \pf$, such that there exist finite mean smoothness constants $L^{\mu}_{f-f'}$ and $L^{\mu}_g$, and $\Lambda$ such that $C_{\Lambda} < \infty$, if $\kappa_{f,g}(\cdot,\cdot)=\kappa_{f',g}(\cdot,\cdot)$ and if the feature dimension is larger than
  \begin{equation}
    \label{eq:unif-kern-approx-sample-complex-cor}
    \ts m \geq 128 \cdot \frac{1}{\epsilon^2} \cdot \cl H_{\epsilon/c}(\Sigma),
  \end{equation}
  with constant $c = 4 C_{\Lambda}(L^{\mu}_{f-f'} + L^{\mu}_g + 2\min(L^{\mu}_{f-f'}, L^{\mu}_g))$, then 
  \begin{equation}
    \ts \big| \wh{\kappa}_{f,g}(\Vec{x},\Vec{y}) - \wh{\kappa}_{f',g}(\Vec{x},\Vec{y}) \big| \leq \epsilon, \quad \forall \Vec{x},\Vec{y} \in \Sigma,
  \end{equation}
  with probability exceeding $1 - 3 \exp( - \frac{m \epsilon^2}{64})$.
\end{corollary}

\begin{proof}
We simply observe that, by linearity of the kernels with respect to their supporting functions, for any $\bs x, \bs y \in \Sigma$, $\wh{\kappa}_{f,g}(\bs x, \bs y) - \wh{\kappa}_{f',g}(\bs x, \bs y) = \wh{\kappa}_{\tilde f,g}(\bs x, \bs y)$ with $\tilde f := f - f'$. The proof then follows by applying Prop.~\ref{prop:kernelApproxUniform} to the RPFs supported by $\tilde f, g \in \pf$, with the vanishing kernel $\bb E\, \wh{\kappa}_{\tilde f,g}(\bs x, \bs y) = \kappa_{f,g}(\bs x, \bs y) - \kappa_{f',g}(\bs x, \bs y) = 0$.    
\end{proof}

In Sec.~\ref{sec:appl-semi-quant}, we will use this corollary in combination with Prop.~\ref{prop:kernelApproxUniform} to compare the performance of a machine learning algorithm (the kernel support vector machine, SVM) on a given classification task when learning and inference are using identical approximated kernels (\ie when the learning is performed using the RFF) or only kernels that are asymptotically equal (when the learning stage uses the expected kernel).

\section{Asymmetric geometry-preserving embedding}
\label{sec:correctPTB}

Our approach can be related to the context of geometry-preserving embedding (or coding) developed in~\cite{boufounos2017representation}. This allows us to provide another version of one of their central results, \cite[Thm~3.2]{boufounos2017representation}, whose proof is incorrect (as described below). While the alternative result we propose looks slightly different, it fulfills the same goal: a non-asymptotic guarantee for the geometry-preserving capabilities of the embedding~\eqref{eq:def_periodic_features} with discontinuous $f$, which holds on infinite signal sets. This section can be seen as a first (theoretical) application of Prop.~\ref{prop:kernelApproxUniform}.

\subsection{Geometry-preserving embedding: the initial approach}
\label{sec:earl-geom-pres}

In~\cite{boufounos2017representation} the authors study when a mapping $\bs \varphi: \Sigma \to \bb C^m$ (such as $\Vec{z}_f$ defined in~\eqref{eq:def_periodic_features} for $f \in \pf$) defines an embedding of $\Sigma$ into $\bb C^m$ approximately preserving the proximity of vectors in $\Sigma$. This proximity is measured by the (local) preservation a distance associated with a $\ell_\sharp$-norm ${\|\cdot\|_\sharp}$ (\eg the $\ell_1$ or the $\ell_2$-norm). Adapting their setting to our conventions\footnote{Hereafter, departing from the general approach of~\cite{boufounos2017representation}, we always consider the simplified case where $\bb C^m$ is equipped with the Euclidean distance, with a squaring of the corresponding distance in \eqref{eq:gen-embed}.}, given some $\epsilon, \delta > 0$, and an invertible function, or \emph{distance map}, $\gamma: \bb R_+ \to \bb R_+$, they study the conditions ensuring that $\bs \varphi$ is a $(\gamma,\delta,\epsilon)$-embedding of $\Sigma$ (endowed with the $\ell_\sharp$-norm) into $\bb C^m$; or mathematically, such that $\bs \varphi$ respects
\begin{equation}
  \label{eq:gen-embed}
  \ts (1-\delta)\, \gamma( \|\Vec{x} - \Vec{y}\|_\sharp) - \epsilon \leq \|\bs \varphi(\Vec{x}) - \bs \varphi(\Vec{y}) \|^2 \leq (1+\delta)\, \gamma( \|\Vec{x} - \Vec{y}\|_\sharp) + \epsilon,  
\end{equation}
for all $\bs x, \bs y \in \Sigma$.

In \eqref{eq:gen-embed}, $\gamma$ maps distances in $\Sigma$ to (squared) distances in $\bb C^m$, and $\delta$ and $\epsilon$ quantify the multiplicative and the additive error, respectively, of the embedding associated with the map $\gamma$. For instance, if $\bs \varphi$ is linear with $\bs \varphi(\bs x) = \bs A \bs x$, there exists many random constructions of the $m \times d$ matrix $\bs A$ with appropriate scaling (\eg random Gaussian matrix or random partial Fourier matrix \cite{foucart2017mathematical}) for which~\eqref{eq:gen-embed} holds with high probability with $\ell_\sharp \equiv \ell_2$, $\epsilon=0$, and $\gamma(t)=t^2$ for $\Sigma = \Sigma_k$ and $m = O(\delta^{-2} k \log (n/k))$. Similarly, in the context of one-bit compressive sensing where $\bs \varphi(\bs x) = (cm)^{-1/2}\, \sign(\bs A \bs x)$ (for some suitable $c>0$), $\|\bs \varphi(\Vec{x}) - \bs \varphi(\Vec{y}) \|^2 $ represents the (scaled) Hamming distance between the two binary vectors $\bs \varphi(\bs x)$ and $\bs \varphi(\bs y)$, and \eqref{eq:gen-embed} is verified with high probability over $\Sigma_k \cap \bb B^d$ with $m=O(\epsilon^{-2} k \log (n/k))$, $\ell_\sharp \equiv \ell_2$, $\delta = 0$, and $\gamma(t)=t$ \cite{Jacques_2013}. The work \cite{boufounos2017representation} extends this analysis to general nonlinear feature maps $\bs \varphi(\cdot) = \bs z_f(\cdot)$ for some periodic function $f$ (such as the universal quantizer $q$). In such a context, the authors show that \eqref{eq:gen-embed} holds with a map $\gamma$ that often displays two regimes: a linear regime for small distances in $\Sigma$ ($\bs x \approx \bs y$), and a saturation regime where $\gamma$ quickly flattens after a certain distance. 

As explained in~\cite[Sec.~4.5]{boufounos2017representation}, this approach is connected to the approximation of a kernel $\kappa: \Sigma \times \Sigma \to \bb R_+$ from the inner product of the images of two vectors, namely, for which $\langle \Vec{z}_f(\Vec{x}), \Vec{z}_f(\Vec{y}) \rangle \approx \kappa(\Vec{x},\Vec{y})$. Assuming $\|\Vec{z}_f\| = 1$ for simplicity, which is the case for complex exponential and universal quantization features, we find  $\ts \|\Vec{z}_f(\Vec{x}) - \Vec{z}_f(\Vec{y}) \|^2 =  2\, (1- \langle \Vec{z}_f(\Vec{x}),\Vec{z}_f(\Vec{y}) \rangle)$. Therefore, if $\bs z_f$ is a $(\gamma,0,2\epsilon)$-embedding of $\Sigma$ into $\bb C^m$, then
\begin{equation}
  \label{eq:equiv-kernel-geom-preserv}
  \ts \kappa(\Vec{x},\Vec{y}) - \epsilon \leq \langle \Vec{z}_f(\Vec{x}), \Vec{z}_f(\Vec{y}) \rangle \leq \kappa(\Vec{x},\Vec{y}) + \epsilon,  
\end{equation}
for all $\bs x, \bs y \in \Sigma$, provided we define the kernel 
\begin{equation}
  \label{eq:equiv-geom-embed-kern}
  \ts \kappa(\Vec{x},\Vec{y}) := 1 - \frac{1}{2} \gamma (\|\Vec{x} - \Vec{y}\|_\sharp).
\end{equation}
The smoothness of $\kappa$ is thus directly connected to the one of $\gamma$; for instance, if $\gamma$ is Lipschitz continuous with constant $L_\gamma \geq 0$ of $\bb R_+$, then, from the invertibility of $\gamma$ over $\bb R^+$, $\kappa$ is Lipschitz continuous with constant $L_\gamma/2$ with respect to any of its argument. Note that, from Lemma~\ref{lem:expected_kernel_with_convolution} and Lemma~\ref{lem:mean_smooth_convolution}, we also know that if $f$ is mean smooth with constant $L^\mu_f$, then $\kappa(\bs x, \bs y) = \sik_{f,f}(\bs x - \bs y)$ is Lipschitz continuous with constant $L_\kappa \leq C_\Lambda L^\mu_f$ with respect to any of its argument (as proved from the bound on $\delta_4$ in the proof of Prop.~\ref{prop:kernelApproxUniform}). This shows that, despite their different origin, the smoothness of $\gamma$ (in the approach~\cite{boufounos2017representation}) and the one of $f$ (in ours) control the one of $\kappa$.

Compared to our approach,~\cite {boufounos2017representation} imposes the periodic function $f$ to be ``Lipschitz continuous by part'' (rather than being mean smooth), as defined hereafter in a setting adapted to our needs.
\begin{definition}[{$T$-part Lipschitz continuity~\cite[Def.~2.1]{boufounos2017representation}}]
  \label{def:T-part-cont}  
  A function $f : \Sigma \rightarrow \bb C$ is $T$-part Lipschitz continuous over $\cl S \subset \Sigma$ with constant $\bar L_f \geq 0$, if there exists a finite partition $\{\cl S_t\}_{t=1}^T$ of $\cl S$ into $T$ disjoint sets (\ie $\bigcup_{t = 1}^T \cl S_t = \cl S$) such that  
  \begin{equation}
    \label{eq:gen-Lip}
    \ts \forall t \in [T],\,\forall \bs x, \bs y \in \cl S_t, \quad |f(\Vec{x}) - f(\Vec{y})| \leq \bar L_f \cdot \|\Vec{x}-\Vec{y}\|_\sharp.
  \end{equation}
  Moreover, $f$ is {\em exactly} $T$-part Lipschitz continuous over $\cl S$ with constant $\bar L_f \geq 0$, which we write $f \in \overline{\rm Lip}(\cl S, T, \bar L_f)$, if it is $T$-part Lipschitz continuous with that constant but is not $(T-1)$-part Lipschitz continuous with the same constant.
\end{definition}
\noindent Note that~\eqref{eq:gen-Lip} both generalizes~\eqref{eq:Lipschitz} to functions from $\Sigma \subset \bb R^d$ to $\bb C$, and localizes~\eqref{eq:Lipschitz} on $\cl S$.    

Following the convention of our paper, the authors of~\cite{boufounos2017representation} then prove the following result. We simplify it to the case $\delta = 0$ and where each component of $\bs z_f$ has at most $T$ parts of continuity (despite its randomness). 

\begin{theorem}[{Adapted from~\cite[Thm~3.2]{boufounos2017representation}}]
  Given $0<\epsilon <1$, an $L_\gamma$-Lipschitz continuous distance map $\gamma: \bb R_+ \to \bb R_+$, and a signal set $\Sigma$ with finite covering number $\cl C_\eta(\Sigma)$ for any radius $\eta >0$, let us assume that, for any \emph{fixed} pair of vectors $\bs x, \bs y \in \Sigma$, the mapping $\bs z_f$ defined in~\eqref{eq:def_periodic_features} satisfies the embedding relation~\eqref{eq:gen-embed} (for $\delta = 0$) with probability exceeding $1 - C e^{-c m  \epsilon^2}$.
  
  Let us suppose that there exists a constant $\bar L_f\geq 0$ such that, for any $\bs x \in \Sigma$, integer $t \geq 1$, radius $\eta > 0$, and given $\cl S_{\bs x}(\eta) := \{\bs u \in \Sigma: \|\bs u - \bs x\|_\sharp \leq \eta\}$ (a neighborhood of $\bs x$ of radius $\eta$),
  \begin{equation}
    \label{eq:bound-proba-T-part}
    \ts \bb P\big[ (\bs z_f(\cdot))_k \in \overline{\rm Lip}(\cl S_{\bs x}(\eta), T, \bar L_f)\big]\ \leq\ p_t(\eta),  
  \end{equation}
  with $p_t$ independent of $\bs x$, and $p_t(\eta) = 0$ if $t > T$ for some integer $T \geq 0$. 

  In this context, defining $\rho(\eta,T) := \sum_{t=2}^{T} p_t(\eta) \log t$ and $\nu := \frac{1}{4} (L_\gamma+ \bar L_f)^{-1}$, provided that $\epsilon^2 \geq C \rho(2\nu \epsilon^2,T)$ and 
  \begin{equation}
    \label{eq:geompres-sample-complexity}
    \ts m \geq C \epsilon^{-2} \big(\cl H_{\nu \epsilon^2}(\Sigma) + \log T\big),  
  \end{equation}
  the mapping $\bs z_f$ is a $(\gamma, 0, 2\epsilon)$-embedding with probability exceeding $1-C e^{- c \epsilon^2 m}$, \ie
    $\bs z_f$ respects
    $$
\ts \gamma( \|\Vec{x} - \Vec{y}\|_\sharp) - 2\epsilon \leq \|\bs z_f(\Vec{x}) - \bs z_f(\Vec{y}) \|^2 \leq \gamma( \|\Vec{x} - \Vec{y}\|_\sharp) + 2\epsilon, \quad \forall \bs x, \bs y \in \Sigma.
$$
  \label{thm:emb_discontinuous}
\end{theorem}
\noindent The statement of this theorem is an easy adaptation of~\cite[Thm~3.2]{boufounos2017representation} where we set $\delta = 0$, $w(\epsilon, \delta) = c \epsilon^2$, $c_0 = c\epsilon$, $T_{\max} = T$ (so that $P_F = 0$), and $\alpha=\epsilon^2 \leq \epsilon \leq 1$.  
\medskip

Note that the first assumption of this theorem (regarding the fact that~\eqref{eq:gen-embed} holds with high probability for any fixed pair of signals) is proven in a separate result, namely~\cite[Thm~4.1]{boufounos2017representation}. This theorem is similar to our Prop.~\ref{prop:kernelApproxFixedPair} (up to an easy extension of this proposition to a finite set of pairs by union bound). Since the flaw developed below is independent of that separate result, we abstract this specific assumption away in this work.
As explained in~\cite[App. E]{boufounos2017representation}, the conditions of this theorem can thus be met for instance in the case where $f$ is the universal quantizer. One can then show that $T=2$, and defining $p_2(\eta) := D \eta$, with $D > 0$ function of $d$ and $\Lambda$, is appropriate for the bound~\eqref{eq:bound-proba-T-part}. Therefore, $\rho( 2 \nu \epsilon^2,T) \leq 2 D \nu \epsilon^2 \leq \epsilon^2/C$ for an appropriate $C > 0$.   

The statement of this theorem bears similarities with our Prop.~\ref{prop:kernelApproxUniform} in the case where $f=g$; in essence, keeping in mind the equivalence~\eqref{eq:equiv-kernel-geom-preserv}, up to a smaller covering radius scaling as $\epsilon^2<\epsilon < 1$ in~\eqref{eq:geompres-sample-complexity}, the constraint~\eqref{eq:geompres-sample-complexity} is similar to~\eqref{eq:unif-kern-approx-sample-complex} if we consider that the $T$-part Lipschitz continuity of $f$ replaces its mean smoothness.

However, the proof of Theorem~\ref{thm:emb_discontinuous} in~\cite[App. B]{boufounos2017representation} is incorrect. Let us see why by sketching their arguments in our system of notations and using $\ell_\sharp = \ell_2$ for the sake of simplicity. Given $\eta>0$, $\bs x \in \Sigma$, and $t \in [T]$, the authors first (implicitly) note that if the random variable $Z(\bs x)$ counts the number of components of $\bs z_f(\cdot)$ that are exactly $t$-part Lipschitz over $\cl S_{\bs x}(\frac{\eta}{2})$ with a given constant $\bar L_f$, then $\bb E Z \leq m p_t$. Therefore, given $c_0 > 0$ and invoking Hoeffding's inequality, they can upper bound the probability that $Z(\bs x) \geq m p_t (1 + c_0) \geq \bb E Z(\bs x) +  m p_t c_0$ with
$$
\bb P[ Z(\bs x) \geq m p_t (1 + c_0)] \leq \exp(-2 c_0^2 m). 
$$

From this bound (using a union bound over all $t\in [T]$), they then determine that $\cl S_{\bs x}(\frac{\eta}{2})$ is partitioned in at most $S := \exp( (1+ c_0)\,\rho(\eta, T) m)$ cells with probability greater than $1 - T e^{-2 c_0^2 m}$. Each cell of this partition of $\cl S_{\bs x}(\frac{\eta}{2})$ has thus a diameter of at most $\eta$. Moreover, by definition, $\bs z_f$ is guaranteed to be Lipschitz continuous with constant $\bar L_f$ over every such cell. 

The author then consider the possibility to pick one point per such cell, called cell center, and to gather them in a finite set of at most $S$ elements. One can then repeat this construction for all vectors $\bs x$ of a $\frac{\eta}{2}$-covering $\Sigma_{\eta/2}$ of $\Sigma$, and collect, for each such vector, all cell centers of its related neighborhood into a global set $\cl G$ of centers of at most $S \times \cl C_{\frac{\eta}{2}}(\Sigma)$ elements. By definition, $\cl G$ is thus a $\eta/2$-covering of $\Sigma$ with the additional property that $\bs z_f$ is $\bar L_f$-Lipschitz continuous over each cell.  

The authors then leverage this local continuity as follows. Since, by hypothesis, the mapping $\bs z_f$ defined in~\eqref{eq:def_periodic_features} satisfies the embedding relation~\eqref{eq:gen-embed} (for $\delta = 0$) with probability exceeding $1 - C e^{-c m  \epsilon^2}$ over any \emph{fixed} pair of vectors $\bs x, \bs y \in \Sigma$, they first expand this property over all pairs of vectors taken in $\cl G \times \cl G$. This is ensured with probability exceeding $1 - C S^2 \cl C^2_{\eta/2}(\Sigma) e^{-c m  \epsilon^2}$, by union bound, since $|\cl G \times \cl G| \leq S^2 \cl C^2_{\eta/2}(\Sigma)$. Next, they extend this property to all $\bs x, \bs y \in \Sigma$ by continuity, exploiting the (local) Lipschitz continuity of $f$ over each cell.   

The flaw, which happens in the first step above, is analogous to how we cannot show the wrong statement $\bb P[ \|\bs g\|^2 < 0] = 1/2$ for a Gaussian vector $\bs g \sim \cl N^d(0,1)$ by assigning another vector $\bs u$ to $\bs g$ in the correct equality $\bb P[ \scp{\bs u}{\bs g} < 0] = 1/2$, valid for $\bs u$ fixed. Indeed, the vectors of $\cl G$, the collection of all cell centers, are built from the random mapping $\bs z_f$ --- each center must be taken in a cell whose frontiers are controlled by the discontinuities of the components of $\bs z_f$. These vectors are thus dependent of both $\bs\Omega \sim \Lambda^m$ and the dither $\bs \xi \sim \cl U^m([0,2\pi))$, through their dependence in $\bs z_f$. Therefore, one cannot ensure that the probability that~\eqref{eq:gen-embed} holds (with $\delta = 0$) on two cell centers $\bs x = \bs x(\bs \Omega, \bs \xi), \bs y = \bs y(\bs \Omega, \bs \xi)$ exceeds $1 - C e^{-c m  \epsilon^2}$, since that probability is itself taken over $\bs \Omega, \bs \xi$. This flaw breaks the proof of~\cite[Thm~3.2]{boufounos2017representation}.

\medskip

\subsection{An alternative geometry-preserving embedding}
\label{sec:altern-geom-pres}

One can use Prop.~\ref{prop:kernelApproxUniform} to get a variant of Thm~\ref{thm:emb_discontinuous} relying on the equivalence between~\eqref{eq:gen-embed} and~\eqref{eq:equiv-kernel-geom-preserv}.
This variant achieves the same high-level goal (\ie a non-asymptotic guarantee on the approximation error achieved by the embedding $\bs z_f$ that holds infinite signal sets even for discontinuous $f$), but the assumptions it relies on differ in two aspects. First, the smoothness of the distance map $\gamma$ is not anymore characterized by its Lipschitz smoothness directly, but by the constant $C_{\Lambda}$, defined by the sampling scheme $\Lambda$ driving the random projections $\bs \Omega$. Second, we use the mean Lipschitz property instead of the $T$-part Lipschitz property as notion of ``generalized smoothness'' for the map $f$. It is not clear if this change is fundamentally necessary to be able to prove a variant of Thm~\ref{thm:emb_discontinuous}, but we leave an investigation of this issue for future work (the universal quantization $q$ satisfies both properties anyway).

In fact, the following corollary shows that one can define novel \emph{asymmetric embeddings} from $\Sigma$ into $\bb C^m$; we can map two vectors of $\Sigma$ with different random feature mappings $\bs z_f$ and $\bs z_g$ achieved with distinct periodic functions $f$ and $g$, respectively, and still show that, under certain conditions on $f$, $g$, and the frequency distribution $\Lambda$, $\|\bs z_f(\bs x) - \bs z_g(\bs y)\|^2$ approximates a distortion of the distance between any $\bs x, \bs y \in \Sigma$ provided $m$ is large compared to the complexity of $\Sigma$. Then, setting $f = g$ provides a specific embedding of $\Sigma$ into $\bb C^m$, in the sense described by~\cite{boufounos2017representation}.

\begin{corollary}[Asymmetric geometry-preserving embedding]
  \label{cor:geom-pres-our-approach}

  Let $\Sigma$ be a compact set with finite covering number,  $f,g \in \pf$ be two real $2\pi$-periodic functions and finite mean smoothness constants $L^{\mu}_f > 0$ and $L^{\mu}_g > 0$, respectively. We assume that the frequency distribution $\Lambda$ is such that $C_{\Lambda} < \infty$, and there exists a real, one-dimensional \pd kernel $\sik_0:\bb R_+ \to [0, 1]$ such that $[\cl F^{-1} \Lambda](\bs u) = \sik_0(\|\bs u\|_\sharp)$ for some norm~${\|\cdot\|_\sharp}$. 

  \noindent For all error level $\epsilon > 0$, provided the feature dimension is larger than
  \begin{equation}
    \label{eq:unif-kern-approx-sample-complex-cor2}
    \ts m \geq 128 \cdot \frac{1}{\epsilon^2} \cdot \cl H_{\epsilon/c}(\Sigma),
  \end{equation}
  with constant $c = 4 C_{\Lambda}(L^{\mu}_f + L^{\mu}_g + 2\min(L^{\mu}_f, L^{\mu}_g))$, we have, with probability exceeding $1- 9 \exp(- \frac{m\epsilon^2}{64})$, 
  \begin{equation}
    \label{eq: geom-pres-our-approach}
    \gamma_{f,g}(\|\bs x - \bs y\|_\sharp) - 4 \epsilon \leq \|\bs z_f(\bs x) - \bs z_g(\bs y)\|^2 \leq \gamma_{f,g}(\|\bs x - \bs y\|_\sharp) + 4 \epsilon,\ \forall \bs x, \bs y \in \Sigma,
  \end{equation}
  according to the distance map $\gamma_{f,g}$ defined by
  $$
  \ts \gamma_{f,g}: s \in \bb R_+ \to \gamma_{f,g}(s) = \|f\|^2 + \|g\|^2 - 2\sum_{k \in \Zbb} \, F_kG_k^* \, \sik_0(|k| s) \in \bb R_+.
  $$
  Therefore, if $f=g$ and if $\gamma_{f,f}$ is invertible, $\bs z_f$ is a $(\gamma_{f,f}, 0, 4\epsilon)$-embedding of $\Sigma$ (equipped with the norm ${\|\cdot\|_\sharp}$) into $\bb C^m$, with $\gamma_{f,f}(0) = 0$ and $\gamma_{f,f}(s) \in [0, 2]$. 
\end{corollary}
\begin{proof}
  Under the hypothesis of this corollary and remembering that $\wh{\kappa}_{f,g}(\Vec{x},\Vec{y}) = \langle \Vec{z}_{f}(\Vec{x}), \Vec{z}_{g}(\Vec{y}) \rangle$, Prop.~\ref{prop:kernelApproxUniform} tells us that the event 
  \begin{equation*}
    \ts |\wh{\kappa}_{f,g}(\Vec{x},\Vec{y}) - \kappa_{f,g}(\Vec{x},\Vec{y})| \leq \epsilon,\ \forall\, \Vec{x},\Vec{y} \in \Sigma,
  \end{equation*}
  holds with probability exceeding  $1 - 3 \exp( - \frac{m \epsilon^2}{64})$. Similarly, with the same probability,
  $|\wh{\kappa}_{f,f}(\Vec{x},\Vec{y}) - \kappa_{f,f}(\Vec{x},\Vec{y})| \leq \epsilon$ and $|\wh{\kappa}_{g,g}(\Vec{x},\Vec{y}) - \kappa_{g,g}(\Vec{x},\Vec{y})| \leq \epsilon$, for all $\bs x, \bs y \in \Sigma$. Therefore, by union bound, these three events jointly hold with probability larger than $1 - 9 \exp( - \frac{m \epsilon^2}{64})$. 

  Conditionally to this combined occurrence, since $\kappa_{f,f}(\bs x,\bs x) = \sik_{f,f}(\bs 0) = \|f\|^2$ and $\kappa_{g,g}(\bs y,\bs y) = \sik_{g,g}(\bs 0) = \|g\|^2$, we find
  $$
  \big|\scp{\bs z_f(\bs x)}{\bs z_f(\bs x)} - \|f\|^2 \big| \leq  \epsilon,\quad \big|\scp{\bs z_g(\bs y)}{\bs z_g(\bs y)} - \|g\|^2 \big| \leq  \epsilon, 
  $$
  and 
  $$
  \|\bs z_f(\bs x) - \bs z_g(\bs y)\|^2 \leq \|f\|^2 + \|g\|^2 - 2 \kappa_{f,g}(\Vec{x},\Vec{y}) + 4 \epsilon = \|f\|^2 + \|g\|^2 - 2\sik_{f,g}(\bs x - \bs y) + 4 \epsilon.
  $$
  Moreover, from Prop.~\ref{prop:expected_kernel}, since $\sik_{f,g}(\bs u) = \sum_{k \in \Zbb} \, F_kG_k^* \, \sik(k \bs u)$
  with  $\sik(\bs u) = (\cl F^{-1} \Lambda)(\Vec{u}) = \sik_0(\|\bs u\|_\sharp)$ and $\bs u \in \bb R^d$, the definition of $\gamma_{f,g}$ provides
  $$
  \gamma_{f,g}(\|\bs u\|_\sharp) = \|f\|^2 + \|g\|^2  - 2 \sik_{f,g}(\bs u),
  $$
  which proves the upper bound of~\eqref{eq: geom-pres-our-approach}, the lower bound being established similarly.

  Since $f,g, \sik_0\in \bb R$, $\sum_k F_k G^*_k \beta_k \in \bb R$ for any real coefficients $\beta_k$, we show easily that $\gamma_{f,g} \in \bb R$ with $\gamma_{f,g}(0) = \|f\|^2 + \|g\|^2 - 2\scp{f}{g} \geq \|f\|^2 + \|g\|^2 - 2\|f\|\,\|g\| \geq 0$. Moreover, if $f=g$, we get $\gamma_{f,f}(0) = 0$ since $\sik_0(0) = 0$, and $\gamma_{f,f}(s) \in [0, 2]$ since $0 \leq \sik_0(s) \leq 1$ for all $s \geq 0$ and $\sum_{k \in \Zbb} \, |F_k|^2 \sik_0(|k| s) \leq \|f\|^2$. 
\end{proof}

In this corollary, the existence of a norm ${\|\cdot\|_\sharp}$ controlling the  behavior of $\cl F^{-1} \Lambda$ is ensured, for instance, if $\Lambda$ is a centered Gaussian distribution, in which case the $\ell_\sharp$-norm is the $\ell_2$-norm. If $\Lambda$ is the Cartesian product of $d$ Cauchy distributions in $\bb R^d$ (with zero location parameter and scale parameter $\tau > 0$), \ie
\begin{equation}
  \label{eq:Cauchy}
  \ts \Lambda(\bs \omega) = \frac{1}{\pi^d \tau^d} \prod_{k=1}^d \frac{\tau^2}{\omega_k^2 + \tau^2},
\end{equation}
then $\cl F^{-1} \Lambda$ amounts to the Laplace distribution and ${\|\cdot\|_\sharp} = {\|\cdot\|_1}$~\cite[Sec.~4.2.2.]{boufounos2017representation}. Moreover, if $\Lambda$ is set to any $\alpha$-stable distribution with $\alpha \geq 1$, \ie a distribution with characteristic function $(\cl F^{-1}\Lambda)(\bs x) \propto \exp(- c \|\bs x\|_\alpha^\alpha)$ with the Gaussian and the Cauchy distributions as special cases, we can reach an (asymmetric) embedding associated with the norm ${\|\cdot\|_\alpha}$~\cite{otero2011generalized}.  

Regarding the distance map $\gamma_{f,g}$, we observe that it does not necessarily vanish at the origin, when $\bs x = \bs y$ in \eqref{eq: geom-pres-our-approach}. As soon as $f \neq g$, a bias exists since 
\begin{equation}
  \label{eq:gamma-fg-origin}
  \ts \gamma_{f,g}(0) = \|f\|^2 + \|g\|^2 - 2\sum_{k \in \Zbb} \, F_kG_k^* = \|f\|^2 + \|g\|^2 - 2\scp{f}{g} = \|f - g\|^2,  
\end{equation}
using $\sik_0(0) = \int_{\bb R^d} \Lambda(\bs \omega) \ud \bs \omega = 1$. For instance, if $f = q$ (with $q$ the universal quantizer defined in~\eqref{eq:universalquantizer}), and $g(\cdot) = \cos(\cdot)$, $\gamma_{q,\cos}(s) = \frac{3}{2} - 2 \Re(F_1)\,\sik_0(s) = \frac{3}{2} - \frac{4}{\pi} \sik_0(s)$ since $\|f\|^2 = 1$, $\|g\|^2 = 1/2$, $2G_k = \delta_{k,1} + \delta_{k,-1}$, and $F_1 = F_{-1} = \frac{2}{\pi}$ from~\eqref{eq:universalquantizer}. Therefore, if $\Lambda$ is a Gaussian distribution with unit standard deviation,
$$
\ts \|\bs z_f(\bs x) - \bs z_g(\bs y)\|^2 \approx \gamma_{q,\cos}(\|\bs x - \bs y\|) = \frac{3}{2} - \frac{4}{\pi}\exp(-\frac{1}{2}\|\bs x - \bs y\|^2).
$$
Compared to the case $f(\cdot)=g(\cdot)=\cos(\cdot)$ where
$$
\ts \gamma_{\cos,\cos}(\|\bs x - \bs y\|) = 1 - \exp(-\frac{1}{2}\|\bs x - \bs y\|^2),
$$
and $\gamma_{\cos,\cos}(0) = 0$, we thus observe a systematic bias $\gamma_{q,\cos}(0) = \frac{3}{2} - \frac{4}{\pi} \approx 0.2268$ at the origin.

This non-vanishing bias\footnote{This bias is here demonstrated when the feature space $\bb C^m$ is equipped with the squared $\ell_2$-distance; the question of its existence for other metrics, such as the $\ell_1$-distance, remains open.} in the case $f \neq g$ is not a drawback per se, since $\gamma_{f,g}$ can still be invertible. For $\gamma_{q,\cos}$ and a Gaussian $\Lambda$ with unit variance, we find
$$
\ts \gamma_{q,\cos}^{-1}(s') = \big(-2\ln( \frac{3\pi}{8} - \frac {\pi}{4} s')\big)^{1/2},\ \text{with}\ s' \in [\frac{3}{2} - \frac{4}{\pi}, \frac{3}{2}].
$$
This shows that, if $\epsilon$ is small enough, we can still reliably infer the distance between $\bs x$ and $\bs y$ from $\|\bs z_f(\bs x) - \bs z_g(\bs y)\|$ provided that $\bs x \approx \bs y$. Indeed, estimating $\gamma_{q,\cos}^{-1}(\|\bs z_f(\bs x) - \bs z_g(\bs y)\|^2) \approx \|\bs x - \bs y\|$ leads to a first order error~\cite{boufounos2017representation} proportional to
$$
\ts \big(\frac{\ud}{\ud s}\gamma_{q,\cos}(s)\big)^{-1} \epsilon = \frac{4}{\pi s}\, \epsilon\exp( \frac{s^2}{2}),  
$$
for $s = \|\bs x - \bs y\|$. As expected from the local nature of the embedding, this error quickly explodes when $s$ is large. 

\begin{remark}
When $f=g$, $\gamma(s) = 2\|f\|^2 - 2\sum_k |F_k|^2 \sik_0(|k| s)$ is invertible iff $\sum_k |F_k|^2 \sik_0(|k| s)$ is invertible. This occurs, for instance, if the one-dimensional kernel $\sik_0$ is differentiable with $\frac{\ud }{\ud s} \sik_0(s) < 0$ for all $s > 0$, which is the case of any symmetric $\alpha$-stable distribution $\Lambda$ for which $(\cl F^{-1} \Lambda)(\bs x) \propto \exp(-c \|\bs x\|^\alpha_\alpha)$. In this case, we easily verify that $\frac{\ud }{\ud s} \gamma(s) > 0$ for $s > 0$, and $\gamma$ is monotonically increasing when $s$ increases, starting from $\gamma(0)=0$. This ensures the injectivity of $\gamma$.     
\end{remark}

\section{Semi-quantized random Fourier features}
\label{sec:quantized}

In this section, we explore one practical application of our general results from Sec.~\ref{sec:nonasymptotic} by instantiating them on the semi-quantized scenario motivated in the Introduction (see Fig.~\ref{fig:intro}). More precisely, we consider the asymmetric RPF setting $\wh{\kappa}_{f,g}(\Vec{x},\Vec{y}) = \langle \Vec{z}_{f}(\Vec{x}), \Vec{z}_{g}(\Vec{y}) \rangle$ in the particular case where: \textit{(i)} one of the signals $\Vec{x}$ is available through its one-bit universal features $\Vec{z}_{q}(\Vec{x}) \in \{-\frac{1}{\sqrt{m}},+\frac{1}{\sqrt{m}}\}^m$ (that is, the first periodic map $f$ is the square wave $q$, alternating between $\pm 1$ with period $2\pi$, see Fig.~\ref{fig:univquant}); \textit{(ii)} the other signal $\Vec{y}$ is available through its classical (full-precision) random Fourier features $\Vec{z}_{\cos}(\Vec{y})$ (that is, the second map $g$ is a cosine). 

Concretely, we start by highlighting a striking general result: when the classical RFF (\ie for which $g(\cdot) = \exp(\im \cdot)$ or $\cos(\cdot)$) are combined with \emph{any} mean smooth function $f \in \pf$, then the asymmetric inner product $\wh{\kappa}_{f,g}(\Vec{x},\Vec{y})$ exactly recovers the initial kernel $\kappa$ that would be approached by symmetric usual RFF $\wh{\kappa}_{g,g}(\Vec{x},\Vec{y})$. Then, to combine this fact with the binary square wave $f = q$, we prove that $q$ is mean smooth (Def.~\ref{def:meanSmoothness}). This finally allows us to obtain a probabilistic uniform bound on the kernel approximation error for the semi-quantized scenario pair, demonstrating in the process how to deal with the scaling issues that appear in such schemes by using the normalization~\eqref{eq:other-normalisation}. 

\paragraph{Expected kernel with a single-frequency nonlinearity:}

Let us begin by noting an interesting consequence of Prop.~\ref{prop:expected_kernel}. From RPF $\Vec{z}_f(\Vec{x})$ captured on $\bs x$ with \textit{any} nonlinearity $f \in \pf$ whose fundamental period is exactly $2\pi$, one can recover in expectation, for a given vector $\bs y$,  the evaluation the shift-invariant kernel $\kappa(\Vec{x},\Vec{y}) = \sik(\Vec{x}-\Vec{y})$ associated with the sampling of the projections $\Vec{\omega}_j \sim \Lambda = \cl F \sik$. 

Indeed, using Prop.~\ref{prop:expected_kernel} in the complex field, and setting $g(\cdot) = \exp(\im \cdot)$ for the RPF of $\Vec{y}$---which in this case is the RFF (Def.~\ref{def:RFF})---ensures that $\kappa_{f,\exp(\im \cdot)}(\Vec{x},\Vec{y}) = F_1 \kappa(\Vec{x},\Vec{y})$. Intuitively, the dithering averages out all the high-frequency components in $f$, leaving only its fundamental frequency. When dealing with real-valued quantities $\kappa, f \in \bb R$, we can use the real RFF (where $g(\cdot) = \Re \exp(\im \cdot) = \cos(\cdot)$) instead, and using the normalized kernel \eqref{eq:other-normalisation} with \eqref{eq:expected_kernel_expression} gives
\begin{equation}
  \label{eq:cosineasymmetric}
  \ts  \dot{\kappa}_{f,\cos}(\bs x, \bs y) = \frac{1}{\Re F_1}  \kappa_{f,\cos}(\Vec{x},\Vec{y}) =
\frac{1}{\Re F_1} \expec{} \langle f(\bs\Omega^\top \Vec{x} + \Vec{\xi}),\cos(\bs\Omega^\top \Vec{y} + \Vec{\xi}) \rangle = \kappa(\Vec{x},\Vec{y}),
\end{equation}
since $\scp{f}{g} = \sum_k F_k G^*_k = \Re F_1$. We thus recover, through $\dot{\kappa}_{f,\cos}$ the initial kernel $\kappa$, thanks to \textit{a rescaling} by $1/\scp{f}{g} = 1/\Re\{F_1\}$ which must be taken into account for a fair comparison. 

\begin{remark}
  In theory, we can thus recover, from $\Vec{z}_f(\Vec{x})$, the kernel $\kappa$ at many different scales by ``probing'' it with $\Vec{z}_{\cos(k\,\cdot)}(\Vec{y})$ for any scale $k$ such that $F_k \neq 0$. However, we observe in practice that the kernel approximation error quickly increases with $k$. This can be understood in the light of Prop.~\ref{prop:kernelApproxUniform}, since one easily show that $L_{\cos(k \, \cdot)} = |k|$ and\footnote{First, $L^\mu_{\cos(k \cdot)} \leq L^\mu_{\exp(ik \cdot)} = |k|$ since $|\!\cos(\alpha) - \cos(\beta)|\leq |e^{\im \alpha} - e^{\im \beta}|$ for all $\alpha,\beta \in \bb R$. Second, from $|\!\cos(k(t+r)) - \cos(kt)| = 2 |\!\sin(k(t+\frac{r}{2}))\sin\frac{k r}{2}|$ for all $t,r \in \bb R$, we get, by fixing $r=\delta$ in (\ref{eq:meanSmooth}), $L^\mu_f \geq \sup_{\delta>0} \frac{2}{\delta}|\!\sin\!\frac{k\delta}{2}|\, \bb E_{t\sim \cl U([0,2\pi])} |\!\sin(k(t+\frac{\delta}{2}))| = \frac{2}{\pi} |k|.$} $\frac{2}{\pi} |k| \leq L^\mu_{\cos(k \, \cdot)} \leq |k|$. Moreover, if $\|f\|^2 = \sum_k |F_k|^2$ is bounded, each rescaling factor $({\Re\{F_k\}})^{-1}$ grows as $k$ increases; for instance, $({\Re\{F_k\}})^{-1} \propto |k|$ for $f = q$. 
\end{remark}      

The asymmetric scheme in~\eqref{eq:cosineasymmetric} is interesting because it allows the same level of control over the approximated kernel as the usual RFF (which is an improvement compared to the scale mixture of RFF kernels imposed by Prop.~\ref{prop:BoufounosKernel}) while still enjoying the freedom to use any type of features $\Vec{z}_f(\Vec{x})$ for one of the signals being compared---a particularly appealing choice being $f = q$, the one-bit universal quantization. However, in order to use Prop.~\ref{prop:kernelApproxUniform} to obtain uniform error bounds, we still need to prove the mean Lipschitz smoothness of this (discontinuous) map.

\paragraph{Mean Lipschitz smoothness of universal quantization:}
We now show that the one-bit universal quantization function $q$ (\ie the square wave) has the mean Lipschitz property---although it is discontinuous. The same strategy could be used to prove the mean Lipschitz smoothness of any function in $\pf$ with a finite number of discontinuities per period.

\begin{proposition}
  \label{prop:meanSmoothUniversalQuantization}
  The one-bit universal quantization function $q$, defined in~\eqref{eq:universalquantizer}, has the mean Lipschitz smoothness property (Def.~\ref{def:meanSmoothness}) with constant
  \begin{equation}
    L^{\mu}_q = \ts \frac{4}{\pi}\|q\|_{\infty} = \frac{4}{\pi}.
  \end{equation}
\end{proposition}
\begin{proof}
  By definition of the mean smoothness property, we must find $L^{\mu}_q$ such that
  \begin{equation*}
    \ts \frac{1}{2\pi} \int_{0}^{2\pi} \max_{r \in [-\delta, \delta]} \{ |q(t + r) - q(t)| \} \: \mathrm{d}t \leq L^{\mu}_q \cdot \delta.
  \end{equation*}
  We start by characterizing the integrand $I_{\delta}(t) := \max_{|r| \leq \delta} \{ |q(t + r) - q(t)| \} $ . Since $q(t)$ is constant (in particular, $q(t) = \pm \|q\|_{\infty}$) everywhere except on discontinuities at $t \in \frac{\pi}{2} +  \pi \bb Z$ where its height changes by an absolute step of $2\|q\|_{\infty}$, we have for $k \in \bb Z$ that (see Fig.~\ref{fig:univquant})
  $$ I_{\delta}(t) = \begin{cases} 
    0 & \frac{\pi}{2} + k\pi + \delta < t < \frac{\pi}{2} + (k+1)\pi - \delta \\
    2 \|q\|_{\infty} & \frac{\pi}{2} + k\pi - \delta \leq t \leq \frac{\pi}{2} + k\pi + \delta.
  \end{cases} $$
  Integrating this over one period gives $\int_0^{2\pi} I_{\delta}(t) \mathrm{d}t = \min(4\delta, 2\pi ) \cdot 2\|q\|_{\infty} \leq 8 \|q\|_{\infty} \cdot \delta$, \ie $L^{\mu}_q = \frac{4}{\pi}\|q\|_{\infty}$.
\end{proof}

\paragraph{Combining quantized and cosine features:}
We are interested in approximating a specific kernel $\kappa(\Vec{x},\Vec{y})$ by the asymmetric features product $\langle \Vec{z}_q(\Vec{x}), \Vec{z}_{\cos}(\Vec{y}) \rangle$.
This product gives on average $\kappa_{q,\cos} =  \frac{2}{\pi} \kappa $ (recall from~\eqref{eq:universalquantizer} that $Q_k = \frac{2}{k \pi}(-1)^{(k-1)/2}$ for $k$ odd and $0$ otherwise), and the re-scaled approximation~$\wt{\kappa}_{q,\cos}$ defined in \eqref{eq:other-normalisation} in this case is given by
\begin{equation}
  \label{eq:semiQuantizedEstimator}
  \ts \wt{\kappa}_{q,\cos}(\Vec{x},\Vec{y}) := \frac{\pi}{2} \cdot \wh{\kappa}_{q,\cos}(\Vec{x},\Vec{y}) = \frac{\pi}{2}\langle \Vec{z}_q(\Vec{x}), \Vec{z}_{\cos}(\Vec{y}) \rangle  \approx \kappa(\Vec{x},\Vec{y}).
\end{equation}

\noindent We bound the error of approximating the kernel over an infinite compact set $\Sigma$ thanks to Prop.~\ref{prop:kernelApproxUniform}.

\begin{corollary}[Uniform kernel approximation error from quantized-complex asymmetric features]
  \label{cor:main}
  Given $\epsilon > 0$, a compact set $\Sigma$, and the frequency distribution $\Lambda$ such that $C_{\Lambda} < \infty$, provided that
  \begin{equation}
    \ts m \geq 32 \pi^2 \cdot \frac{1}{\epsilon^2} \cdot \cl H_{\epsilon/((8 + 6\pi) C_{\Lambda})}(\Sigma),
  \end{equation}
  the following kernel approximation bound holds uniformly:
  \begin{equation}
    \ts \big| \wt{\kappa}_{q,\cos}(\Vec{x},\Vec{y}) - \kappa(\Vec{x},\Vec{y}) \big| \leq \epsilon,\quad  \forall \Vec{x},\Vec{y} \in \Sigma,
  \end{equation}
  with probability exceeding $1 - 3 \exp( - \frac{m \epsilon^2}{16 \pi^2})$
\end{corollary}
\begin{proof}
  Apply Prop.~\ref{prop:kernelApproxUniform} with $f = q$, $g = \cos$, using that $L^{\mu}_q = \frac{4}{\pi}$ (from Prop.~\ref{prop:meanSmoothUniversalQuantization}) and $L^{\mu}_{\cos} = 1$: for any given $\epsilon'>0$, if $m \geq 128 \cdot \frac{1}{\epsilon'^2} \cdot \cl H_{\epsilon'/c}(\Sigma)$ with $c = (12+\frac{16}{\pi}) C_{\Lambda}$, 
  \begin{equation*}
    \ts \bb P\left[ \exists \, \Vec{x},\Vec{y} \in \Sigma \: : \: \left| \wh{\kappa}_{q,\cos}(\Vec{x},\Vec{y}) -\frac{2}{\pi} \kappa(\Vec{x},\Vec{y}) \right| > \epsilon' \right] \leq 3 \exp \left( - \frac{m \epsilon'^2}{64} \right).
  \end{equation*}
  To take into account the scaling of the kernel, set $\epsilon = \frac{\epsilon'}{\Re Q_1} = \frac{\pi}{2}\epsilon'$.
\end{proof}

\begin{example}
	Consider a final time our example of a union of $S$ $s$-dimensional subspaces (see Ex.~\ref{ex:UoS-example}) combined with the Gaussian kernel with unit bandwidth (and $C_\Lambda = 1$). In this case, the kernel approximation error $| \wt{\kappa}_{q,\cos}(\Vec{x},\Vec{y}) - \kappa(\Vec{x},\Vec{y}) |$ is uniformly bounded over $\Sigma$, with high probability, provided that the number of features satisfies $m \geq  C \epsilon^{-2} ( s \log( \frac{8 + 6\pi}{\epsilon}) + \log S)$, which reduces to $m \geq  C \epsilon^{-2}  s \log( \frac{(8 + 6\pi) e d}{s \epsilon})$ for bounded $s$-sparse signals.
\end{example}

Corollary~\ref{cor:main} provides a theoretical guarantee justifying the semi-quantized scheme presented in the Introduction. In the next section, we further validate this approach from numerical simulations.

\section{Experiments}
\label{sec:experiments}

In all our experiments, we are interested in approximating a kernel $\kappa(\Vec{x},\Vec{y})$, associated with the RFF sampled with $\Lambda^m$, by the inner product of random periodic features. We focus on (combinations of) the two types of features discussed in the previous section: the ``real'' random Fourier features $\Vec{z}_{\cos}(\Vec{x}) = \frac{1}{\sqrt{m}} \cos(\bs \Omega^\top \Vec{x} + \Vec{\xi}) \in \bb R^m$, and the universal features $\Vec{z}_{q}(\Vec{x}) = \frac{1}{\sqrt{m}} q(\bs \Omega^\top \Vec{x} + \Vec{\xi}) \in \{-\frac{1}{\sqrt{m}},+\frac{1}{\sqrt{m}}\}^m$, where $m$ is the number of features (or dimension), and where we generate $\bs \Omega \sim \Lambda^m$ and $\Vec{\xi} \sim \cl U^m([0,2\pi))$.
Recalling the rescaling \eqref{eq:other-normalisation} for fair comparisons of the approximated kernels with $\kappa$, we thus consider three possible combinations: the classical (real) random Fourier features (with $\|f\|^2 = \|\cos(\cdot)\|^2= \frac{1}{2}$),
$$ \ts \wt{\kappa}_{\cos,\cos}(\Vec{x},\Vec{y}) = 2 \langle \Vec{z}_{\cos}(\Vec{x}), \Vec{z}_{\cos}(\Vec{y}) \rangle  \approx \kappa(\Vec{x},\Vec{y}) ,$$
our asymmetric ``semi-quantized'' scheme (with $\scp{f}{g}= 2/\pi$),
$$ \ts \wt{\kappa}_{q,\cos}(\Vec{x},\Vec{y}) = \frac{\pi}{2} \langle \Vec{z}_{q}(\Vec{x}), \Vec{z}_{\cos}(\Vec{y}) \rangle  \approx \kappa(\Vec{x},\Vec{y}) ,$$
and the fully quantized inner product from~\cite{Boufounos2013efficientCodingQuantized} (with $\|q\|^2 = \sum_k |Q_k|^2 = 1$),
$$
\ts \wt{\kappa}_{q,q}(\Vec{x},\Vec{y}) =  \langle \Vec{z}_{q}(\Vec{x}), \Vec{z}_{q}(\Vec{y}) \rangle  \approx \kappa_{q,q}(\Vec{x},\Vec{y}) = \sum_{k \in \bb Z} |Q_k|^2 \kappa(k\Vec{x},k\Vec{y}) \neq \kappa(\Vec{x},\Vec{y}).
$$

\subsection{Qualitative analysis of the expected kernel}
\label{sec:qual-analys-expect}

As a first experiment, we visually demonstrate that our asymmetric product $\wt{\kappa}_{q,\cos}$ indeed approaches a target kernel $\kappa$. As target, we use the Gaussian kernel $\kappa(\Vec{x},\Vec{y}) = \exp(-\frac{\|\Vec{x}-\Vec{y}\|_2^2}{2\sigma^2})$ (for which $\Lambda$ is the Gaussian distribution $\cl N(\Vec{0},\sigma^{-2}\bs I_d)$), as well as the Laplace kernel $\kappa(\Vec{x},\Vec{y}) = \exp(-\frac{\|\Vec{x}-\Vec{y}\|_1}{\tau})$ (where $\Lambda$ is the Cauchy distribution~\eqref{eq:Cauchy}), both in dimension $d = 5$.

We evaluate the three inner products $\wt{\kappa}_{\cos,\cos}$, $\wt{\kappa}_{q,\cos}$ and $\wt{\kappa}_{q,q}$ on $n = 2000$ pairs of vectors $\{(\Vec{x}_i,\Vec{y}_i)\}_{i = 1}^n$, that are generated as follows. We first sample $\Vec{x}_i \in \bb R^5$ according to a standard normal distribution, then pick $\Vec{y}_i = \Vec{x}_i + \lambda_i \Vec{u}_i$, where $\Vec{u}_i$ is a randomly chosen unit vector (\ie normalized such that $\|\Vec{u}_i\|_p = 1$ with $p = 2$ for the Gaussian kernel and $p=1$ for the Laplace one), and $\lambda_i = \frac{(i-1) \lambda_{\max}}{(n-1)}$ is a controlled distance which is incremented for each pair, linearly increasing from $0$ to $\lambda_{\max} = 5$. This ensures that we test the kernel approximations uniformly in the desired range of distances $\|\Vec{x}-\Vec{y}\|_p$.

We then generate one realization of $\bs \Omega, \bs \xi$ (with $m = 200$ for the Gaussian kernel, and $m=2000$ for the Laplace kernel, values which were arbitrarily chosen to get pleasing visualizations), which we use to compute the real RFF $\big\{\big(\Vec{z}_{\cos}(\Vec{x}_i), \Vec{z}_{\cos}(\Vec{y}_i)\big) \big\}_{i=1}^n$ as well as the universal quantization features $\big\{\big(\Vec{z}_q(\Vec{x}_i), \Vec{z}_{q}(\Vec{y}_i) \big) \big\}_{i=1}^n$, from which we get $n$ evaluations of the classical RFF inner product $\wt{\kappa}_{\cos,\cos}(\Vec{x}_i,\Vec{y}_i) = 2 \langle \Vec{z}_{\cos}(\Vec{x}_i), \Vec{z}_{\cos}(\Vec{y}_i)\rangle$, the asymmetric product $\wt{\kappa}_{q,\cos}(\Vec{x}_i,\Vec{y}_i) = \frac{\pi}{2} \langle \Vec{z}_{q}(\Vec{x}_i), \Vec{z}_{\cos}(\Vec{y}_i) \rangle$, and the fully quantized product $\wt{\kappa}_{q,q}(\Vec{x}_i,\Vec{y}_i) =  \langle \Vec{z}_{q}(\Vec{x}_i), \Vec{z}_{q}(\Vec{y}_i) \rangle$.

Those evaluations are shown as black dots in Fig.~\ref{fig:exp1_Gaussian_Laplacian} for the Gaussian and Laplace kernels in the top and bottom rows, respectively. As predicted by the theory, both the RFF product $\wt{\kappa}_{\cos,\cos}$ and our semi-quantized product $\wt{\kappa}_{q,\cos}$ concentrate around the target kernel $\kappa$ (in red). As expected from~\cite{boufounos2017representation}, in the fully quantized case the product $\wt{\kappa}_{q,q}$ rather concentrates around a different ``distorted'' kernel, $\kappa_{q,q}$. Note that we increased the feature space dimension $m$ tenfold for the Cauchy kernel, which reduced the variance of the approximation. However, it is difficult to notice a substantial difference of approximation quality between the plain RFF $\wt{\kappa}_{\cos,\cos}$ and the semi-quantized asymmetric scheme $\wt{\kappa}_{q,\cos}$. We thus perform a more quantitative exploration of the error $|\wt{\kappa}_{q,\cos} - \kappa|$ in the next experiment.

\begin{figure}[t]
  \centering
  \newcommand{\insertlabel}[1]
  {\hspace{-0.27\textwidth}\raisebox{1.7cm}{#1}\hspace{0.27\textwidth}}
  \subfloat{
    \includegraphics[width=0.3\textwidth]{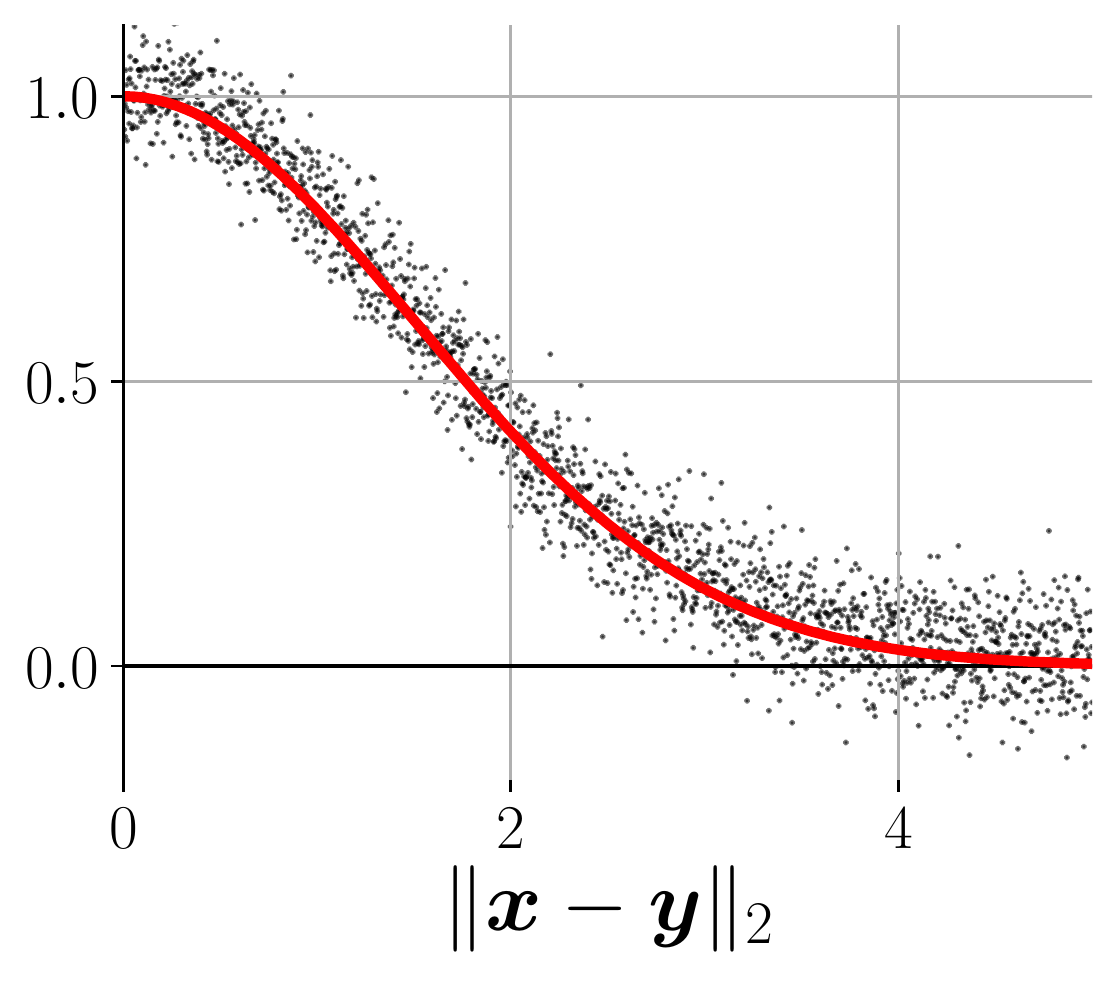}
    \insertlabel{(a)}
    \label{fig:exp1_Gaussian_cc}}
  \subfloat{
    \includegraphics[width=0.3\textwidth]{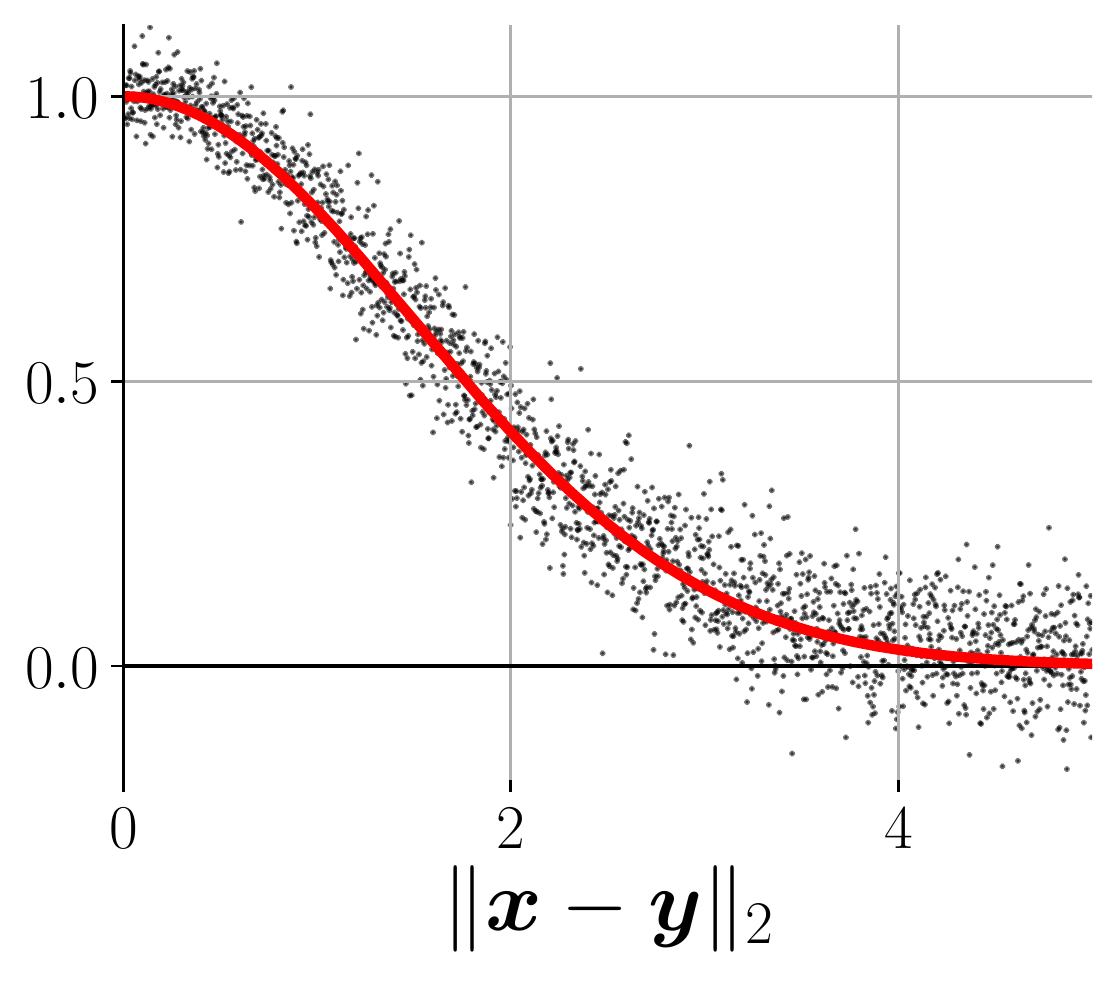}
    \insertlabel{(b)}
    \label{fig:exp1_Gaussian_cq}}
  \subfloat{
    \includegraphics[width=0.3\textwidth]{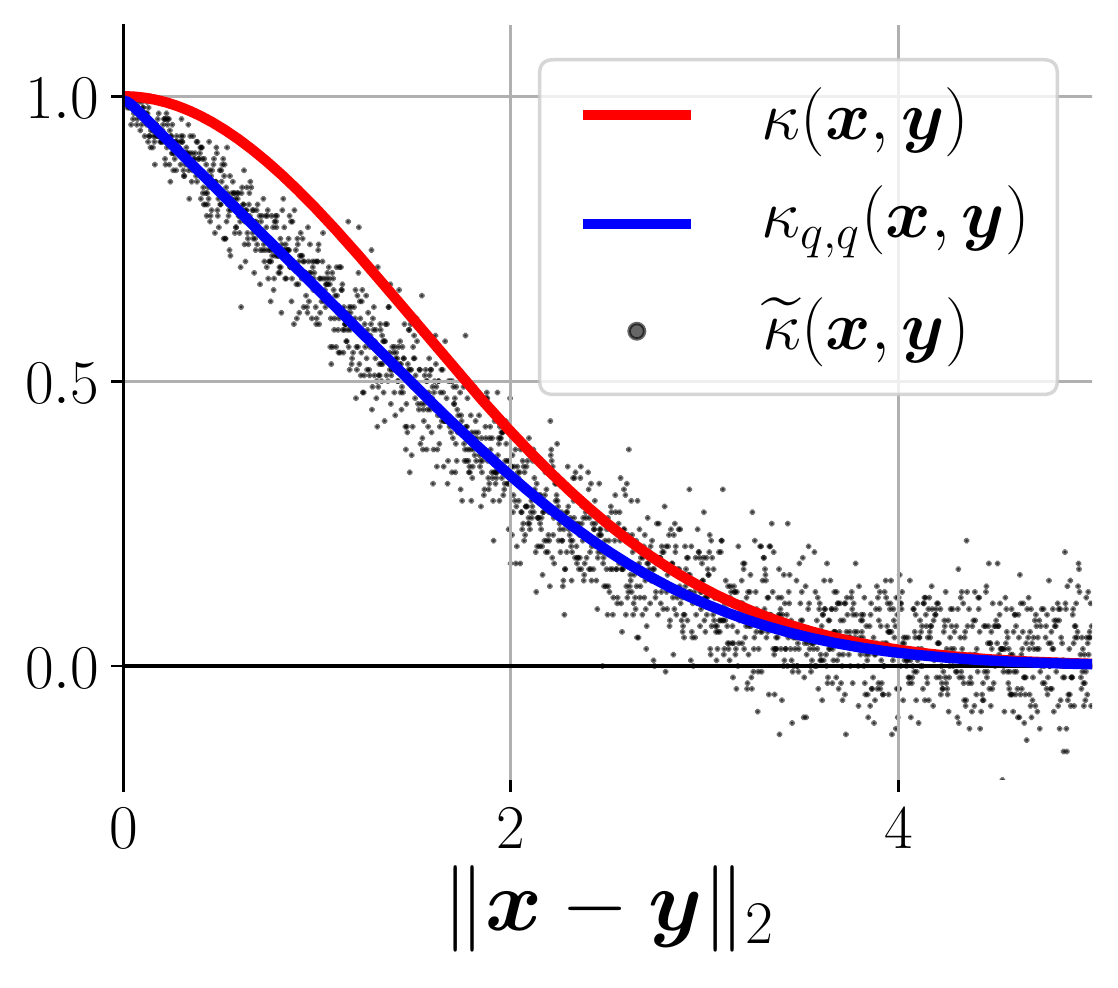}
    \insertlabel{(c)}
    \label{fig:exp1_Gaussian_qq}}\\
    \renewcommand{\insertlabel}[1]
  {\hspace{-0.27\textwidth}\raisebox{1.3cm}{#1}\hspace{0.27\textwidth}}
    \subfloat{
      \includegraphics[width=0.3\textwidth]{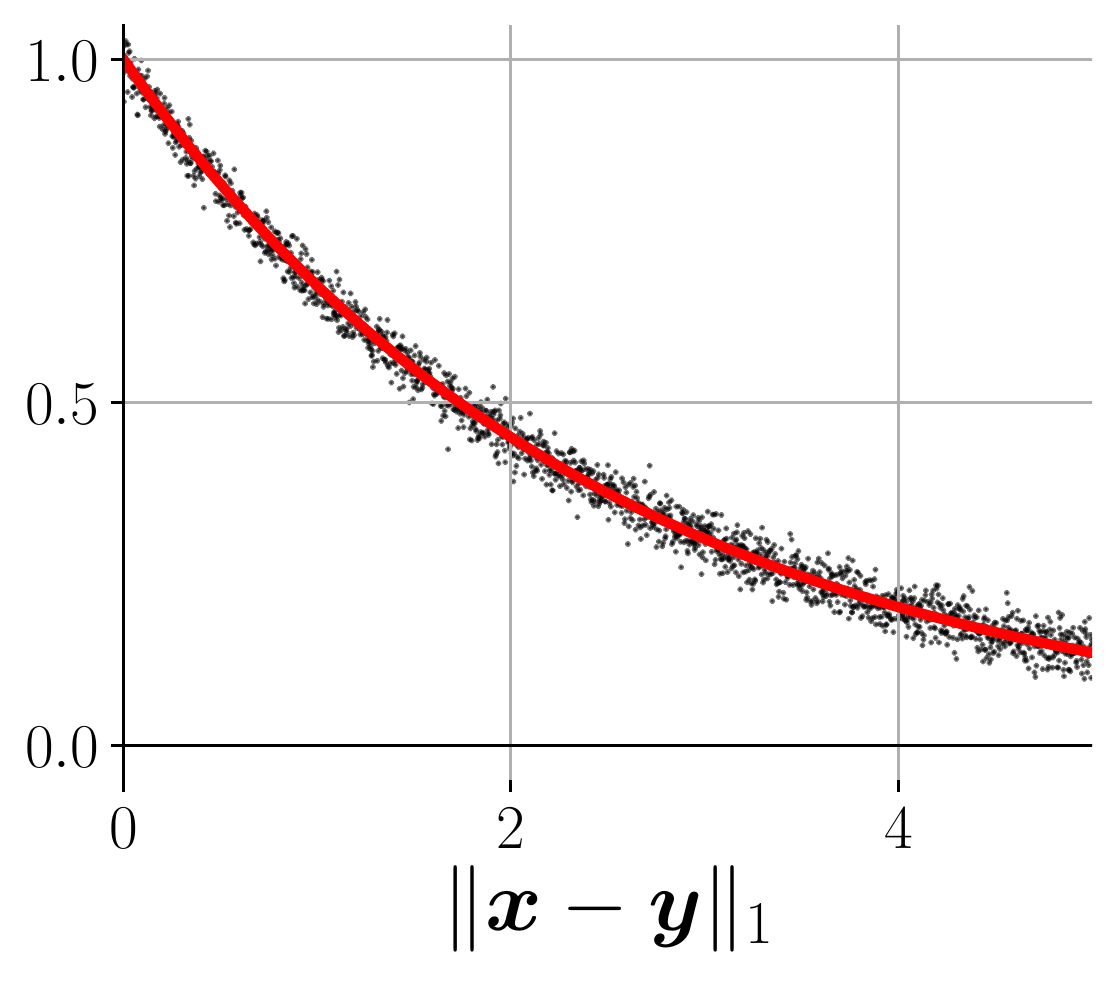}
      \insertlabel{(d)}
    \label{fig:exp1_Laplace_cc}}
  \subfloat{
    \includegraphics[width=0.3\textwidth]{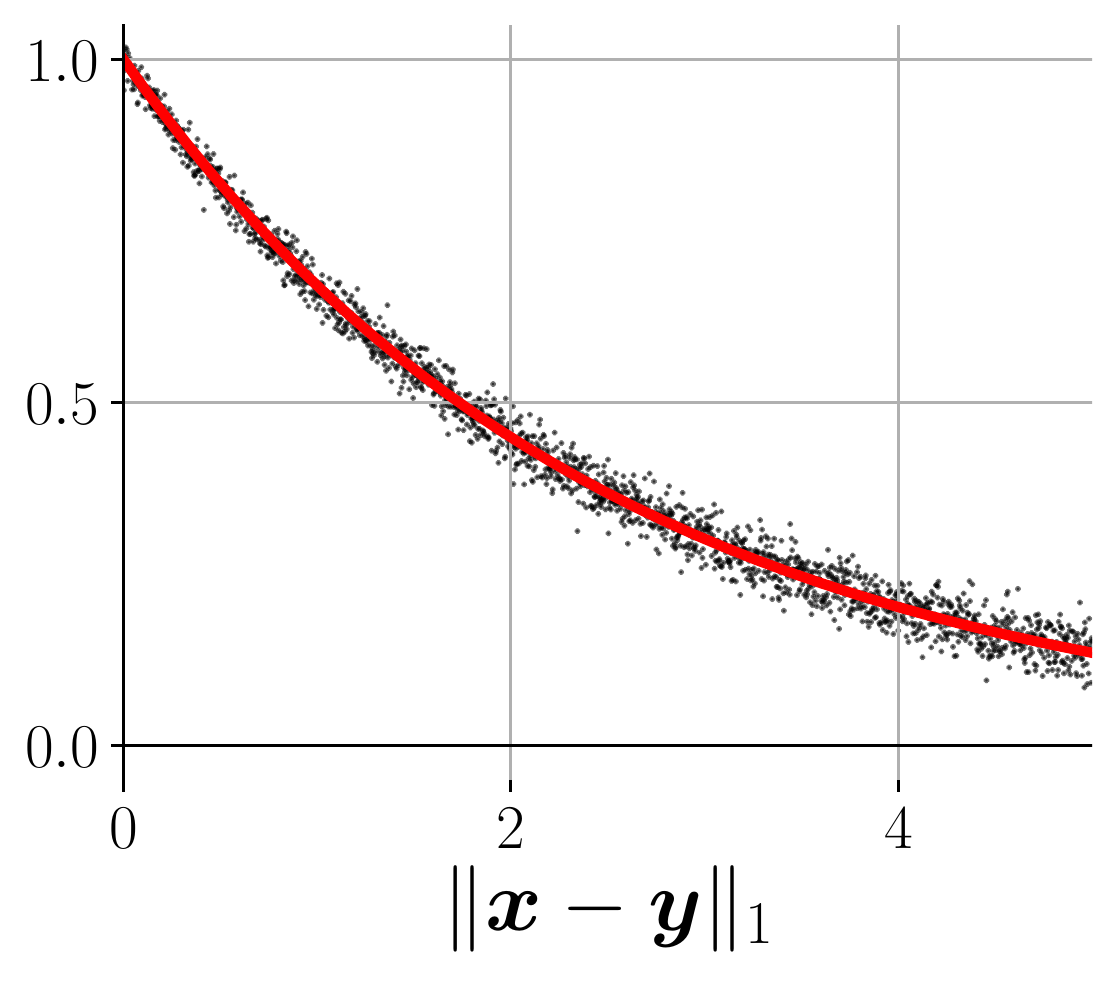}
    \insertlabel{(e)}
    \label{fig:exp1_Laplace_cq}}
  \subfloat{
    \includegraphics[width=0.3\textwidth]{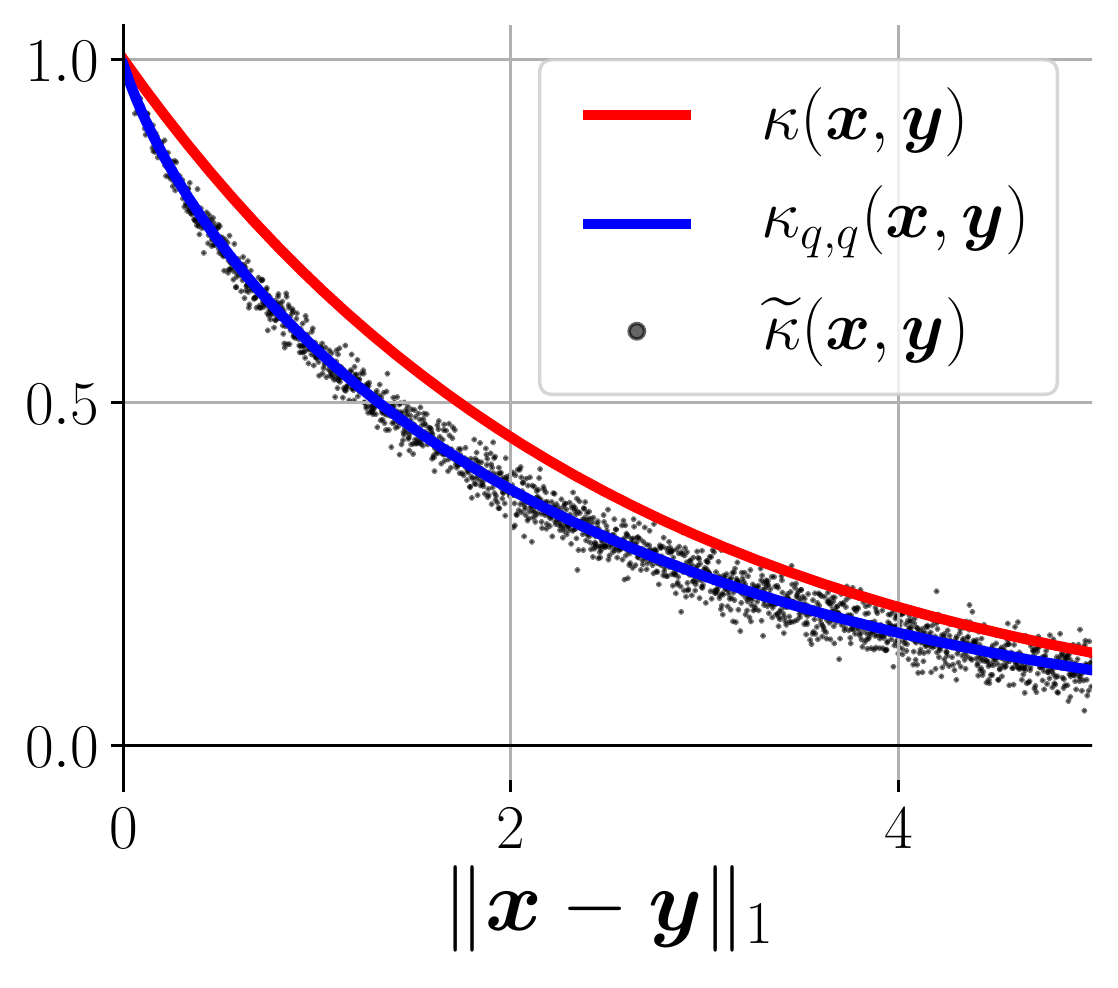}
    \insertlabel{(f)}
    \label{fig:exp1_Laplace_qq}}
  \caption{Comparison between the target kernel $\kappa(\Vec{x},\Vec{y})$ (red curves), and the approximations (the black scatter plots each evaluated over $n=200$ pairs $\{(\Vec{x}_i,\Vec{y}_i)\}_{i=1}^n$) using, for (a,d),  the plain random Fourier features $\wt{\kappa}_{\cos,\cos}(\Vec{x}_i,\Vec{y}_i) = 2 \langle \Vec{z}_{\cos}(\Vec{x}_i), \Vec{z}_{\cos}(\Vec{y}_i)\rangle$, for (b,e), our asymmetric cosine-quantized pair $\wt{\kappa}_{q,\cos}(\Vec{x}_i,\Vec{y}_i) = \frac{\pi}{2} \langle \Vec{z}_{q}(\Vec{x}_i), \Vec{z}_{\cos}(\Vec{y}_i) \rangle$, and for (c,f), only quantized features $\wt{\kappa}_{q,q}(\Vec{x}_i,\Vec{y}_i) =  \langle \Vec{z}_{q}(\Vec{x}_i), \Vec{z}_{q}(\Vec{y}_i) \rangle$. In the last case, the ``distorted'' expected kernel $\expec{}\wt{\kappa}_{q,q} = \kappa_{q,q}$ is shown in blue. For (a-c; top row), the comparison is made for the target Gaussian kernel $\kappa(\Vec{x},\Vec{y}) = \exp(-\frac{\|\Vec{x}-\Vec{y}\|_2^2}{2\sigma^2})$ with scale $\sigma = 1.5$, and the
approximated kernels use $m = 200$ random features evaluated. For (d-f; bottom row), the target kernel is the Laplace kernel $\kappa(\Vec{x},\Vec{y}) = \exp(-\frac{\|\Vec{x}-\Vec{y}\|_1}{\tau})$ with scale $\tau = 1.5$, and its different approximations are set with $m = 2000$.
  }
  \label{fig:exp1_Gaussian_Laplacian}
\end{figure}

\subsection{Quantitative analysis of the approximation error}
\label{sec:quant-analys-appr}

To perform a more quantitative analysis of the kernel approximation from Cor.~\ref{cor:main}, we perform another set of experiments that highlight the evolution of the worst-case error (associated with the hybrid estimation~\eqref{eq:semiQuantizedEstimator}),
$$
\ts \epsilon _{q,\cos}(\Sigma) := \sup_{\Vec{x},\Vec{y} \in \Sigma}|\wt{\kappa}_{q,\cos}(\Vec{x},\Vec{y}) - \kappa(\Vec{x},\Vec{y})|,
$$
as a function of $m$. In this synthetic experiment, we work with a finite set of signals $\Sigma = \{\Vec{x}_i \in \bb R^d\}_{i=1}^n$ obtained from a Gaussian distribution $\Vec{x}_i \distiid \cl N(\Vec{0},\wt{\sigma}^2\bs I_d)$ in with $\wt{\sigma} = 10$. We target a Gaussian kernel $\kappa$ of bandwidth $\sigma = 0.25$, and evaluate the absolute approximation error $\epsilon_{q,\cos}(\Sigma)$ over all vector pairs of $\Sigma$. We record the largest error encountered this way, and repeat this process for several feature dimensions~$m$.

First, we let $n = |\Sigma|$, the number of signals, vary between $10$ and $500$ (by sampling $21$ equally-spaced values for $\log_{10}(n)$), and generate a new dataset in dimension $d = 32$ each time. For each value of $m$ (varying uniformly between $100$ and $1300$), we repeat $50$ independent draws of $\bs \Omega$ and $\Vec{\xi}$ and report Fig.~\ref{fig:tmpexp21} the number of times that  $\epsilon_{q,\cos}(\Sigma) \leq \bar{\epsilon}$ for a fixed threshold $\bar{\epsilon} = 0.15$ (\ie we report the empirical ``success rate'' of the embedding). As expected, the feature space dimension $m$ needed to succeed (highlighted in red for $50\%$ success rate) scales as $O(\log n)$. We also show in dashed yellow the same transition for the worst-case error $\epsilon _{\cos,\cos}(\Sigma)$ (evaluating $|\wt{\kappa}_{\cos,\cos} - \kappa|$ over all vector pairs of $\Sigma$) committed by the plain RFF, which shows the price to pay for quantization. Roughly speaking, the same success rate is achieved for $\wt{\kappa}_{q,\cos}(\Vec{x}_i,\Vec{x}_j)$ as for $\wt{\kappa}_{\cos,\cos}(\Vec{x}_i,\Vec{x}_j)$ provided we take $\sim 33\%$ more random features, which still corresponds to a bitrate reduction for the features of the fist signal $\Vec{x}_i$. Finally, for the sake of comparison we also show in blue the same success rate but when measuring the proximity error between the two approximations (semi-quantized and usual RFF), \ie $|\wt{\kappa}_{q,\cos} - \kappa_{\cos,\cos}|$, which relates to our bound in Cor.~\ref{cor:proximity-approx-kernels}. 

Second, we fix one single dataset $\Sigma$ of $n = 200$ signals in $\bb R^5$, but record the precise value of the worst-case error $\epsilon := \epsilon _{q,\cos}(\Sigma)$ for each of the $50$ draw of $\bs \Omega, \Vec{\xi}$ at different values of $m$ (this time varying along a logarithmic scale). We display the various errors $\epsilon$ obtained as box-plots in Fig.~\ref{fig:tmpexp22}. As can be seen by comparison with the $-1/2$ slope in red, the error is controlled with high probability (discarding the outliers from the box-plots) provided that $m = O(\epsilon^{-2})$, as expected from Prop.~\ref{prop:kernelApproxUniform}. 

\begin{figure}
  \centering
  \subfloat[Empirical succes rate \emph{w.r.t.} $m$ and $n$.]{
    \includegraphics[width=0.46\textwidth]{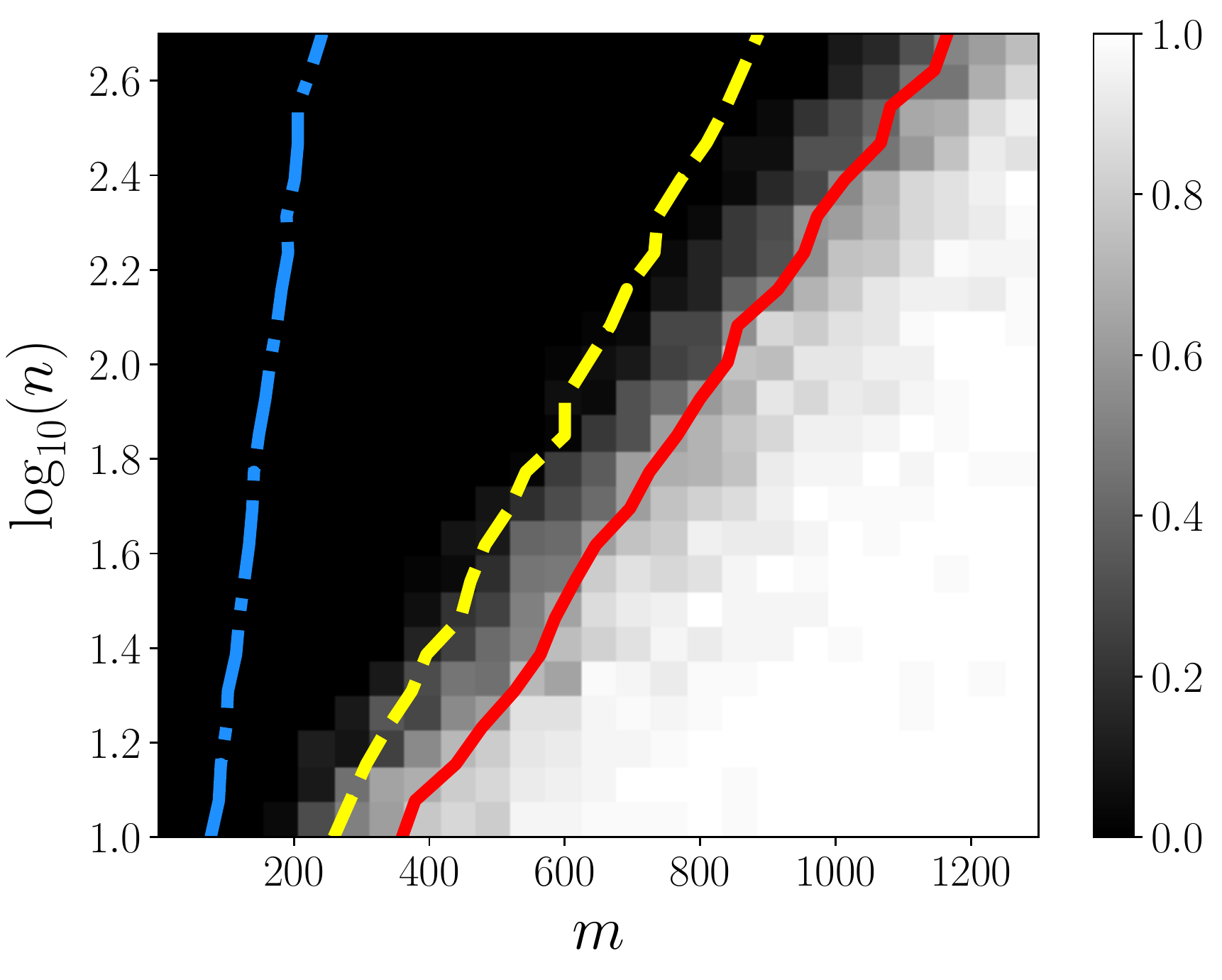}
    \label{fig:tmpexp21}}
  \hspace{12px}
  \subfloat[Largest detected error \emph{w.r.t.} $m$.]{	\raisebox{1mm}{\includegraphics[width=0.46\textwidth]{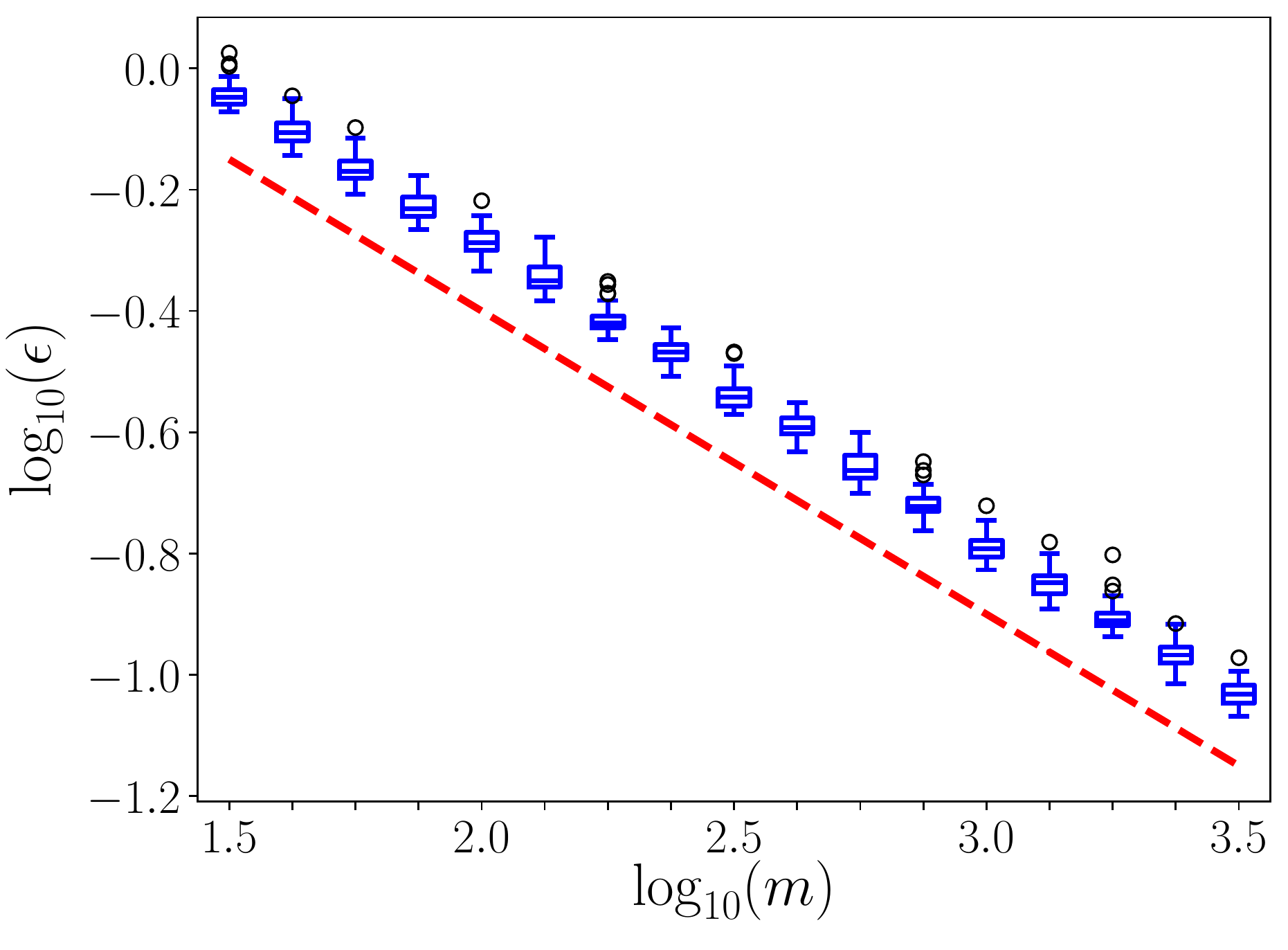}}
    \label{fig:tmpexp22}}
  \caption{(a) Empirical success rate (from $100\%$ success in white to $0\%$ in black) of the kernel approximation (defining success as $\epsilon _{q,\cos}(\Sigma) < \bar{\epsilon} = 0.15$), as a function of $m$ (log scale) for varying dataset size $n = |\Sigma|$. The transition to $50\%$ success or more is highlighted in solid red; the same curve is shown for the success rate of the classical RFF (when $\epsilon _{\cos,\cos}(\Sigma) < \bar{\epsilon}$) in dashed yellow. The blue line represents the success rate related to the proximity between those two kernel approximations, \ie when $\sup_{\Vec{x},\Vec{y} \in \Sigma}|\kappa_{q,\cos}(\Vec{x},\Vec{y}) - \kappa_{\cos,\cos}(\Vec{x},\Vec{y})| < \bar{\epsilon}$.
  (b) Largest kernel approximation error $\epsilon := \epsilon _{q,\cos}(\Sigma)$ as a function of $m$ for $50$ draws of $\bs\Omega$ and $\Vec{\xi}$ (the blue box-plots). The dashed red line shows the slope $\log_{10}(\epsilon) \sim -\frac{1}{2} \log_{10}(m) $, for reference.}
  \label{fig:tmpexp21-22}
\end{figure}

\subsection{Application: semi-quantized support vector machines}
\label{sec:appl-semi-quant}

As a last experiment, we demonstrate how the asymmetric features can be used in practice, for the particular case of Support Vector Machine classification~\cite{boser1992training,scholkopf2002learning}, where the goal is to assign a class label $y' \in \bb Z$ to new query vectors $\Vec{x}' \in \Sigma$ from labeled training data $\cl T := \{(\Vec{x}_i,y_i)\}_{i = 1}^n$. In the binary classification case (labels $y_i \in \{\pm 1\}$), given a kernel $\kappa^\sharp: \Sigma \times \Sigma \to \bb R$, the learned SVM classifier $\theta$ predicts the class of an incoming vector $\Vec{x}'$ as
\begin{equation}
  \label{eq:SVM}
\ts  \theta(\Vec{x}') = \sign \big( \sum_{i \in \cl S^*} \alpha_i y_i \, \kappa^\sharp(\Vec{x}',\Vec{x}_i) + b \big),
\end{equation}
where $\cl S^* \subset [n]$ is the index set of support vectors $\Vec{x}_i$, $0 < \alpha_i \leq R$ are the related weights, and $b$ is a bias (or ``intercept'') term. The quantities $\{\cl S^*, \alpha_i, b\}$ are the parameters to be learned during the training stage, while the kernel $\kappa^\sharp$ and regularization strength $R > 0$ (where a smaller $R$ corresponds to more regularization) are hyper-parameters to be set beforehand. In the multi-class case (where $y_i \in [N]$ for $N$ classes), we use the ``one-versus-rest\footnote{also known as ``one-versus-all''.}'' strategy where one binary classifier is trained to recognize each class.

In this experiment, given a simple classification task described below, we propose to train the SVM with a given kernel $\kappa^\sharp_{\rm L}$, and to test the classification of new samples with another kernel $\kappa^\sharp_{\rm T}$ that approximates $\kappa^\sharp_{\rm L}$, hence assessing how the classifier $\theta$ is impacted by this modification. We consider two options. In the first we train a kernel SVM on the raw data $\cl T$ with a ``true'' kernel $\kappa^\sharp_{\rm L} = \kappa$  and use the approximated kernels provided by random periodic features (setting $\kappa^\sharp_{\rm T}$  to the kernels $\wt{\kappa}_{\cos,\cos}$, $\wt{\kappa}_{q,\cos}$ and $\wt{\kappa}_{q,q}$ defined from \eqref{eq:other-normalisation}) \emph{only at the inference stage}. In this mode, which is the viewpoint we adopted in most of this work (\eg in Prop.~\ref{prop:kernelApproxUniform}), we thus interpret the RPF inner products as a means to approximate as well as possible the given kernel $\kappa$.

In a second case, we directly train a linear SVM on the cosine random Fourier features of the training set $\cl T' := \{(\Vec{z}_{\cos}(\Vec{x}_i), y_i)\}_{i = 1}^{n}$, which amounts to using $\kappa^\sharp_{\rm L} = \wt{\kappa}_{\cos,\cos}$ in \eqref{eq:SVM} as the reference kernel during training. At the testing stage,  we still set  $\kappa^\sharp_{\rm T}$ to $\wt{\kappa}_{\cos,\cos}$, $\wt{\kappa}_{q,\cos}$ and $\wt{\kappa}_{q,q}$. In this scenario, the random periodic features are rather (implicitly) used to define a specific kernel $\wt{\kappa}_{\cos,\cos}$ that generalizes as well as possible without caring about the approximation $\wt{\kappa}_{\cos,\cos} \approx \kappa$; this view is more faithful to recent research on the generalization capabilities of learning from RFF~\cite{zhang2018lowPrecisionRFF,yang2012nystrom,rudi2017generalization,gerace2020generalisation}. Our other RPF products used at the test ($\wt{\kappa}_{q,\cos}$ and $\wt{\kappa}_{q,q}$) are then to be understood as approximations to $\kappa^\sharp_{\rm L} = \wt{\kappa}_{\cos,\cos}$ rather than to the original $\kappa$, as explained in Cor.~\ref{cor:proximity-approx-kernels}, and as measured by the blue curve in Fig.~\ref{fig:tmpexp21}. 

\subsubsection*{Synthetic data}

Specifically, for both contexts, we generate a synthetic dataset of $10\,000$ samples in $\bb R^2$ by generating a mixture of $4$ Gaussians for each of the $N = 5$ different classes, separated into $n = 8000$ training and $2000$ testing samples (see Fig.~\ref{fig:exp3}, left). Regarding the true kernel, we set it to a Gaussian kernel $\kappa(\Vec{x},\Vec{x}') = \sik(\bs x - \bs x') = \exp(- \frac{\|\Vec{x}-\Vec{x}'\|_2^2}{2 \sigma^2})$ with bandwidth $\sigma = 2$, and fixed the regularization $R$ either to $5.0$ (mild regularization) or $0.25$ (strong regularization). For various feature space dimensions $m$, we generate the projections $\bs \Omega \sim \Lambda^m$ (with $\Lambda = \cl F \sik$, see Sec.~\ref{sec:background}) and dithering $\Vec{\xi} \sim \cl U^m([0,2\pi))$, and thus train (with Scikit-learn \cite{scikit-learn}) both an SVM classifier from the raw data with the Gaussian kernel, and another linear SVM from the associated cosine random Fourier features. We then evaluate these classifiers (in ``inference mode'') on the separate test set, using the different random features inner products and report the median accuracy (out of $25$ draws) as a function of $m$ in Fig.~\ref{fig:exp3}, right. In the four plots of this figure, we use a specific color coding of the RPF for both the incoming query vector $\bs x'$ and the learned SVM support vectors $\{\bs x_i\}_{i\in\cl S^*}$, as summarized in Table~\ref{tab:coding-col-fig-exp3} for convenience. Note that, according to this table, the green and the blue curves are associated with the scenarios represented in the Introduction in Fig.~\ref{fig:intro}a and Fig.~\ref{fig:intro}b, respectively, and the red curves relate to the (symmetric) quantized approach of~\cite{boufounos2015quantization}. 

\begin{table}[!t]
  \centering
  \caption{Scenarios and color coding for Fig.~\ref{fig:exp3}.}
  \scalebox{0.9}{\begin{tabular}{|c||c|c|c|}
    \hline
    Color&\multicolumn{2}{|c|}{RPF for}&Kernel at the test: $\kappa^\sharp_{\rm T}(\Vec{x}',\Vec{x}_i)$\\
    \cline{2-3}
    &Query vector $\bs x'$&Support vectors $\{\bs x_i\}_{i\in\cl S^*}$&\\
    \hline
    black&$\Vec{z}_{\cos}(\Vec{x}')$&$\Vec{z}_{\cos}(\Vec{x}_i)$&$\wt{\kappa}_{\cos,\cos}(\Vec{x}',\Vec{x}_i) = 2 \langle \Vec{z}_{\cos}(\Vec{x}'), \Vec{z}_{\cos}(\Vec{x}_i)  \rangle \approx \kappa(\Vec{x}',\Vec{x}_i)$\\
    green&$\Vec{z}_{q}(\Vec{x}')$&$\Vec{z}_{\cos}(\Vec{x}_i)$&$\wt{\kappa}_{q,\cos}(\Vec{x}',\Vec{x}_i) = \frac{\pi}{2}\,\langle \Vec{z}_{q}(\Vec{x}'), \Vec{z}_{\cos}(\Vec{x}_i)  \rangle \approx \kappa(\Vec{x}',\Vec{x}_i)$\\
    blue&$\Vec{z}_{\cos}(\Vec{x}')$&$\Vec{z}_{q}(\Vec{x}_i)$&$\wt{\kappa}_{\cos,q}(\Vec{x}',\Vec{x}_i) = \frac{\pi}{2}\,\langle \Vec{z}_{q}(\Vec{x}'), \Vec{z}_{\cos}(\Vec{x}_i)  \rangle \approx \kappa(\Vec{x}',\Vec{x}_i)$\\
    red&$\Vec{z}_{q}(\Vec{x}')$&$\Vec{z}_{q}(\Vec{x}_i)$&$\wt{\kappa}_{q,q}(\Vec{x}',\Vec{x}_i) = \langle \Vec{z}_{q}(\Vec{x}'), \Vec{z}_{\cos}(\Vec{x}_i)  \rangle$\\
    \hline
  \end{tabular}}
  \label{tab:coding-col-fig-exp3}
\end{table}

\begin{figure}[!t]
  \centering
  \includegraphics[width=0.98\linewidth]{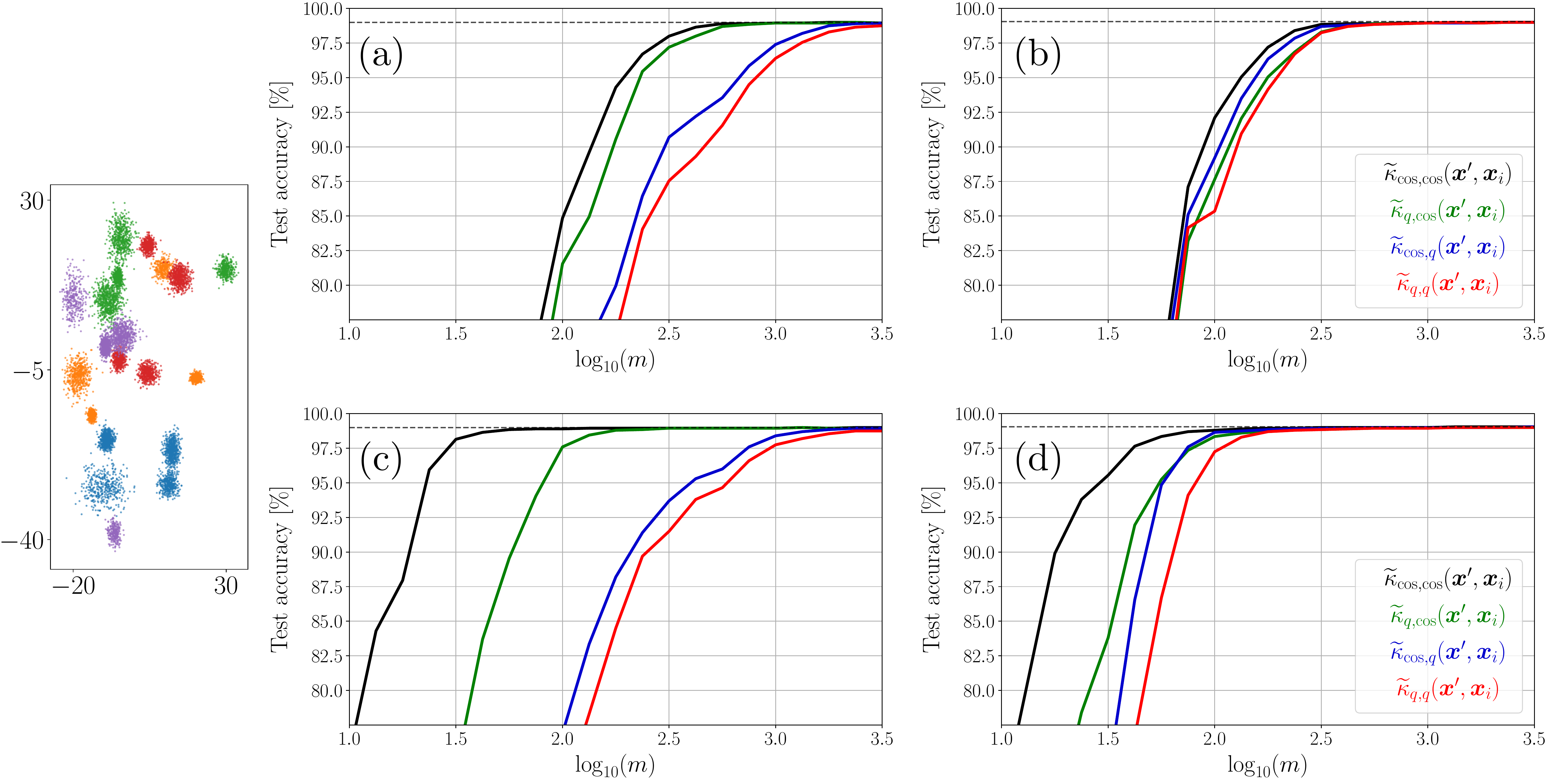}
  \caption{Left: Considered dataset (each arbitrary color corresponds to one of the 5 classes).
  Right: Test accuracy of an SVM classifier learned with a Gaussian kernel on the raw training set $\cl T$ (top (a-b)) or on their cosine RFF (bottom (c-d)), for regularization parameters $R = 5$ (weak regularization; left (a-c)) and $R = 0.25$ (strong regularization; right (b-d)), and evaluated with several RPF combinations as a function of their dimension $m$. Classification performance is measured on the test set according to the scenarios and curve colors described in Table~\ref{tab:coding-col-fig-exp3}.
  The curves are the median out of 25 independent draws of $\bs \Omega$ and $\Vec{\xi}$. The dashed line indicates the test accuracy of the exact SVM classifier (no random features are used).}
  \label{fig:exp3}
\end{figure}

When approximating the ``exact'' SVM classifier (where $\kappa^\sharp_{\rm L} = \kappa$; top row in Fig.~\ref{fig:exp3}), a substantial number of features is required to reach the same performances ($m \approx 300$ to obtain accuracy $>97.5\%$). As could be intuitively expected, the drop of accuracy is larger when more quantization of the features is being performed (the price to pay is particularly high when the support vectors are quantized, in blue and red). This difference is probably related to the fact that the SVM decision function is relatively sensitive to the position of the support vectors (in the feature space), because they directly lie on the decision boundary. Notice also the importance of regularization: although changing $R$ does not incur a noticeable change on the exact SVM accuracy (the horizontal dashed lines), it appears in Fig.~\ref{fig:exp3}b that stronger regularization improves the RPF-based classifiers. 

When the SVM is directly trained on the random Fourier features (\ie $\kappa^\sharp_{\rm L} = \tilde \kappa_{\cos,\cos}$; bottom row in Fig.~\ref{fig:exp3}) the accuracy of $\tilde \kappa_{\cos,\cos}(\Vec{x}',\Vec{x}_i)$ (black) is strongly boosted (reaching $97.5\%$ accuracy or higher with less than $m = 50$ features). The role of regularization is here exacerbated: reducing $R$ hurts the performances of the plain RFF classifier, but is necessary to maintain good accuracy with the semi-quantized schemes for a reasonable feature dimension $m$. It appears that proper regularization is needed to account for the quantization noise.

Finally, note that in all cases, the fully quantized scheme (in red) is not much worse than our semi-quantized solution with dataset quantization (the kernel mismatch shown Fig.~\ref{fig:exp1_Gaussian_qq} is apparently not too harmful in this case), but still suffers from strictly more classification errors. Therefore, it can be replaced by one of the asymmetric schemes for a negligible cost, it should be done. However, because we focus on approximating an imposed kernel without considerations for the underlying machine learning model, we did not compare to the fully-quantized case where the quantized features are already used during the training \cite{boufounos2015quantization}, \ie where the distorted kernel $\kappa_{q,q} = \bb E \wt{\kappa}_{q,q}$ is directly embraced to train the SVM rather than being used as an approximation for $\kappa$.\\
	
\subsubsection*{Real data: remote classification of hyperspectral pixels}

In a last experiment, to prove the concept of using semi-quantized features in a concrete and practical setting, we consider the problem of hyperspectral pixel classification, \ie determine the class of a spatial pixel given its electromagnetic spectral response across $d$ wavelengths. Kernel SVMs have been a quite popular solution to this challenge: our approach in particular is inspired by~\cite{gualtieri1999support,melgani2004classification} as well as~\cite{haridas2015gurls} for the use of RFF, but many more references can be found in the extensive review~\cite{mountrakis2011support}.

To have a concrete and quantifiable measure of the computational gains allowed by the quantization of RFF, we focus on the quantized query context (illustrated Fig.~\ref{fig:intro}a). More precisely, we consider the scenario of an aircraft (or satellite) equipped with a hyperspectral sensor that must send its readings $\Vec{x}'$ for remote classification of the pixels it observes. We assume this task is entrusted to a kernel SVM, involving a weighted sum of $\kappa^\sharp_{\rm T}(\Vec{x}',\Vec{x}_i)$ terms. Since the communication link between the satellite and the remote server is presumably costly, it is important that the number of bits used to encode this query, noted $b$, is as small as possible. 

We compare three strategies. First, the baseline strategy is to send the ``raw'' measurements $\Vec{x}' \in \bb R^d$, which requires $b = Bd$ bits, where $B$ designates the bit-depth of full-precision readings (we consider $B = 64$ bits in our experiments). In this strategy, the kernel in the learning and testing stages is the ``true'' Gaussian kernel, \ie $\kappa^\sharp_{\rm L} = \kappa^\sharp_{\rm T} = \kappa$. Second, the usual RFF strategy is to send the full-precision RFF $\Vec{z}_{\cos}(\Vec{x}') \in \bb R^m$, which requires $b = B m$ bits. Following the observations from the previous experiments, we both learn and test on the RFF kernel, \ie $\kappa^\sharp_{\rm L} = \kappa^\sharp_{\rm T} = \wt{\kappa}_{\cos,\cos}$. Finally, the quantized RFF query strategy is to send the quantized RFF $\Vec{z}_{q}(\Vec{x}') \in \{ - \frac{1}{\sqrt{m}}, + \frac{1}{\sqrt{m}} \}^m$, which takes up $b = m$ bits\footnote{Strictly speaking, we would actually need to transmit $\frac{1}{2}(\sqrt m \Vec{z}_{q}(\Vec{x}') +1) \in \{0,1\}^m$ and to remotely recover $\Vec{z}_{q}(\Vec{x}')$ from this binary stream, assuming $m$ is known to the receiver.}. Although the kernel at test time is now the hybrid product $\kappa^\sharp_{\rm T} = \wt{\kappa}_{q,\cos}$, we still learn using the usual RFF kernel $\kappa^\sharp_{\rm L} = \wt{\kappa}_{\cos,\cos}$. Note that to simplify the comparison, we thus consider only a naive encoding of the query, neglecting the use of \eg entropy coding strategies.

We use the standard Indian Pines dataset~\cite{IndianPinesDataset}, a hyperspectral volume which contains $10\,249$ labeled pixels\footnote{The size of the full volume is $145 \times 145$ which gives $21025$ pixels in total, but many of them are unlabeled, which we discard for this experiment.}, measured across $d=200$ wavelengths\footnote{The initial volume contains $220$ wavelengths, but following the workflow commonly adopted with this dataset, we removed the water absorption bands (\ie the spectral indices [104-108], [150-163], and 220 from the initial dataset).}, separated into $N = 16$ classes (see Fig.~\ref{fig:indian_pines}). We first separated $20\%$ of those pixels into testing set, which left $n = 8204$ pixels for training. In order to select the hyper-parameters (kernel bandwidth $\sigma$ and regularization strength $R$), following their sensitivity observed in the previous experiment, we performed a separate cross-validation (with 5 folds from the training set) for each of the three individual strategies.
We then evaluated the test set accuracy reached by each strategy, while letting $m$ vary for the two RFF-based strategies, and report the results Fig.~\ref{fig:exp4_subfig}. In this figure, the baseline strategy is represented by the red dot, and the usual (resp. quantized) RFF strategies are represented by the black (resp. green) solid curves, which are obtained by varying $m$.

\begin{figure}[!t]
	
	\centering
	\subfloat[Indian Pines dataset.]{
		\includegraphics[width=0.42\textwidth]{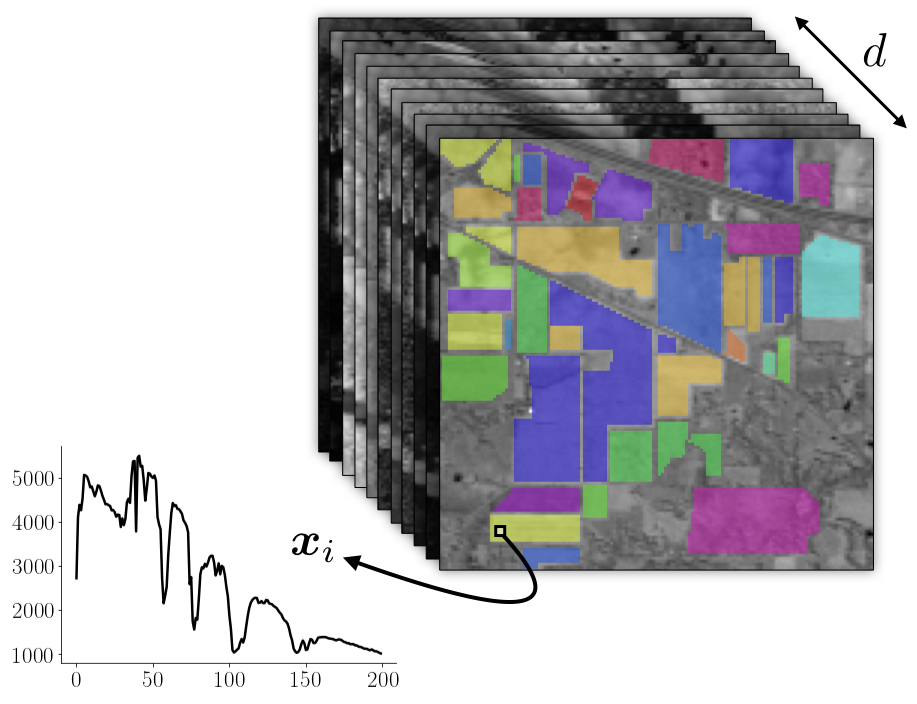}
		\label{fig:indian_pines}}
	\hspace{12px}
	\subfloat[Test accuracy \emph{w.r.t.} bit-rate $b$.]{	\raisebox{1mm}{\includegraphics[width=0.50\textwidth]{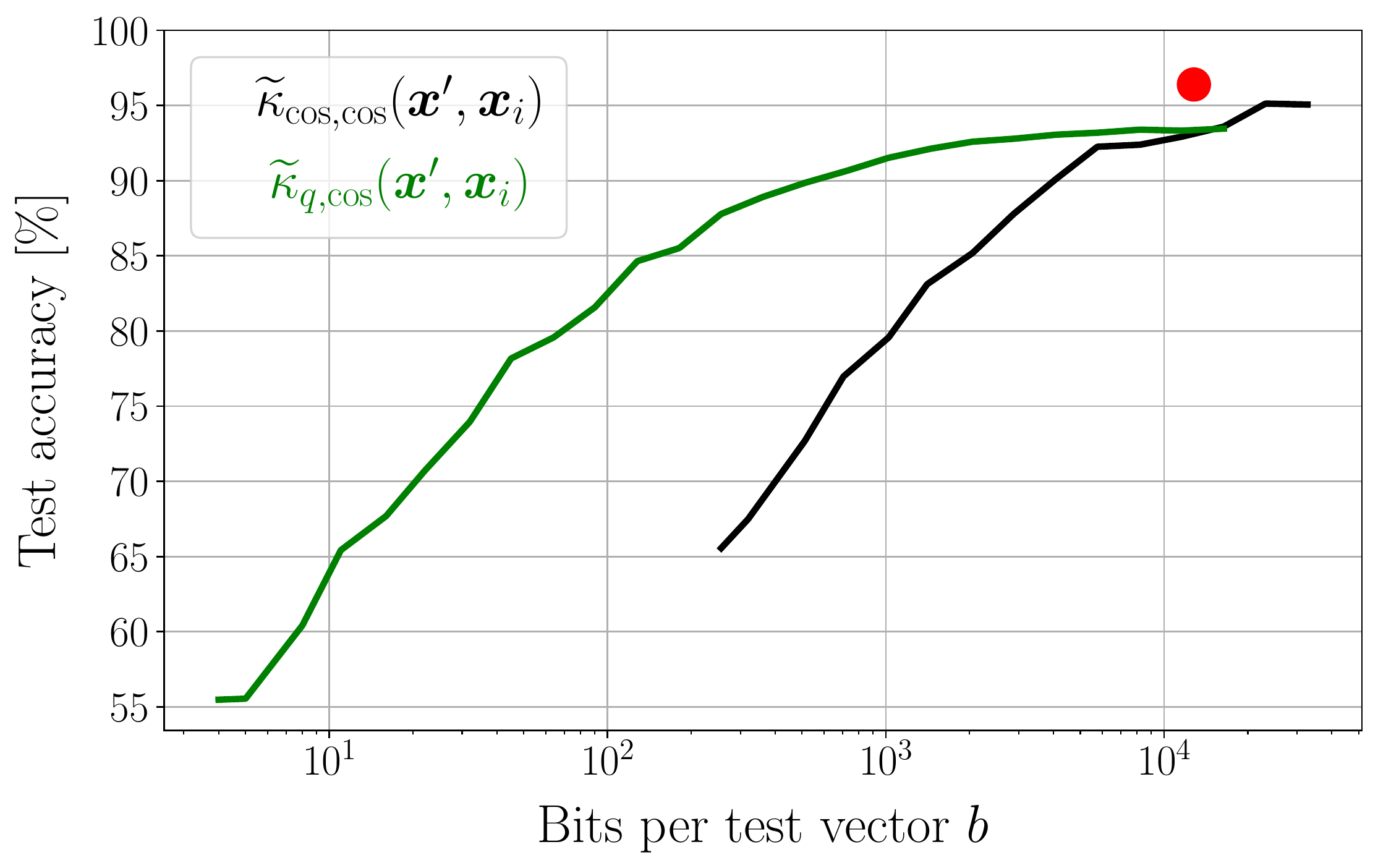}}
		\label{fig:exp4_subfig}}
	\caption{(a) Schematic representation of the Indian Pines hyperspectral volume, containing $10\,249$ labeled pixels $\Vec{x}_i \in \bb R^d$, which feature $d = 200$ wavelengths.
	(b) Test accuracy as a function of the number of bits $b$ to be transmitted to the remote server, for each of the three considered strategies: the baseline strategy (red dot) sends the raw test vector $\Vec{x}'$; the RFF strategy (solid black) sends the full-precision RFF $\Vec{z}_{\cos}(\Vec{x}')$ of varying size $m$; and the quantized RFF strategy (solid green) sends the one-bit equivalent $\Vec{z}_{q}(\Vec{x}')$.
	The curves are the median out of 30 independent draws of $\bs \Omega$ and $\Vec{\xi}$.}
	\label{fig:exp4}
\end{figure}

The baseline (red) achieves the best accuracy overall, but at the price of a quite substantial bandwidth usage. When using full-precision random features (black), the accuracy is only slightly reduced, but only if a relatively large number of random features $m$ is used, which does bring substantial bitrate reduction. Indeed, to reduce the bandwidth $b$ by, say, an order of magnitude, the full-precision RFF strategy must sacrifice more than $10\%$ accuracy, which is probably not acceptable in practice. On the other hand, with the quantized query RFF strategy, we are able to achieve this same bitrate reduction by an order of magnitude at the cost of only about $4\%$, which sounds more reasonable. Overall, (keeping in mind that more involved compression methods could be applied to transmit the raw measurement $\bs x'$ in our scenario above, and hence still apply the first classification strategy after decompression) the quantized RFF strategy performs better whenever the bitrate $b$ is significantly smaller than the baseline bitrate, hence showing the potential of the approach.

\section{Conclusion}
\label{sec:conclusion}

We introduced the framework of asymmetric random periodic features, where random projections are passed through two different periodic maps, and whose inner products are used to approximate a kernel. We provided an expression of this kernel, with a uniform error bound holding on infinite compact signal sets, provided the periodic maps satisfy a property we called the \emph{mean smoothness}. The mean smoothness holds for some discontinuous maps such as the one-bit universal quantization (a square wave). As a first (theoretical) application of those developments, we generalized the local geometry-preserving embeddings from~\cite{boufounos2017representation}, and corrected an error in their main result in the process. For a second more practical application, we studied (theoretically and empirically) semi-quantized kernel approximations, and showed how the impact of quantization can be controlled.

As highlighted by our last experiments, these theoretical guarantees do not necessarily ensure an accurate control over the generalization performances in a machine learning context. Indeed, it seems crucial to incorporate the random periodic features directly into the training stage, and to anticipate that some features might be quantized later (for example, when picking the regularization strength). We leave for future work the question of efficiently incorporating the asymmetric random periodic features strategy during the training stage. 

In the two applicative scenarii we proposed in the Introduction (Fig.~\ref{fig:intro}), the ultimate objective is to improve the bitrate-kernel accuracy trade-off. Allowing a finer (but still coarse) quantization of the RFF (\ie coding each entry on $b > 1$ bits instead of $1$ as we proposed, for example with the $b$-bit universal quantization~\cite{boufounos2017representation}) is a promising idea to reach a better trade-off. Since the mean smoothness of these finer quantization functions can also be verified, our results would carry over easily to this multi-bit quantization of the RFF, but a detailed analysis of this approach is also left for future work. More generally, the RPF with other periodic maps might also be worth studying. 
Finally, let us mention that our results could be used to obtain formal guarantees in the context of compressive learning \cite{gribonval2017compressiveStatisticalLearning,keriven2016compressive}, where the dataset sketch (which pools the RPF of each dataset sample; see Sec.~\ref{sec:related-work}) and the algorithm estimating the dataset distribution (such as CLOMPR\cite{keriven2016compressive}) relies on asymmetric periodic functions, \eg the universal quantizer and the cosine function \cite{schellekens2018quantized}.  

\bibliographystyle{unsrt}


\end{document}